\documentclass[mnsc,nonblindrev]{informs3}

\OneAndAHalfSpacedXI 

\usepackage{algorithm}
\usepackage{algcompatible}
\usepackage{enumerate}

\newenvironment{procedure}[1][htb]{%
    \floatname{algorithm}{Procedure}
  \begin{algorithm}[#1]%
  }{\end{algorithm}}

\usepackage[utf8]{inputenc}
\newtheorem{defn}{Definition}
\newtheorem{thm}{Theorem}
\newtheorem{lem}{Lemma}
\newtheorem{prop}{Proposition}
\newtheorem{cor}{Corollary}

\newtheorem{assu}{Assumption}

\usepackage[normalem]{ulem}

\newtheorem{argument}{Argument}

\newcommand{\E}{{\rm E}}

\usepackage{natbib}
 \bibpunct[, ]{(}{)}{,}{a}{}{,}%
 \def\bibfont{\small}%

 \usepackage{multirow}
 \usepackage{xcolor}
 \newcommand{\red}[1]{\textcolor{black}{#1}}
  \newcommand{\magenta}[1]{\textcolor{black}{#1}}
  \newcommand{\orange}[1]{\red{#1}}

\usepackage[bookmarks=false]{hyperref}
\hypersetup{
   colorlinks=true,
   linkcolor=[RGB]{17, 14, 178},
   citecolor=[RGB]{17, 14, 178},
   filecolor=magenta,      
   urlcolor=[RGB]{17, 14, 178},
	pdfpagemode=FullScreen,
   }

\TheoremsNumberedThrough     
\ECRepeatTheorems

\EquationsNumberedBySection 

\MANUSCRIPTNO{MS-SMS-23-00694} 

\begin{document}



\RUNTITLE{Sample Optimality of Greedy Procedures}

\TITLE{The (Surprising) Sample Optimality of Greedy Procedures for Large-Scale  Ranking  and  Selection}

\ARTICLEAUTHORS{%
\AUTHOR{Zaile Li}
\AFF{School of Management, Fudan University, Shanghai, China, \\\EMAIL{zaileli21@m.fudan.edu.cn}}
\AUTHOR{Weiwei Fan}
\AFF{Advanced Institute of Business and School of Economics and Management, Tongji University, Shanghai, China, \EMAIL{wfan@tongji.edu.cn}} 
\AUTHOR{L. Jeff Hong}
\AFF{School of Management and School of Data Science, Fudan University, Shanghai, China, \\\EMAIL{hong\_liu@fudan.edu.cn}}
} 

\ABSTRACT{%
Ranking and selection (R\&S) aims to select the best alternative with the largest mean performance from a finite set of alternatives. Recently, considerable attention has turned towards the large-scale R\&S problem which involves a large number of alternatives. Ideal large-scale R\&S procedures should be sample optimal, i.e., the total sample size required to deliver an asymptotically non-zero probability of correct selection (PCS) grows at the minimal order (linear order) in the number of alternatives, $k$. Surprisingly, we discover that the na\"{i}ve greedy procedure, which keeps sampling the alternative with the largest running average, performs strikingly well and appears sample optimal. To understand this discovery, we develop a new boundary-crossing perspective and prove that the greedy procedure is sample optimal 
 \orange{for the scenarios where the best mean maintains at least a positive constant away from all other means as $k$ increases.}
 We further show that the derived PCS lower bound is asymptotically tight for the slippage configuration of means with a common variance. \orange{For other scenarios, we consider the probability of good selection and find that the result depends on the growth behavior of the number of good alternatives: if it remains bounded as $k$ increases, the sample optimality still holds; otherwise, the result may change. }
 Moreover, we propose the explore-first greedy procedures by adding an exploration phase to the greedy procedure. The procedures are proven to be sample optimal and consistent under the same assumptions. Last, we numerically investigate the performance of our greedy procedures in solving large-scale R\&S problems.
}


\maketitle

%


\section{Introduction}
\label{sec: intro}
Selecting the alternative with the largest mean performance from a finite set of alternatives is an important class of ranking-and-selection (R\&S) problems and has emerged as a fundamental research topic in simulation optimization. Since the pioneering work of \cite{bechhofer1954single}, two different types of formulations and corresponding sample-allocation algorithms (known as procedures) have dominated the literature. Fixed-precision formulation typically aims to guarantee a target level of the probability of correct selection (PCS) using as small a sampling budget as possible and the procedures include the ones of \cite{bechhofer1954single}, \cite{paulson1964sequential}, \cite{rinott1978two}, \cite{kim2001fully} and \cite{jeff2006fully} among others. Fixed-budget formulation intends to allocate a predetermined total sampling budget to all alternatives, statically or sequentially, to optimize a certain objective, e.g., the PCS. The procedures include the ones of \cite{chen2000simulation}, \cite{chick2001new} and \cite{frazier2008knowledge} among others. For readers who are interested in the general development of R\&S, we refer to the comprehensive reviews of \cite{kim2006selecting}, \cite{chen2015ranking} and \cite{hong2021review}.

In recent years, due to the fast increase of computing power, large-scale R\&S \magenta{that involves a considerably large number of alternatives} has received a significant amount of research attention. Most of the works in the literature try to adapt classical procedures to parallel computing environments, and they include the APS procedure \citep{luo2015fully}, the GSP procedure  \citep{ni2017efficient}, the AOCBA and AKG procedures \citep{kaminski2018parallel} and the PPP procedure \citep{zhong2022speeding} among others. These works have made substantial progress in enlarging the solvable problem size, from thousands to now millions of alternatives. Readers are referred to \cite{hunter2017parallel} for an overview of the approach. Lately, the research focus has shifted to designing procedures that are inherently large-scale and parallel. These procedures are very different from classical ones, specifically, they tend to work well for large-scale problems but not necessarily for small-scale ones. They include the knockout-tournament (KT) procedures of \cite{zhong2022knockout}, the fixed-budget KT procedures of \cite{hong2022solving} and the parallel adaptive survivor selection (PASS) procedures of \cite{pei2022parallel}.

An important lesson that we learned from these latest developments is that large-scale R\&S problems are fundamentally different from small-scale problems. When the number of alternatives $k$ is large, the most important efficiency measure of R\&S procedures is the growth order of the total sample size with respect to $k$ in order to deliver a meaningful (i.e., non-zero) PCS asymptotically. This growth order dominates the total sample size for large-scale R\&S problems. \cite{zhong2022knockout} and \cite{hong2022solving} demonstrate that most well-known R\&S procedures, including the most famous fixed-precision procedures like the KN procedure of \cite{kim2001fully} and the most famous fixed-budget procedures like the OCBA procedure of \cite{chen2000simulation}, have a growth order of $O(k\log k)$, while the optimal (i.e., known attainable lower) bound is $O(k)$. Therefore, these procedures are not sample optimal, and numerical experiments show that they may have poor performance when $k$ is only moderately large, e.g., a few thousands. Therefore, in our point of view, it is vital to ensure the sample optimality when designing large-scale R\&S procedures.


To the best of our knowledge, there are only two types of procedures that are proved to be sample optimal for large-scale R\&S problems in the literature.  One is the median elimination (ME) procedure of \cite{even2006action}, which is proposed to solve the best-arm identification problem that is closely related to R\&S problems \citep{audibert2010best}. The procedure is a fixed-precision procedure, and it proceeds round by round. In each round, all surviving alternatives receive the same number of observations, their sample means are sorted, and the lower half are eliminated.  The ME procedure has many variants, e.g., \cite{jamieson2013finding} and \cite{hassidim2020optimal}, and it has also been extended to solve fixed-budget problems by 
\cite{karnin2013almost} and  \cite{zhao2022revisiting}. The other sample-optimal procedure is the (fixed-precision) KT procedure \citep{zhong2022knockout} and its fixed-budget version \citep{hong2022solving}. The procedures proceed round by round like a knockout tournament. In each round, the surviving alternatives are grouped in pairs and an alternative is eliminated from each pair. The last surviving alternative is declared as the best. The procedures allow the use of common random numbers between each pair of alternatives, and they are very suitable for parallel computing environments. 



Notice that both the ME and KT procedures use the halving structure, i.e., eliminating half of the alternatives in each round, with carefully planned sample-allocation schemes in each round to balance exploration and exploitation and to ensure the sample optimality. Despite their sample optimality, the halving structures of these procedures are rigid. Therefore, a natural question to ask is whether there exist other types of procedures that are less rigid but also sample optimal. 

Through a preliminary numerical study, we discover that the simple greedy procedure, which always allocates the next observation to the alternative with the largest running average after first allocating one observation to each of the $k$ alternatives, appears to be sample optimal. 
For illustration, we consider a simple R\&S problem \red{maintaining a fixed mean difference between the best alternative and all others as $k$ varies. In the problem, the performance of alternative $i=1, \ldots, k$ follows a normal distribution with mean $\mu_i$ ($\mu_1=0.1$ and $\mu_i=0$ for $i>1$) and variance $\sigma^2=0.25$.}
We then plot the estimated PCS for both the OCBA procedure and the greedy procedure with a total budget of $B=100k$ while varying $k$ from $2^2$ to $2^{12}$ in Figure \ref{fig: simple_example}. From the plot, we can see that the PCS of the OCBA procedure plunges to zero as $k$ increases, while the PCS of the greedy procedure appears to stabilize around 25\% as $k$ increases. 
\red{It is important to note that, the setting of fixed mean difference considered above may not be sufficient for large-scale R\&S problems. When $k$ increases, the means of some alternatives may become very close to the mean of the best, leading to the low-confidence scenario discussed by \cite{peng2018gradient}. For this scenario, numerical results show that the PCS could decrease to zero, encouraging us to consider the probability of selecting a good alternative (PGS). As we will elaborate later, the greedy procedure can also achieve the sample optimality regarding PGS, given that the number of good alternatives is bounded by a constant.}

\begin{figure}[htbp]
      \centering
      \includegraphics[width=0.5 \textwidth]{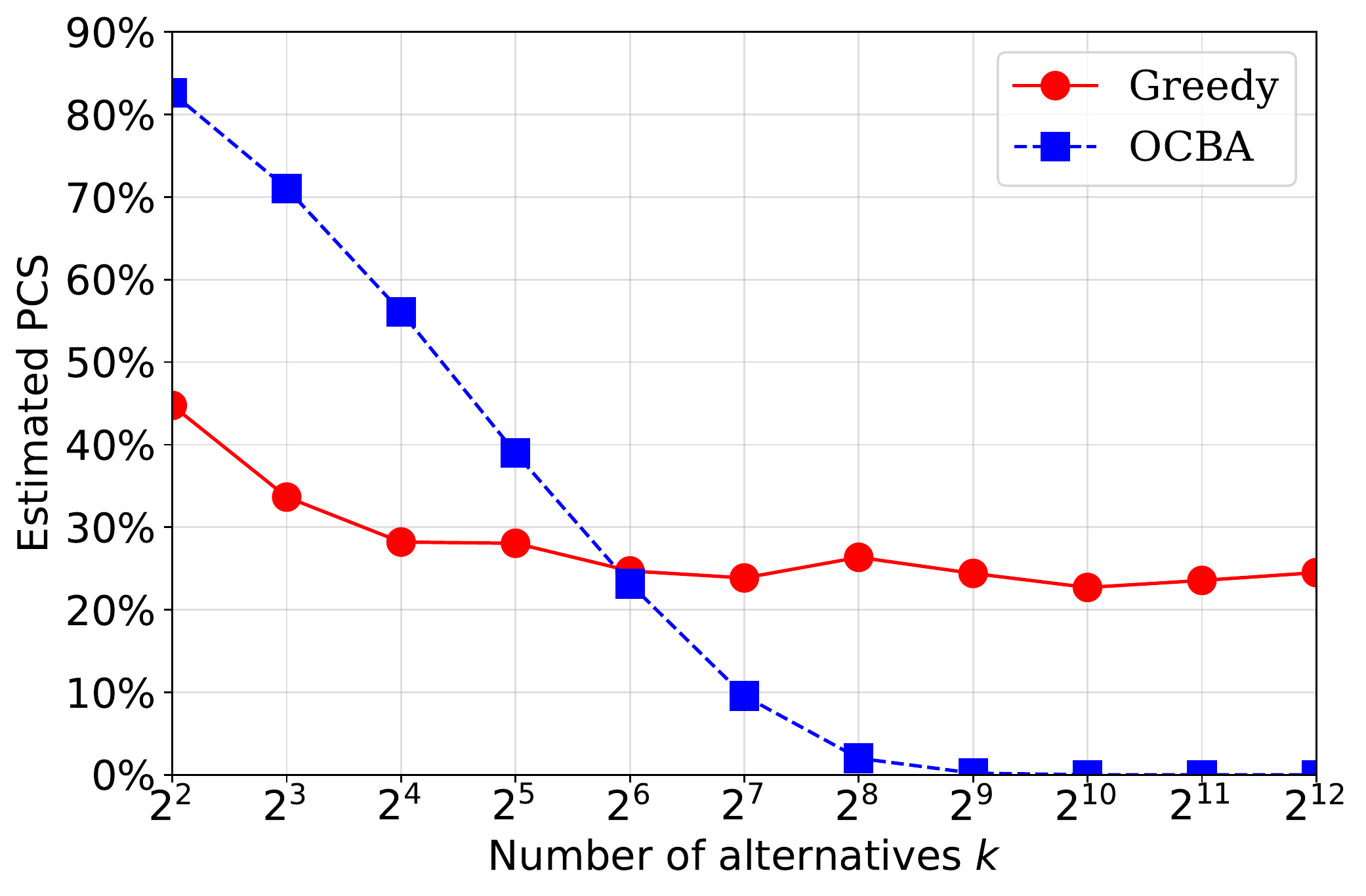}
      \caption{A comparison between the greedy procedure and the OCBA procedure.}
      \label{fig: simple_example}
\end{figure}

Our discovery is surprising because the greedy procedure is often perceived as na\"{i}ve and poor, and it is difficult to imagine that it is a candidate for the sample optimality. Furthermore, this discovery is even more counter-intuitive, because we often think a good R\&S procedure has to balance the exploration-exploitation tradeoff, but the greedy procedure is purely exploitative. 
In this paper, our goal is to prove that the greedy procedure is indeed sample optimal under a fixed-budget formulation, to understand why it is sample optimal, and to use the knowledge to develop procedures that have competitive performance for large-scale R\&S problems.

Although the greedy procedure is simple and easy to understand, analyzing its PCS is challenging. This is because the sample best process of the greedy procedure is a complicated stochastic process that depends on all observations of all alternatives and the sample sizes of all alternatives are intertwined with all sample means. In this paper, we develop a boundary-crossing perspective in Section \ref{sec:bc} that treats the (unknown) best as a boundary. From this new perspective, we observe two interesting insights about the greedy procedure. Firstly, the sampling process of each inferior alternative can be captured by its corresponding boundary-crossing process. \magenta{More specifically, the sample size allocated to each inferior alternative is essentially determined by the first time that its running sample-mean process crosses the aforementioned boundary}. Secondly, the PCS of the greedy procedure can be represented by the probability that the sum of the $k-1$ random boundary-crossing times of all $k-1$ inferior alternatives is smaller than the total sampling budget $B$. Then, by the strong law of large numbers, the total sampling budget required to ensure a non-zero PCS can be in the order of $k$. In other words, we can show that the greedy procedure is sample optimal. Besides, as we explain later, this new boundary-crossing perspective also provides an important insight into why the greedy procedure can be sample optimal even without the halving structure as the ME and KT procedures.

More precisely, based on the boundary-crossing perspective, we derive an analytical non-zero lower bound of the PCS for the greedy procedure as $k$ goes to infinity, \magenta{when the mean difference between the best alternative and the other alternatives remains at least a positive constant irrespective of $k$ \red{(i.e., Assumption \ref{assu:asyreg} holds)}.} We show that the total sample size to ensure this asymptotic lower bound grows linearly in $k$, proving the sample optimality of the procedure.
Furthermore, we prove that the lower bound of the PCS for the greedy procedure is tight under the slippage configuration of means with a common variance. This is also an interesting result. Aside from the work of \cite{bechhofer1954single} and \cite{frazier2014fully} under the exact same configuration, we are not aware of other R\&S procedures, especially sequential R\&S procedures, that have a tight PCS lower bound when $k \ge 3$. Most procedures cannot avoid the use of Bonferroni-type inequalities (e.g., \citealt{kim2001fully} and \citealt{chen2000simulation}), so they only obtain PCS lower bounds that are not tight. 

The greedy procedure has a clear drawback. It is inconsistent, i.e., its worst-case PCS has an asymptotic upper bound that is strictly less than 1 no matter how large the sampling budget is. This is easy to understand because the best alternative may never be revisited if its first observation is very poor. To solve this problem, we develop the explore-first greedy (EFG) procedure, which allocates $n_0\ge 1$ observations to each alternative at the beginning. Notice that the EFG procedure becomes the greedy procedure if $n_0=1$. We show that this equally simple procedure is also sample optimal with a tight PCS lower bound. Furthermore, we show that there are two interesting results as $n_0$ increases to infinity. First, the PCS lower bound may go to 1. Therefore, the EFG procedure is consistent. Second, the proportion of the total sample size needed for the greedy phase decreases to zero. This is a surprising result, because the equal allocation procedure (used in the initial phase) is not sample optimal. However, by adding only a tiny effort of the greedy phase, the equal allocation procedure can be turned into sample optimal. The numerical study also confirms this interesting finding, showing the potential for the greedy procedure to be a remedy to other non-sample-optimal procedures. Moreover, it is worth noticing that we are not restricted to use the equal allocation in the exploration phase. By allocating more exploration budget to competitive alternatives, we propose an enhanced version of the EFG procedure which is also consistent, and name it the $\mbox{EFG}^+$ procedure. As we will see in Section \ref{sec: numerical}, the $\mbox{EFG}^+$ procedure significantly outperforms the EFG procedure. \magenta{In addition, we propose a parallel version of the $\mbox{EFG}^+$ procedure to enhance its practical effectiveness in solving large-scale R\&S problems. In contrast to the standard $\mbox{EFG}^+$ procedure, the parallel procedure performs batched simulations at each stage of the sequential sampling process by using parallel computing resources.}


\red{The PCS results achieved above rely critically on Assumption \ref{assu:asyreg} that the mean difference between the best alternative and all other alternatives remains above a positive constant as $k$ increases. In case the assumption is not met, we also analyze the PGS of the greedy procedures. We rigorously prove that if the total number of good alternatives $g(k)$ remains bounded as $k$ increases, the greedy procedures can achieve the sample optimality in terms of the PGS. In other words, they can maintain a non-zero PGS in the limit $k\to\infty$ given that the total budget $B$ grows linearly in $k$ (i.e., $B=O(k)$). This scenario corresponds to problems where the emergence of new good alternatives is negligible as $k$ increases, e.g., the large-scale throughput maximization problem considered by \cite{luo2015fully} and \cite{ni2017efficient}. }

\red{In practice, there are also scenarios where $g(k) \rightarrow \infty$ as $k \rightarrow \infty$.  For instance, when the alternatives are generated by discretizing a bounded region, considered by \cite{yakowitz2000global} and also in \cite{chia2013limit}, the number of good alternatives that fall in a small fixed neighborhood of the optimal point may scale linearly in the total number of alternatives $k$ as finer discretization is employed; furthermore, the range of the region may expand during discretization, leading to a growth in a sub-linear order (e.g., $O(\sqrt{k})$) of the number good alternatives. In our numerical experiments, we also examine the EFG procedure's PGS for $g(k) = O(\sqrt{k})$ and $g(k) = O(k)$. Intuitively, having more good alternatives makes good selection easier and thus letting $g(k)\rightarrow \infty$ given $B=O(k)$ could lead to a PGS approaching 1 as $k$ increases. Surprisingly, we find that the expected behavior holds for $g(k) = O(\sqrt{k})$ but disappears when $g(k) = O(k)$. This indicates that the performance of the greedy procedures may depend critically on the growth order of the number of good alternatives.}

We summarize the contributions of this paper as follows. 
\begin{itemize}
    \item We discover that the greedy procedure is sample optimal for large-scale R\&S problems, adding a very simple but unexpected procedure to the very short list of sample-optimal R\&S procedures.
    
    \item We propose a boundary-crossing perspective to prove the sample optimality of the greedy procedure \orange{under the indifference-zone assumption (i.e., Assumption \ref{assu:asyreg})}, and show that the resulting PCS lower bound is asymptotically tight under the slippage configuration with a common variance.
    
    \item We develop the EFG and $\mbox{EFG}^+$ procedures that are consistent and sample optimal, and extend them to parallel computing environments. They work well for large-scale R\&S problems compared to other sample-optimal procedures in the literature. 
     \item \orange{We prove the sample optimality regarding PGS when the number of good alternatives is bounded as $k$ increases. We also identify and clarify other scenarios for which the procedures might work well or might not work well.}
    \item Our research further corroborates that a new mindset may be necessary to move from small-scale to large-scale R\&S problems.
\end{itemize}

\magenta{While this paper has yielded interesting insights and contributions, it is important to acknowledge certain limitations that could guide further refinements and explorations. First, we assume that the mean performances of alternatives are unstructured, which restricts our attention to procedures that do not utilize the structure information of ``nearby" alternatives. However, it is worth noting that leveraging this information may be useful in improving the efficiency of the procedures \citep{semelhago2019gaussian, semelhago2021rapid}.
Second, we establish the sample optimality under the assumption that the best mean maintains at least a positive constant away from all others as $k$ increases, and we are only interested in selecting the best. In practice, this is a strong assumption. We relax this assumption by considering the probability of good selection. We establish the sample optimality in terms of the probability of good selection, under the additional assumption that the number of good alternatives is bounded as $k$ increases. Last, we assume the observations collected to be independent across alternatives, which gives up the use of common random numbers that are often used in R\&S procedures to improve the efficiency \citep{nelson1995using, chick2001newCRN, zhong2022knockout}.}
%

We end this section with three additional remarks. First, there is an interesting similarity between the greedy procedure and the PASS procedures of \cite{pei2022parallel}, as they both embed a structure of comparisons with an adaptive standard (or boundary). The PASS procedures have an explicit standard learned gradually through the aggregated sample mean of the surviving alternatives, while the greedy procedure, implicitly uses the sample mean of the (unknown) best as the standard.
Second, alongside our work, the surprisingly good performance of the greedy procedure has also been noticed in the field of multi-armed bandit by \cite{kannan2018smoothed} and \cite{bayati2020unreasonable} among others. However, their focus is on the growth order of the cumulative regret, and the greedy procedure is only shown to be sub-optimal \citep{jedor2021greedy} unless a diversity condition on the alternatives is satisfied \citep{bastani2021mostly}, making them significantly different from our work. \magenta{Last, boundary-crossing mechanisms are common in sequential R\&S procedures. For instance, frequentist procedures, such as the KN procedure of \cite{kim2001fully} and the IZ-free procedure of \cite{fan2016indifference}, use two-sided boundaries to define continuation regions that guide elimination decisions, and Bayesian procedures, such as the ESP procedures of \cite{chick2012sequential}, stop the sampling of an alternative when its posterior mean exceeds a boundary.} \orange{In these procedures, the boundaries are pre-determined; they are either explicitly specified \citep{kim2001fully} or implicitly defined by a free-boundary problem for the optimal stopping of a diffusion process \citep{chick2012sequential}. However, the boundary in our procedures is both implicit and adaptive, thus eliminating the need for choosing the boundary in advance.}

The rest of this paper is organized as follows. In Section \ref{sec: problem}, we provide the problem formulation and introduce the greedy procedure. In Section \ref{sec:bc}, we review the greedy procedure from a boundary-crossing perspective, based on which we prove its sample optimality in Section \ref{sec: GreedyOpt}. In Section \ref{sec: ExploreFirst}, we introduce the EFG procedure to resolve the inconsistency issue, \magenta{extend it to the EFG$^+$ procedure, and discuss its variants in parallel computing environments}. In Section \ref{sec: numerical}, we conduct a comprehensive numerical study to verify the theoretical results and to understand the performance of our procedures for large-scale R\&S problems. We conclude in Section \ref{sec: conclusion} and include the proofs in the e-companion.

\section{Problem Statement}
\label{sec: problem}

Consider a fixed-budget formulation of an R\&S problem. The problem has $k$ alternatives, and each alternative $i$ may be represented by a random variable $X_i$, $i=1,\ldots,k$. We are given a total budget of $B$ observations. Our goal is to sequentially allocate the observations to the alternatives to select the alternative with the largest mean performance, i.e., $\argmax_{i=1,\ldots,k} \E[X_i]$. Following the convention of the simulation literature (e.g., \citealt{bechhofer1954single}, \citealt{kim2001fully} and \citealt{zhong2022knockout}), we assume that $X_i$ follows a normal distribution with mean $\mu_i$ and variance $\sigma_i^2$ for all $i=1,\ldots,k$, and the observations collected from each alternative are independent. Furthermore, we assume that the best alternative is unique. Then, without loss of generality, we may assume that $\mu_1 > \mu_i$ for all $i=2,\ldots, k$, and our goal is to identify which alternative is alternative $1$. We further define the probability of correct selection (PCS) as the probability that the alternative $1$ is selected at the end of the procedure, and we use the PCS as a criterion to evaluate the effectiveness of R\&S procedures.

\subsection{Sample Optimality}

In this paper, we are interested in solving large-scale R\&S problems where the number of alternatives $k$ is large. In particular, we are interested in understanding how the PCS and the necessary budget $B$ are affected by $k$. Therefore, we consider an asymptotic regime where $k\to\infty$. There are two remarks that we want to make on this asymptotic regime. First, although practical R\&S problems typically involve only a finite number of alternatives, different procedures appear to have very different behaviors as $k$ increases, as illustrated by Figure \ref{fig: simple_example} in Section \ref{sec: intro}. Therefore, we believe that the asymptotic regime of $k\to\infty$ provides a good framework under which these differences may be understood. Second, while letting $k\to\infty$, we assume that the difference between the best and the second best remains \orange{above a positive constant $\gamma$}, i.e., $\mu_1-\max_{i=2, \ldots, k}\mu_i\geq \gamma$.  Notice that this treatment is similar to the indifference-zone assumption widely used in R\&S procedures (e.g., \citealt{bechhofer1954single}, \citealt{kim2001fully} and \citealt{zhong2022knockout}). However, we want to emphasize that there is a difference. In our setting, we only need the existence of such $\gamma>0$ while the indifference-zone parameter is assumed known in advance. We will consider the situation where this assumption is violated in Section \ref{subsec: GreedyOptPGS}.

Under the asymptotic regime of $k\to\infty$, following the definition of \cite{hong2022solving}, we present the following definition of the sample optimality. 

\begin{defn}
\label{def: rate_optimality}
A fixed-budget R\&S procedure is sample optimal if there exists a positive constant $c>0$ such that
\begin{equation}\label{eqn:ro}
\liminf_{k\to\infty}\, {\rm PCS}>0\ {\rm for}\ B=ck.
\end{equation}
\end{defn}


As pointed out by \cite{hong2022solving}, to ensure an asymptotically non-zero PCS, the total budget has to grow at least linearly in $k$. That's why we call procedures that satisfy Equation (\ref{eqn:ro}) sample optimal. Many fixed-budget R\&S procedures, including the OCBA procedures and the ones based on the large-deviation principles of \cite{glynn2004large}, allocate the budget $B$ by either solving or approximating the optimal solution of the following optimization problem
\begin{eqnarray*}
&{\rm maximize} & \quad \Pr\left\{\bar X_1(n_1) > \bar X_i(n_i),\, i=2,\ldots,k\right\} \\ &{\rm subject\ to} & \quad n_1+n_2+\ldots+n_k=B, \\
&& \quad n_i>0,\ i=1,2,\ldots,k,
\end{eqnarray*}
where $\bar X_i(n_i)$ is the same mean of $n_i$ observations of alternative $i$ for $i=1,\ldots,k$. \cite{hong2022solving} show that, to ensure an asymptotically non-zero PCS, the total budget of these procedures has to grow at least in the order of $k\log k$. Therefore, these procedures are typically not sample optimal, and they are inefficient for large-scale R\&S procedures even though their performances for small-scale problems may be superb. Therefore, for large-scale R\&S problems, we believe that the sample optimality is a critical requirement. \cite{hong2022solving} show that the fixed-budget KT procedure is sample optimal, and \cite{zhao2022revisiting} prove that a specific fixed-budget version of the ME procedure is sample optimal. 
In Section \ref{subsec: GreedyOpt}, we present a simple greedy procedure that is also sample optimal.

\subsection{The Greedy Procedure}
\label{subsec: GreedyOpt}

Consider a greedy procedure that allocates the first $k$ observations to all $k$ alternatives, so that each alternative has an observation, and subsequently allocates the next observation to the alternative with the largest running average until the sampling budget $B$ is exhausted. See below for a detailed presentation of the procedure.

\begin{procedure}[hbt!]
\caption{\textbf{Greedy Procedure}}
\label{procedure: Greedy}
\begin{algorithmic}[1]
\vspace*{4pt}
\REQUIRE $k$ alternatives $X_1,\ldots,X_k$ and the total sampling budget $B$.
\STATE For all $i=1,\ldots,k$, take an observation $X_{i1}$ from alternative $i$, set $\bar{X}_i(1)={X}_{i1}$, and let $n_i=1$.
\WHILE{$\sum_{i=1}^k n_i < B$}
\STATE Let $s = \arg\max_{i\in\{1,\ldots,k\}}  \bar{X}_{i}\left(n_i\right)$ and take one observation $x_{s}$ from alternative $s$;
\STATE Update $\bar X_{s}(n_{s}+1) = \frac{1}{n_{s}+1}\left[n_{s}\bar X_{s}(n_s) + x_s\right]$ and let $n_s = n_s+1$;
\ENDWHILE
 \STATE Select $\arg\max_{i\in\{1,\ldots,k\}}  \bar{X}_{i}\left(n_i\right)$  as the best.
\end{algorithmic}
\end{procedure}

The greedy procedure is very simple to understand and simple to implement. It is probably one of the simplest sequential R\&S procedures that one can conceive. However, because it is purely exploitative in nature, it is difficult to imagine such a simple and na\"{i}ve procedure would have a competitive performance, let alone the sample optimality, when solving large-scale problems. That is what we intend to show in this paper.

Although the greedy procedure itself is simple and straightforward, rigorously characterizing its sample optimality is by no means trivial. Let $n_i(t)$ denote the sample size of alternative $i$ when a total of $t$ observations have been allocated with $k\le t \le B$. Then,
\begin{eqnarray*}
Y(t) &=& \max_{i\in\{1,\ldots,k\}}\, \bar X_i(n_i(t)), \quad t=k,k+1,\ldots,B,\\
s(t) &=& \argmax_{i\in\{1,\ldots,k\}}\, \bar X_i(n_i(t)),\quad t=k,k+1,\ldots,B,
\end{eqnarray*}
denote the processes of running sample maximum and its identity, respectively. Notice that $s(t)$ is the alternative that receives the $t+1$st observation. These processes fully characterize the dynamics of the greedy procedure. However, they are difficult to analyze, because they depend on all previous observations of all alternatives. Therefore, a new perspective is necessary when analyzing the sample optimality of the greedy procedure.

\section{The Boundary-Crossing Perspective}\label{sec:bc}

In this section, we propose to analyze the greedy procedure from a boundary-crossing perspective and illustrate how this perspective facilitates the analysis of the sample optimality. To streamline the presentation, we start by considering a simplified case of an R\&S problem and then extend the analysis to the general case.

\subsection{A Simplified Case}
\label{subsec: bc_simple}

Suppose that the best alternative (i.e., alternative $1$) can be observed without random noise, i.e., $\bar X_1(n)=\mu_1$ for any $n \geq 1$. Then, it is important to notice that at any $t=k,k+1,\ldots,B$ and for any $i=2,\ldots,k$, if $\bar X_i(n_i(t))<\mu_1$, alternative $i$ will never be sampled again. Therefore, we may view $\mu_1$ as a boundary. If all inferior alternatives have sample-mean values falling below $\mu_1$ after using $t$ observations of the budget with $t\le B$, i.e., $Y(t)=\mu_1$ and $s(t)=1$, then it guarantees a correct selection, i.e., $s(B)=1$.\footnote{Notice that we ignore the event $\max_{i=2,\ldots,k}\,\bar X_i(n_i(t))=\mu_1$, because it is a probability-zero event under the normality assumption.}

Let $N_i(\mu_1)=\inf\left\{n\geq 1: \bar X_i(n) < \mu_1\right\}$ denote the boundary-crossing time of alternative $i$, i.e., the minimal number of observations that alternative $i$ needs to have its sample mean falling below the boundary $\mu_1$, for any $i=2,\ldots,k$. Then, it is clear that $\left\{\sum_{i=2}^{k} N_i(\mu_1) + 1\le B\right\}$ denotes a correct-selection event where the ``+1'' is the one observation allocated to alternative 1.  Then,
\begin{equation}\label{eqn:pcs}
{\rm PCS}=\Pr\left\{\sum_{i=2}^{k} N_i(\mu_1) + 1\le B\right\}.
\end{equation}
Equation \eqref{eqn:pcs} provides a new perspective to look at the greedy procedure, and we call it the boundary-crossing perspective. First, it shows that, for the simplified case, the event of correct selection is determined only by the sum of the boundary-crossing times of all inferior alternatives. Second, it is critical to notice that $N_2(\mu_1),\ldots, N_{k}(\mu_1)$ are mutually independent and the order of observations allocated to inferior alternatives is irrelevant to the PCS. We illustrate the boundary-crossing perspective in Figure \ref{fig: simplified} using a simple example with 4 alternatives. In this example, $N_2(\mu_1)=5$, $N_3(\mu_1)=1$, $N_4(\mu_1)=13$, and the alternative 1 is selected if $B\ge 20$. Furthermore, the figure illustrates that how the three inferior alternatives cross the boundary $\mu_1$ is irrelevant to the event of correct selection.

\begin{figure}[thbp]
  \centering
  \includegraphics[width =0.75 \textwidth]{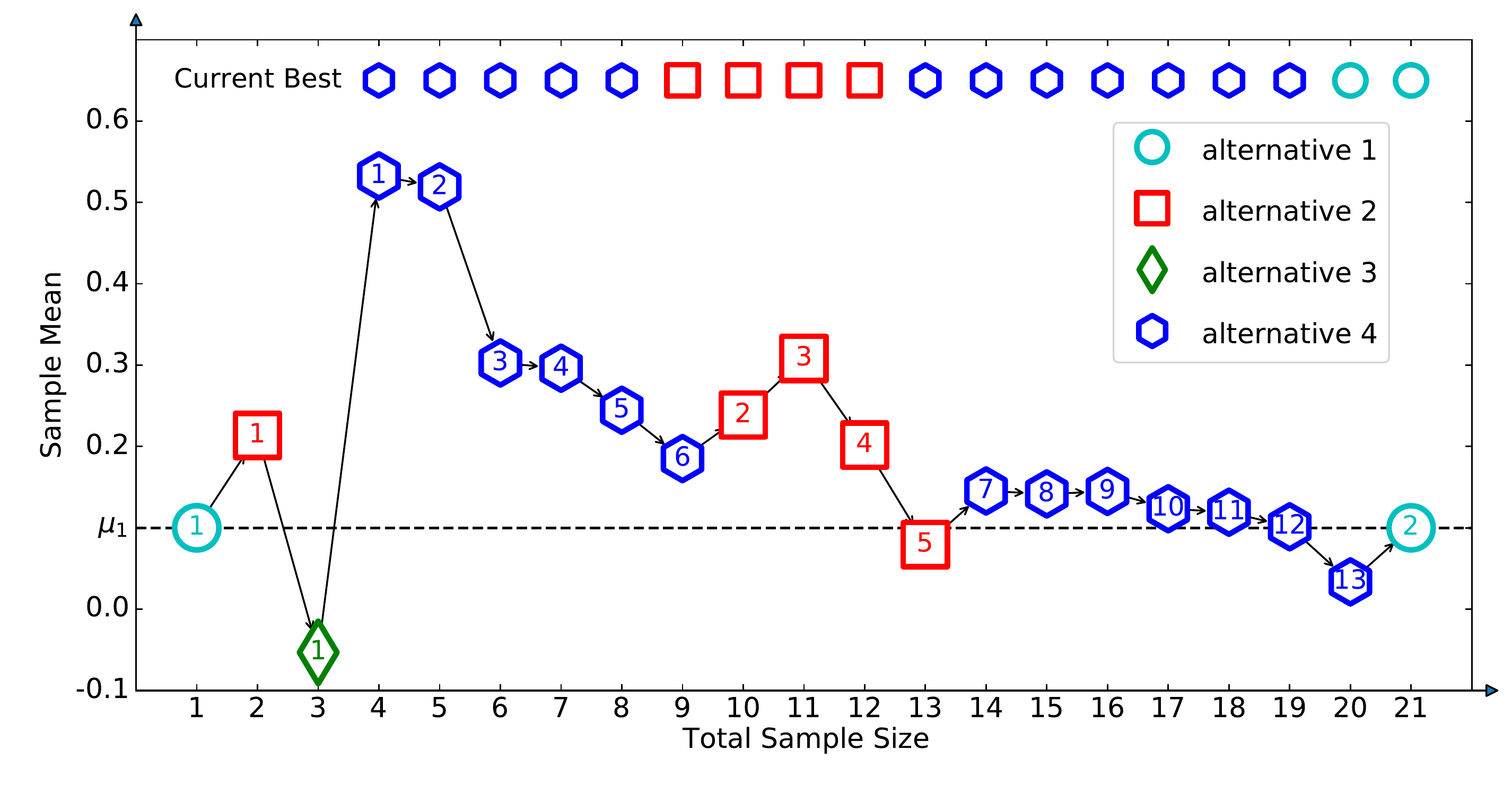}
  \caption{Sequential sampling process of the greedy procedure for a simplified R\&S problem with $4$ alternatives. 
  The numbers inside the markers represent sample sizes. 
  }
  \label{fig: simplified}
\end{figure}

\subsection{The General Case}
\label{subsec: bc_general}

Now we consider the general R\&S problem where all alternatives are observed with random noise. Compared to the simplified case presented in Section \ref{subsec: bc_simple}, there is no deterministic boundary $\mu_1$ that is readily available. To use the boundary-crossing perspective, we need to define a new boundary. Let $\bar X_1^* = \min_{n\ge 1} \bar X_1(n)$ denote the minimum of the running-average process of alternative $1$, \magenta{which is well-defined as it is finite 
 almost surely\footnote{\magenta{By the strong law of large numbers, the sequence $\{\bar X_1(n): n = 1,2,\ldots\}$ converges to $\mu_1$ almost surely. Consequently, $\bar{X}_1^*$ is bounded and finite with probability 1 due to the boundness of convergent sequences.}}}. Notice that $\bar X_1^*$ is a natural choice of the boundary, because inferior alternatives whose sample means are below $\bar X_1^*$ are dominated by alternative $1$ and will never be sampled again. 

Similar to the simplified case presented in Section \ref{subsec: bc_simple}, we let $N_i(\bar X_1^*)=\inf\left\{n\ge 1: \bar X_i(n)<\bar X_1^*\right\}$ denote alternative $i$'s boundary-crossing time for all $i=2,\ldots,k$. Unlike the simplified case, however, $\left\{\sum_{i=2}^{k} N_i(\bar X_1^*) +1 \le B\right\}$ does not necessarily imply a correct-selection event. This is because, in the simplified case, once alternative $1$ is the running sample best, it remains as the best until all budget is used; while in the general case, alternative $1$ has randomness, it may collect a large number of observations and suddenly appear inferior to other alternatives. Then, there may not be enough budget left to support all inferior alternatives to reach their boundary-crossing times. 

To avoid it from happening, we require the alternative $1$ to reach its minimum within the budget. After that, once it becomes the sample best, it remains the best until all budget is used, as in the simplified case. Let $N_1^*=\argmin_{n\ge 1} \bar X_1(n)$. Then, $\left\{\sum_{i=2}^{k} N_i(\bar X_1^*)+N_1^*\le B\right\}$ implies a correct-selection event and, therefore,
\begin{equation}\label{eqn:pcs2}
{\rm PCS}\ge \Pr\left\{\sum_{i=2}^{k} N_i(\bar X_1^*)+N_1^*\le B\right\}.
\end{equation}
Notice that Equation (\ref{eqn:pcs2}) is an inequality, which is different from Equation (\ref{eqn:pcs}) in the simplified case. This is because there may exist other scenarios of correct selection, for instance, all inferior alternatives have already crossed the boundary $\bar X_1^*$ before alternative $1$ reaches the minimum. However, surprisingly, in Section \ref{subsec: tightness}, we show that the PCS lower bound derived using Equation (\ref{eqn:pcs2}) is tight as $k\to\infty$ in some specific problem setting.

The boundary-crossing perspective of Equation (\ref{eqn:pcs2}) is a critical result of this paper. It greatly simplifies the analysis of the PCS for the greedy procedure. In the next section, based on this perspective, we are able to prove the sample optimality of the greedy procedure.

\section{Sample Optimality of The Greedy Procedure}
\label{sec: GreedyOpt}

In this section, we establish the sample optimality of the greedy procedure. First, we define the asymptotic regime more explicitly in the following assumption.
\begin{assu}\label{assu:asyreg}
    There exist constants $\gamma>0$ and $\bar\sigma^2>0$ such that $\mu_1-\max_{i=2,\ldots,k}\,\mu_i \ge\gamma$ and $\max_{i=1,\ldots,k}\,\sigma_i^2\le \bar\sigma^2$ regardless of how large $k$ is.
\end{assu}
Notice that in Assumption \ref{assu:asyreg}, the existence of $\gamma>0$ may be easily understood as the traditional indifference-zone (IZ) formulation \citep{bechhofer1954single}, but it is slightly different. Under the IZ formulation, the decision maker is assumed to know the IZ parameter (i.e., the minimal mean difference) and needs to specify its value in the procedure. Recently, \cite{ni2017efficient} argue that this assumption is unlikely to be true for large-scale problems. In Assumption \ref{assu:asyreg}, however we only require the existence of $\gamma$ and do not 
need to know or to specify its value. In Assumption \ref{assu:asyreg}, the existence of $\bar\sigma^2$ is also a reasonable assumption, as it avoids situations where some variances may go to infinity as $k\to\infty$. Again, we only need its existence and do not need to know its value.

\subsection{Preliminaries on the Running Average}
\label{subsec: runningaverage}
To establish the sample optimality, as shown in Section \ref{sec:bc}, we need to repeatedly use properties of the running-average process $\{\bar X_i(n), n=1,2,\ldots\}$, including its minimum and its boundary-crossing time. In this subsection, we prove several results of the running average that are useful in the rest of this paper.

Let $Z_1,Z_2,\ldots$ be a sequence of independent standard normal observations. Let $\bar Z(n)= \frac{1}{n}\sum_{j=1}^n Z_j$ be the sample average of the first $n$ observations. Then, $\{\bar Z(n), n=1,2,\ldots\}$ denotes the running-average process of the standard normal distribution. First, we consider the minimum of the process. In the following lemma, we show that the running-average process reaches its minimum in a finite number of observations almost surely. The proof of the lemma is included in \ref{subsec: proofargmin}.
\begin{lem}\label{lem:argmin}
    The running-average process $\{\bar Z(n), n=1,2,\ldots\}$ reaches its minimum in a finite number of observations almost surely, i.e., $\Pr\left\{\argmin_{n\ge 1}\bar Z(n)<\infty\right\}=1$.
\end{lem}

Second, we consider the boundary-crossing time. Let $\tilde N(x)=\inf\{n:\bar Z(n)<x\}$. Notice that $\bar Z(n)\to 0$ by the strong law of large numbers. Then, $\tilde N(x)$ displays different behaviors depending on the value of $x$; see, e.g., Corollaries 8.39 and 8.44 of \cite{siegmund1985sequential}. We summarize the necessary results of their behaviors in the following lemma.

\begin{lem}\label{lem:normalbc}
    For any $x>0$, $\Pr\{\tilde N(x)<\infty\}=1$ and 
    \begin{equation}\label{eqn:normbc1}
        \E\left[\tilde N(x)\right]=\exp \left(\sum_{n=1}^{\infty} \frac{1}{n} \Phi\left(-\sqrt{n}x \right)\right),
    \end{equation}
    and furthermore, $\E[\tilde N(x)]$ is continuous and strictly decreasing in $x\in(0,\infty)$, where $\Phi(\cdot)$ denotes the cumulative distribution function of the standard normal distribution.
    
    For any $x<0$, $\E[\tilde N(x)]=\infty$ and
    \begin{equation}\label{eqn:normbc2}
        \Pr\left\{\tilde N(x)<\infty\right\}=1-\exp \left(-\sum_{n=1}^{\infty} \frac{1}{n} \Phi\left(\sqrt{n}x \right)\right),
    \end{equation}
   and furthermore,  $\Pr\{\tilde N(x)<\infty\}$ is continuous and strictly increasing for $x\in(-\infty,0)$.
    
    When $x=0$, $\Pr\{\tilde N(x)<\infty\}=1$ and $ \E[\tilde N(x)]=\infty$. 
\end{lem}
\begin{proof}{Proof.}

The proof is included in \ref{subsec: proofNormalBC}, where we only prove the continuity and strict monotonicity of $\E[\tilde N(x)]$ for $x \in (0, \infty)$. The other results are direct consequences of Corollaries 8.39 and 8.44 of \cite{siegmund1985sequential}. \hfill\Halmos
\end{proof}
\vspace*{6pt}

Notice that there is an interesting link between the minimum and the boundary-crossing time of the running-average process, i.e., when the minimum is above $x$, the boundary-crossing time is infinite, and vice versa. Mathematically, we may write
\[
\left\{\min_{n\ge 1} \bar Z(n) > x\right\}=\left\{\tilde N(x)=\infty\right\}.
\]
Therefore, when $x<0$ and by Equation (\ref{eqn:normbc2}), we have
\begin{equation}\label{eqn:normbc3}
    \Pr\left\{\min_{n\ge 1} \bar Z(n) > x\right\}=\exp \left(-\sum_{n=1}^{\infty} \frac{1}{n} \Phi\left(\sqrt{n}x \right)\right).
\end{equation}
To simplify the notation, we let 
\[
C(x)=\exp \left(\sum_{n=1}^{\infty} \frac{1}{n} \Phi\left(-\sqrt{n}x \right)\right).
\]
Then, by Equations (\ref{eqn:normbc1}) and (\ref{eqn:normbc3}), we have $\E[\tilde N(x)]=C(x)$ when $x>0$ and $\Pr\{\min_{n\ge 1} \bar Z(n) > x\}=[C(-x)]^{-1}$ when $x<0$, respectively.

\subsection{Sample Optimality}
\label{subsec: GreedyOpt_opt}

Following Equation (\ref{eqn:pcs2}), we may establish the sample optimality of the greedy procedure based on the following three arguments. The more rigorous analyses based on these arguments are provided in the proof of the theorem.

\begin{argument}
\label{argu1}
$N_1^*<\infty$ a.s. and, therefore, it does not affect the sample optimality.
\end{argument}

Notice that $\bar X_1(n)=\mu_1+\sigma_1 \bar Z(n)$ statistically, where $\bar Z(n)$ is the sample mean of $n$ independent standard normal observations. Then, it is clear that \[N_1^*=\argmin_{n\ge 1} \bar X_1(n)=\argmin_{n\ge 1} \bar Z(n)<\infty\ \ a.s.\] by Lemma \ref{lem:argmin}. In the consideration of the sample optimality, the total budget $B=ck$ and $k\to\infty$. Therefore, a finite number of observations of $N_1^*$ do not affect the sample optimality.

\begin{argument}
\label{argu2}
Under the condition $\bar X_1^*>\mu_1-\gamma_0$ with $0<\gamma_0<\gamma$, there exists a constant $h>0$, which may depend on $\gamma_0$, such that $\limsup_{k\to\infty} \frac{1}{k}\sum_{i=2}^{k} N_i(\bar X_1^*)\le h$ a.s.
\end{argument}


Under the condition $\bar X_1^*>\mu_1-\gamma_0$, for every $i=2,\ldots,k$, because $\mu_1-\mu_i \ge \gamma$ and $\sigma_i^2\le\bar\sigma^2$, 
\begin{equation}\label{eqn: Nbound}
\begin{aligned}
    N_i\left(\min_{n \geq 1}\bar X_1(n)\right) =& \inf\left\{n\geq 1: \bar{X}_i(n)<  \min_{n \geq 1}\bar X_1(n)\right\}  \\
     \leq &  \inf \{n\geq 1: \bar X_i(n) < \mu_1-\gamma_0\} \\
    =&\inf\left\{n\geq 1: \frac{\bar{X}_i(n)-\mu_i}{\sigma_i} < \frac{\mu_1-\mu_i-\gamma_0}{\sigma_i}\right\}  \\
\leq &\inf\left\{n\geq 1: \bar Z_i(n) <\frac{\gamma-\gamma_0}{\bar \sigma}\right\} \\
     = &\tilde N_i\left(\frac{\gamma-\gamma_0}{\bar \sigma}\right), 
\end{aligned}
\end{equation} 
where $\tilde N_i(x)$ is the $i$th observation of $\tilde N(x)$. Therefore, by the independence of $\tilde N_i(x)$, $i=2,\ldots,k$, the strong law of large numbers and Equation (\ref{eqn:normbc1}), we have
\begin{equation}\label{eqn: AvgBC_SLLN}
\begin{aligned}
    \lefteqn{ \limsup_{k\to\infty} \frac{1}{k} \sum_{i=2}^{k} N_i(\bar X_1^*) }\\
    &\leq  \limsup_{k\to\infty} \frac{k-1}{k}  \frac{1}{k-1} \sum_{i=2}^{k} \tilde N_i\left(\frac{\gamma-\gamma_0}{\bar \sigma}\right) = \E\left[\tilde N\left(\frac{\gamma-\gamma_0}{\bar \sigma}\right)\right] = C\left(\frac{\gamma-\gamma_0}{\bar \sigma}\right) \ \ a.s.
\end{aligned}
\end{equation}

\begin{argument}
\label{argu3}
If the total budget $B=ck$ and $c> C\left(\frac{\gamma-\gamma_0}{\bar \sigma}\right)$, as $k\to\infty$, the PCS is at least $\Pr\{\bar X_1^*>\mu_1-\gamma_0\}>0$.
\end{argument}

Notice that Argument 2 holds under the condition $\bar X_1^*>\mu_1-\gamma_0$. If the condition holds, the total budget $B=ck$ is sufficient to ensure $B\ge \sum_{i=2}^{k} N_i(X_1^*) + N_1^*$, thus a correct selection by Equation (\ref{eqn:pcs2}). Therefore, by Equation (\ref{eqn:normbc3}),
\begin{eqnarray}
    {\rm PCS} &\ge&  \Pr\{\bar X_1^*>\mu_1-\gamma_0\} \nonumber\\
    &=& \Pr\left\{\min_{n\ge 1} \bar X_1(n) > \mu_1-\gamma_0\right\}= \Pr\left\{\min_{n\ge 1} \bar Z(n) > -{\frac{\gamma_0}{\sigma_1}}\right\} = \left[C\left(\frac{\gamma_0}{\sigma_1}\right)\right]^{-1}>0.\label{eqn:pcs-equiv}
\end{eqnarray}

\vspace{11pt}
With the above three arguments, we can rigorously prove the following theorem on the sample optimality of the greedy procedure. Slightly different from the above arguments, we select $\gamma_0\in (0,\gamma)$ such that the PCS lower bound may be maximized given the total budget $B=ck$. The proof is deferred to  \ref{subsec: proofGreedyOpt}.
\begin{thm}
\label{thm: Greedy_PCS_opt}
Suppose that Assumption \ref{assu:asyreg} holds. If the total sampling budget $B$ satisfies $B=ck$ and $c > C\left(\frac{\gamma}{\bar \sigma}\right)$,
the PCS of the greedy procedure satisfies
\[
\liminf\limits_{k \to \infty} {\rm PCS} \geq \Pr\left\{\min_{n\ge 1} \bar X_1(n) > \mu_1-\gamma_0\right\} = \left[C\left(\frac{\gamma_0}{\sigma_1}\right)\right]^{-1}>0,
\]
where $\gamma_0$ is a positive constant satisfying $\gamma_0 \in (0, \gamma)$ and $C\left(\frac{\gamma-\gamma_0}{\bar \sigma}\right) =c$.
\end{thm}

Theorem \ref{thm: Greedy_PCS_opt} confirms that the greedy procedure is indeed sample optimal. It indicates that greedy procedure tends to perform better than existing procedures that are sub-optimal in order when solving large-scale R\&S problems, which explains the findings observed in Figure \ref{fig: simple_example}.

Another interesting observation from Theorem \ref{thm: Greedy_PCS_opt} is that the PCS goes to 1 if the best alternative has no noise, i.e., $\sigma_1=0$. This goes back to the simplified case considered in Section \ref{subsec: bc_simple}. Notice that 
\[
\limsup_{k\to\infty} \frac{1}{k}\left[\sum_{i=2}^{k} N_i(\mu_1) +1 \right] \le C\left(\frac{\gamma}{\overline\sigma}\right)\ \ a.s.
\]
Then, by Equation (\ref{eqn:pcs}), the PCS goes to 1 if $c=B/k>C\left(\frac{\gamma}{\bar\sigma}\right)$.

\subsection{Tightness of the PCS Lower Bound}
\label{subsec: tightness}

It is interesting to point out that the asymptotic PCS lower bound obtained in Theorem \ref{thm: Greedy_PCS_opt} is actually tight. Specifically, there exists a problem configuration of the means $\mu_i$ and variances $\sigma^2_i$, $i=1, \ldots, k,$ under which, given a total sampling budget $B=C\left(\frac{\gamma-\gamma_0}{\overline\sigma}\right) k$, the asymptotic PCS of the greedy procedure is exactly the probability that the best alternative's running-average process remains above the boundary $\mu_1 - \gamma_0$. This configuration is called the slippage configuration with a common variance (SC-CV), where the non-best alternatives have the same means and all alternatives have the same variances. We have the following theorem on the tightness of the asymptotic PCS.
\begin{thm}
\label{thm: Greedy_PCS_tightness}
Suppose that $\mu_1-\mu_i=\gamma > 0$ for all $i=2,\ldots,k$ and $\sigma_i^2=\sigma^2>0$ for all $i=1,\ldots,k$. Then, for any sampling budget with $B=ck$ and $c>C\left(\frac{\gamma}{ \sigma}\right)$, the PCS of the greedy procedure satisfies
\[
\lim\limits_{k \to \infty} {\rm PCS} =\Pr\left\{ \min_{n\ge 1} \bar X_1(n) > \mu_1-\gamma_0 \right\} = \left[C\left(\frac{\gamma_0}{\sigma}\right)\right]^{-1},
\]
where $\gamma_0$ is a positive constant satisfying $\gamma_0 \in (0, \gamma)$ and $C\left(\frac{\gamma-\gamma_0}{ \sigma}\right) =c$.
\end{thm}

The proof of this theorem is somewhat tedious. So we defer it to \ref{subsec: proofPCStightness}, and only introduce the main intuitions here. Considering Theorem \ref{thm: Greedy_PCS_opt}, we only need to show 
\[
\limsup_{k\to\infty}\, \mbox{PCS} \leq\Pr\left\{ \min_{n\ge 1} \bar X_1(n) > \mu_1-\gamma_0 \right\}.
\]
Equivalently, it suffices to show the probability of incorrect selection (PICS) satisfies\footnote{We ignore the event $\min_{n\ge 1} \bar X_1(n) = \mu_1-\gamma_0$ as it is a probability-zero event under the normality assumption.}
\[
  \liminf_{k\to\infty}\, \mbox{PICS} \geq \Pr\left\{ \min_{n\ge 1} \bar X_1(n) < \mu_1-\gamma_0 \right\}.
\]
This statement requires that, under the SC-CV and provided that $B=C\left(\frac{\gamma-\gamma_0}{ \sigma}\right) k$ and $k$ is sufficiently large, if the best alternative's sample mean drops below the boundary $\mu_1 - \gamma_0$, a false selection is inevitable for the greedy procedure.  

The intuition is simple. In short, under the SC-CV, when $k$ is sufficiently large, if $\bar X_1(n)$ ever drops below $\mu_1 - \gamma_0$, the total budget $B=C\left(\frac{\gamma-\gamma_0}{ \sigma}\right)k$ is not sufficient for the alternative $1$ to become the sample-best again. We take a sample-path viewpoint to explain this intuition. Let $m=\min\{n\ge 1: \bar X_1(n)<\mu_1-\gamma_0\}$ be the first time that the running-average process of $\bar X_1(n)$ falls below $\mu_1-\gamma_0$. Let $\bar X_1(m)=\mu_1-\gamma_0-\varepsilon$, where $\varepsilon>0$. For alternative $1$ to become the sample-best again, all $k-1$ inferior alternatives have to fall below $\mu_1-\gamma_0-\varepsilon$, therefore, a total budget of $\sum_{i=2}^{k} N_i(\mu_1-\gamma_0-\varepsilon)$ is needed. Notice that under the SC-CV, $N_i(\mu_1 - \gamma_0-\varepsilon)$ are independent and identically distributed. Then, $\sum_{i=2}^{k} N_i(\mu_1-\gamma_0-\varepsilon)$ is roughly $C\left(\frac{\gamma-\gamma_0-\varepsilon}{ \sigma}\right) k$. Notice that $C(\cdot)$ is a strictly decreasing function as stated in Lemma \ref{lem:normalbc}. Then, $\sum_{i=2}^{k} N_i(\mu_1-\gamma_0-\varepsilon)$ is strictly greater than the given total sampling budget $B=C\left(\frac{\gamma-\gamma_0}{ \sigma}\right) k$ when $k$ is sufficiently large, implying that the budget is insufficient and a false selection is inevitable.

Theorem \ref{thm: Greedy_PCS_tightness} indicates that the SC-CV is the worst case for the greedy procedure, and the asymptotic PCS lower bound of the greedy procedure stated in Theorem \ref{thm: Greedy_PCS_opt} cannot be improved. This result further shows that the boundary-crossing perspective introduced in Section \ref{sec:bc} fully captures the mechanism behind the surprising performance of the greedy procedure for large-scale R\&S problems. 

Theorem \ref{thm: Greedy_PCS_tightness} also provides an interesting linkage between the fixed-budget formulation and the fixed-precision formulation. Even though the greedy procedure is a fixed-budget procedure, the theorem provides an approach to turn it into a fixed-precision procedure with known and equal variances and known indifference-zone parameter ($\gamma$ in our case). In that sense, Theorem \ref{thm: Greedy_PCS_tightness} makes an interesting theoretical contribution. In the fixed-precision R\&S literature, aside of early work of \cite{bechhofer1954single}, to the best of our knowledge, only the BIZ procedure of \cite{frazier2014fully} is known to have a tight PCS bound. Most procedures, including the sample-optimal procedures such as median elimination procedures and KT procedures,  \magenta{use the Bonferroni-type inequalities (e.g., the Bonferroni's or Slepian's inequality) so that the PCS bound is not tight. Reducing the impact of the Bonferroni-type inequalities in R\&S is also an active research topic and has been explored by \cite{nelson1995using}, \cite{eckman2022posterior} and \cite{wangbonferroni}. In the analysis of the greedy procedure, the Bonferroni-type inequalities are not used and a tight PCS bound is achieved.}

Furthermore, we highlight here that, under the SC-CV, there is an upper limit of the PCS no matter how large the total budget $B$ is and the limit is strictly less than 1. 
Notice that by Theorem \ref{thm: Greedy_PCS_tightness}, $\lim\limits_{k \to \infty} {\rm PCS} =\Pr\left\{ \min_{n\ge 1} \bar X_1(n) > \mu_1-\gamma_0 \right\}$ where $\gamma_0\in (0,\gamma)$. Then, the upper bound of the PCS is reached as $\gamma_0$ approaches $\gamma$. Intuitively, under the SC-CV, if the first observation of the alternative $1$ is below $\mu_1-\gamma$, which is the mean value of all other alternatives, a false selection is inevitable. The result is summarized in the following corollary.

\begin{cor}
\label{cor: inconsistency}
Suppose that $\mu_1-\mu_i=\gamma > 0$ for all $i=2,\ldots,k$ and $\sigma_i^2=\sigma^2>0$ for all $i=1,\ldots,k$. Then, for any sampling budget with $B=ck$ and $c>C\left(\frac{\gamma}{ \sigma}\right)$, the PCS of the greedy procedure satisfies
\[
\lim\limits_{k \to \infty} {\rm PCS} \leq \Pr\left\{ \min_{n\ge 1} \bar X_1(n) > \mu_1-\gamma\right\} = \left[C\left(\frac{\gamma}{\sigma}\right)\right]^{-1}<1.
\]
\end{cor}
\begin{proof}{Proof.}
The conclusion is straightforward by noticing that $\gamma_0<\gamma$ in Theorem \ref{thm: Greedy_PCS_tightness}. \hfill\Halmos
\end{proof}

Corollary \ref{cor: inconsistency} reveals a drawback of the greedy procedure, i.e., it is not consistent. Even with $c\to\infty$, the total budget of $B=ck$ is not sufficient to ensure that the worst-case PCS goes to 1. We discuss how to solve this problem in Section \ref{sec: ExploreFirst}.

\subsection{Insights behind the Sample Optimality}

Traditional procedures cannot be sample optimal mainly because they typically view the PCS as the probability that the best alternative beats all $k-1$ inferior alternatives. 
Then, as $k$ increases, more pairwise comparisons are involved and therefore, to keep the PCS from decreasing to zero, the procedures must increase the budget allocation in each pairwise comparison. As a consequence, the total sampling budget $B$ should grow faster than order of $k$ as $k\rightarrow \infty$, leading to the lack of sample optimality of the procedures. To alleviate this, both the ME and KT procedures adopt a halving structure in which the number of necessary pairwise comparisons is reduced. Particularly, the best alternative only has to beat $\lceil \log_2 k \rceil$ alternatives or survive $\lceil \log_2 k \rceil$ rounds to ensure a correct selection. This helps to reduce the growth order of the sampling budget and yield the sample optimality.

The greedy procedure achieves the sample optimality in a completely different manner. Instead of reducing the number of necessary pairwise comparisons, the greedy procedure achieves the sample optimality by aggregating and removing the randomness in the $k-1$ pairwise comparisons as $k\to\infty$. From the boundary-crossing perspective, for the inferior alternative involved in each pairwise comparison, its received sampling budget can be represented by its boundary-crossing time with respect to the best. Although each boundary-crossing time is random, it is found from  the strong law of large numbers that their summation grows linearly in $k$ in a deterministic way. In other words, the randomness of the total sampling budget of all inferior alternatives is removed as $k\to\infty$. Given this, it is intuitively clear why the greedy procedure can be sample optimal without the halving structure.

\subsection{From Correct Selection to Good Selection}
\label{subsec: GreedyOptPGS}

In the previous subsections, we analyze the greedy procedure with the objective of selecting the unique best. In this subsection, we consider an alternative objective of selecting alternatives that are good enough. An alternative $i$ is regarded as \textit{good} (or acceptable) if the mean $\mu_i$ is within an indifference zone (IZ)  of the best mean $\mu_1$, controlled by a user-specified IZ parameter $\delta > 0$, i.e., $\mu_i > \mu_1 - \delta$. A good correction occurs if any one of the good alternatives is selected by the R\&S procedure, and its probability is referred to as the probability of good selection (PGS).  The objective of achieving good selection is of particular interest for large-scale R\&S problems \citep{ni2017efficient}, because, when the number of alternatives $k$ is large, selecting the best may be prohibitively difficult while selecting a good alternative is generally easier. In this subsection, we analyze whether the greedy procedure is also sample optimal with respect to the PGS.

We also take the boundary-crossing perspective. Let the set $\cal G$ denote the set of good alternatives and the set $\cal N$ denote the rest of the alternatives. Notice that $|{\cal G}|+|{\cal N}|=k$. Let $\bar Y_\delta(r)=\max_{i\in{\cal G}} \bar X_i(n_i)$, where $n_i$ denotes the sample size of alternative $i$ and $r=\sum_{i\in{\cal G}} n_i$. Notice that 
\[
\min_{r\ge 1} \bar Y_\delta(r) = \max_{i\in{\cal G}}\ \min_{n\ge 1}\bar X_i(n).
\]
Then, analogous to Equation (\ref{eqn:pcs2}), we have
\begin{equation}\label{eqn:pgs}
{\rm PGS}\ge \Pr\left\{\sum_{i\in{\cal N}} N_i\left(\min_{r\ge 1} \bar Y_\delta(r)\right)+ \argmin_{r\ge 1} \bar Y_\delta(r)\leq B \right\}.
\end{equation}
However, Equation (\ref{eqn:pgs}) is significantly more difficult to analyze than Equation (\ref{eqn:pcs2}), because it is difficult to find a tight upper bound of $\argmin_{r\ge 1} \bar Y_\delta(r)$. One bound that one may use is that
\begin{equation}\label{eqn:pgs2}
\argmin_{r\ge 1} \bar Y_\delta(r) \le \sum_{i\in{\cal G}} \argmin_{n\ge 1}\bar X_i(n),
\end{equation}
i.e., the $\bar Y_\delta (r)$ process reaches its minimum before all $\bar X_i(n)$ processes reach their minimums for all $i\in{\cal G}$. However, because for each alternative $i \in \mathcal{G}$ the $\argmin_{n\ge 1}\bar X_i(n)$ does not admit a finite expectation\footnote{Notice that for each alternative $i$, if the running-average process ever falls below the true mean $\mu_i$, i.e., $N_i(\mu_i) < \infty$, the inequality $\argmin_{n\ge 1}\bar X_i(n) \geq N_i(\mu_i)$ naturally holds. Since $N_i(\mu_i)= \tilde N_i(0)$, by Lemma \ref{lem:normalbc} we further have that $N_i(\mu_i)<\infty$ with probability 1 and $\E[ N_i(\mu_i)]=\infty$. Thus, $\E\left[\argmin_{n\ge 1}\bar X_i(n)\right] \geq \E[ N_i(\mu_i)]=\infty$.}, we cannot use Equation (\ref{eqn:pgs2}) to bound $\argmin_{r\ge 1} \bar Y_\delta(r)$ by $c_1|{\cal G}|$, and thus $c_1 k$, for some constant $c_1>0$ when $|{\cal G}|\to\infty$ as $k\to\infty$. Therefore, we can only prove the following proposition that assumes that $|{\cal G}|<\infty$ as $k\to\infty$, i.e., the number of good alternatives is bounded by a constant no matter how large $k$ is. The proof of the proposition is included in  \ref{subsec: proofgreedyPGS}.

\begin{prop}
\label{prop: Greedy_PGS_opt}
Suppose that $\sigma^2_i \leq \bar \sigma^2 < \infty$ for all $i=1,\ldots,k$ and $\limsup_{k\to\infty}|{\cal G}|<\infty$ for a given $\delta>0$. If the total sampling budget $B$ satisfies $B/k = c$ and $c > C\left(\frac{\delta}{\bar \sigma}\right)$, the PGS of the greedy procedure satisfies
\[
\liminf\limits_{k \to \infty} {\rm PGS} \geq \limsup_{k\to\infty}\Pr\left\{ \exists\ i\in{\cal G} : \min_{n\ge 1}\bar X_{i}(n) > \mu_1-\delta_0 \right\}\geq \Pr\left\{\min_{n\ge 1}\bar X_{1}(n) > \mu_1-\delta_0\right\},
\]
where $\delta_0$ is a positive constant satisfying $\delta_0 \in (0, \delta)$ and $C\left(\frac{\delta-\delta_0}{\bar \sigma}\right) =c$.

\end{prop}

Notice that the correct selection is a special case of the good selection, where there is only one good alternative. Therefore, Theorems \ref{thm: Greedy_PCS_opt} and \ref{thm: Greedy_PCS_tightness} are special cases of Proposition \ref{prop: Greedy_PGS_opt}. 
\magenta{Moreover, Proposition \ref{prop: Greedy_PGS_opt} is versatile enough to encompass scenarios where multiple alternatives share an equally best mean performance. For these scenarios, selecting any of these alternatives is considered correct. It is evident that the associated PCS is equivalent to the PGS when $\delta$ is set as the mean difference between the ``tied" best alternatives and others with strictly smaller mean performances.  As a result, we obtain a version of Theorem \ref{thm: Greedy_PCS_opt} that accounts for multiple tied-best alternatives.}
Furthermore, from the proposition, we can see that a larger $\delta$ value often leads to a larger PGS because we may choose a large $\delta_0$ with the same budget. For the more general case where ${\cal G}$ may contain an infinite number of alternatives as $k\to\infty$, we need a better upper bound of $\argmin_{r\ge 1} \bar Y_\delta(r)$. We leave it as a topic for future research.

\magenta{Apart from the PCS and PGS discussed above, the expected opportunity cost (EOC) serves as another widely used measure for procedure effectiveness. As a comparison, the EOC measures the linear loss caused by a possible false selection, while the PCS and PGS measures the 0-1 loss.}
\magenta{As demonstrated in the literature \citep{chick2001new, branke2007selecting, gao2017new}, the procedures designed for one measure often exhibit favorable performance when evaluated in terms of another. In light of this connection, examining whether the greedy procedure remains sample optimal in terms of the EOC would be of interest. However, we think that solving this issue is beyond the scope of this paper, and we leave it as a topic for future research.}

\section{The Explore-First Greedy Procedure}
\label{sec: ExploreFirst}

Despite the greedy procedure being sample optimal, it has an important theoretical drawback. As shown in Corollary \ref{cor: inconsistency}, even as the $c$ in the total sampling budget $B=c k$ grows to infinity, the limiting PCS is strictly less than 1,  indicating that the greedy procedure is inconsistent. 
In the literature, fixed-budget R\&S procedures are typically shown to bear the property of consistency (see, e.g., \citealt{frazier2008knowledge}, \citealt{wu2018analyzing} and \citealt{chen2022balancing}). We believe that it is a fundamental requirement for fixed-budget R\&S procedures. Therefore, in this section, we investigate how to strengthen the greedy procedure to make it consistent.

The conventional definition of consistency for R\&S procedures only considers  a fixed problem size $k$ \citep{frazier2008knowledge}. For a total sampling budget $B$, it examines whether $\lim_{B \to \infty} \rm{PCS} = 1$ or equivalently, whether there exists a correspondent $B$ such that $\rm{PCS} \geq \alpha$ for any $\alpha \in (0, 1)$. 
However, in defining and analyzing the sample optimality of fixed-budget procedures, the asymptotic regime of letting $k \to \infty$  is involved. Therefore, the conventional
definition of consistency is extended here to account for $k$ explicitly. In this paper, we define the consistency of sample-optimal fixed-budget procedures as follows. 

\begin{defn}
\label{def: large_scale_consistency}
A sample-optimal fixed-budget R\&S procedure is consistent if the PCS of the procedure satisfies that, for any $\alpha \in (0, 1) $, there exists a constant $c>0$ such that,
 $$\liminf\limits_{k \to \infty}\, {\rm PCS} \geq  \alpha, \mbox{ for } B=ck.$$
\end{defn}

The above definition requires that, as long as the total sampling budget grows faster than the order of $k$ (e.g., in the order of $k\log \log k$), the PCS should converge to 1 even when $k$ increases to infinity. Notice that this requirement is impossible to achieve for those non-sample-optimal procedures. Take the OCBA procedure as an example. For the procedure, a total sampling budget that grows in the order of $k\log \log k$ does not even suffice to maintain a non-zero PCS as the problem size $k$ becomes sufficiently large.

\subsection{The Procedure Design}
\label{subsec: EFGdesign}
The reason behind the greedy procedure's inconsistency is easy to understand. Consider the SC-CV configuration aforementioned.  Corollary \ref{cor: inconsistency} states that no matter how large the total sampling budget is, the limiting PCS of the greedy procedure is upper bounded by $\Pr\left\{\min_{n\ge 1} \bar X_{1}(n) > \mu_1 -\gamma\right\}$. Therefore, if $\bar X_1(n)$ ever falls below $\mu_1-\gamma$, an incorrect selection will occur no matter how large the sampling budget is.

To resolve the inconsistency issue, instead of allocating one observation to each alternative at the beginning of the procedure, we allocate $n_0$ observations to each alternative. This may be viewed as adding an exploration phase to the purely exploitative greedy procedure. We name the new procedure the explore-first greedy (EFG) procedure, and its detailed description is listed in Procedure \ref{procedure: ExploreFirstGreedy}.

\begin{procedure}[ht]
\caption{\textbf{Explore-First Greedy Procedure}}
\label{procedure: ExploreFirstGreedy}
\begin{algorithmic}[1]
\REQUIRE {$k$ alternatives $X_1,\ldots,X_k$}, the total sampling budget $B$, and the exploration size $n_0$.
\STATE For all $i=1,\ldots,k$, take $n_0$ independent observations $X_{i1},\ldots,X_{in_0}$ from alternative $i$, set $\bar X_i(n_0)=\frac{1}{n_0}\sum_{j=1}^{n_0} X_{ij}$, and let $n_i=n_0$. 
\WHILE{$\sum_{i=1}^k n_i < B$}
\STATE Let $s = \arg\max_{i=1,\ldots,k}  \bar{X}_{i}\left(n_i\right)$ and take one observation $x_{s}$ from alternative $s$;
\STATE Update $\bar X_{s}(n_{s}+1) = \frac{1}{n_{s}+1}\left[n_{s}\bar X_{s}(n_s) + x_s\right]$ and let $n_s = n_s+1$;
\ENDWHILE
 \STATE Select $\arg\max_{i\in\{1,\ldots,k\}} \bar{X}_{i}\left(n_i\right)$  as the best.
\end{algorithmic}
\end{procedure}

In the EFG procedure, roughly speaking, the greedy phase plays the role of maintaining the sample optimality, while the exploration phase plays the role of increasing the PCS. Without the greedy phase, the exploration phase itself becomes the equal allocation procedure that is sub-optimal in order \citep{hong2022solving}. Provided with a total sampling budget of $B=c k$, the EFG procedure allocates a sampling budget $n_0 k$ to the exploration phase and finishes the greedy phase with the remaining sampling budget $n_g k$ where $n_g= c- n_0$. After the total sampling budget is exhausted, the alternative with the largest sample mean is selected as the best. Notice that the greedy procedure is a special case of the EFG procedure with $ n_0=1$.

\subsection{Preliminaries on Running Average}
\label{subsec: runningave_explore}

Similar to Section \ref{sec: GreedyOpt}, we may establish the sample optimality and the consistency of the EFG procedure by analyzing the associated running-average process $\{\bar{X}_i(n), n=n_0,n_0+1,\ldots\}$, for each alternative $i=1,2,\ldots,k$, which starts from $n=n_0$. Following Section \ref{subsec: runningaverage}, we consider the running-average process $\{\bar{Z}(n), n=n_0,n_0+1,\ldots\}$ of the standard normal distribution. First, we show in the following lemma that it still reaches its minimum within a finite number of observations almost surely, regardless of the starting time $n_0$. The proof is straightforward based on Lemma~\ref{lem:argmin}, and we provide it in \ref{subsec: proofargminEFG}.
\begin{lem}\label{lem:argmin_EFG}
    The running-average process $\{\bar Z(n), n=n_0,n_0+1,\ldots\}$ with $n_0\ge 1$ reaches its minimum in a finite number of observations almost surely, i.e., $\Pr\left\{\argmin_{n\geq n_0}\bar Z(n)<\infty\right\}=1$.
\end{lem}

Secondly, let $\tilde{N}(x;n_0)=\inf\{n\geq n_0: \bar{Z}(n)<x\}$. Notice that when $n_0=1$, $\tilde{N}(x;n_0)$ becomes the boundary-crossing time $\tilde{N}(x)$ that is discussed in Section \ref{subsec: runningaverage}. There is an interesting relationship between $\tilde{N}(x;n_0)$ and $\tilde{N}(x)$. Specifically, we have
\begin{eqnarray*}
\lefteqn{ \inf\{n\geq 1: \bar{Z}(n)<x\}\le \inf\{n\geq n_0: \bar{Z}(n)<x\} }\\
&\le & n_0\inf\{m\geq 1: \bar{Z}(mn_0)<x\} = n_0\inf\left\{m\geq 1: \bar{Z}(m)<{\frac{x}{\sqrt{n_0}}}\right\},
\end{eqnarray*}
where the last equality follows the fact that $\bar{Z}(mn_0)=\sqrt{n_0}\,\bar{Z}(m)$ statistically. Then, the inequality above can be rewritten as
\begin{equation}\label{eqn: NvsNE}
\tilde{N}(x)\leq \tilde{N}\left(x;n_0\right)  \leq n_0 \tilde{N}\left(\frac{x}{\sqrt{n_0}}\right).
\end{equation}
Combining Equation (\ref{eqn: NvsNE}) and Lemma~\ref{lem:normalbc}, we can easily prove the following lemma, whose proof is omitted.
\begin{lem}\label{lem:normalbc_EFG}
  Given any positive integer $n_0$, then we have:
  \begin{enumerate}[(1)]
      \item for any $x>0$, $\Pr\left\{\tilde N(x;n_0)<\infty\right\}=1$ and $\E\left[\tilde{N}(x;n_0)\right]<\infty$;
      \item for any $x<0$, $0<\Pr\left\{\tilde N(x;n_0)<\infty\right\}<1$ and $\E\left[\tilde{N}(x;n_0)\right]=\infty$;
      \item for $x=0$, $\Pr\left\{\tilde N(x;n_0)<\infty\right\}=1$ and $ \E\left[\tilde N(x;n_0)\right]=\infty$.
  \end{enumerate}
\end{lem}
Unlike in Lemma~\ref{lem:normalbc}, however, the exact expression of $\E\left[\tilde{N}(x;n_0)\right]$ is not available when $n_0 \geq 2$.
For the sake of easy presentation, we let $C(x; n_0) = \E\left[\tilde{N}(x;n_0)\right]$ in the rest of this paper.

\subsection{Sample Optimality}
\label{subsec: EFGrateOPT}
To analyze the PCS of the EFG procedure, we use again the boundary-crossing perspective. Similar to Section \ref{subsec: bc_general}, we let $\bar X_1^*(n_0) = \min_{n\ge n_0} \bar X_1(n)$ be the minimum of the running-average process of the alternative $1$ with $n \ge n_0$, and use it as the boundary. Then, let $N_i(\bar X_1^*(n_0); n_0)=\inf\left\{n\ge n_0: \bar X_i(n)<\bar X_1^*(n_0) \right\}$ denote alternative $i$'s boundary-crossing time with respect to the boundary $\bar X_1^*(n_0)$ for all $i=2,\ldots,k$. As in the case of $n_0=1$ in Section \ref{subsec: bc_general}, once the alternative $1$  becomes the sample best after it reaches its minimum $\bar X_1^*(n_0)$, it remains as the best until all budget is exhausted. Let $N_1^*(n_0)=\argmin_{n\ge n_0} \bar X_1(n)$. Similar to Equation \eqref{eqn:pcs2}, the PCS of the EFG procedure satisfies
\begin{equation}
\label{eqn:pcs_EFG}
{\rm PCS}\ge \Pr\left\{\sum_{i=2}^{k} N_i(\bar X_1^*(n_0); n_0) +N_1^*(n_0) \le B\right\},
\end{equation}
where $B$ is the total sampling budget. Notice that Equation \eqref{eqn:pcs_EFG} generalizes Equation \eqref{eqn:pcs2} from $n_0=1$ to $n_0 \geq 1$.

Combing Equation \eqref{eqn:pcs_EFG} and Lemma \ref{lem:normalbc_EFG}, we are able to prove the following theorem on the sample optimality of the EFG procedure and also the asymptotic tightness under the SC-CV. The proof almost replicates the proofs of Theorems \ref{thm: Greedy_PCS_opt} and \ref{thm: Greedy_PCS_tightness} with simple modifications on the necessary notation, and we include it in  \ref{subsec: proofEFGopt}.

\begin{thm}
\label{thm: ExploreFirstGreedy_opt}
Suppose that Assumption \ref{assu:asyreg} holds. 
 For any $n_0 \geq 1$,
if the total sampling budget $B$ satisfies
$B/k = (n_0 + n_g)$ and $n_g  >C \left(\frac{\gamma}{ \bar \sigma}; n_0\right)-n_0$, 
the PCS of the EFG procedure satisfies
\[
\liminf\limits_{k \to \infty} {\rm PCS} \geq \Pr\left\{\min_{n\ge n_0} \bar X_1(n) > \mu_1-\gamma_0\right\},
\]
where $\gamma_0$ is a positive constant satisfying $\gamma_0 \in (0, \gamma)$ and $C\left(\frac{\gamma-\gamma_0}{\bar \sigma}; n_0\right) =n_g+n_0$. Furthermore, if for every non-best alternative $i=2,\ldots,k$, $\mu_i=\mu_1-\gamma$, and for every alternative $i=1,\ldots,k$, $\sigma_i^2 = \sigma^2>0$, 
the PCS of  the EFG procedure satisfies
\[
\lim\limits_{k \to \infty} {\rm PCS} = \Pr\left\{\min_{n\ge n_0} \bar X_1(n) > \mu_1-\gamma_0\right\}.
\]
\end{thm}

Theorem \ref{thm: ExploreFirstGreedy_opt} confirms the sample optimality of the EFG procedure in solving large-scale R\&S problems. Moreover, similar to Theorem \ref{thm: Greedy_PCS_tightness}, the PCS lower bound is tight under the SC-CV.
Recall that, when $n_0$ = 1, the EFG procedure is the greedy procedure. Therefore, we may view Theorem \ref{thm: ExploreFirstGreedy_opt} as a generalization of Theorems \ref{thm: Greedy_PCS_opt} and \ref{thm: Greedy_PCS_tightness} from $n_0=1$ to any $n_0\ge 1$. 

Similarly, we may also generalize the PGS result of Proposition \ref{prop: Greedy_PGS_opt} from $n_0=1$ to any $n_0\ge 1$. We include it in  \ref{subsec: proofEFGPGS} for completeness.

\subsection{Consistency}
\label{subsec: consistency}
Now we prove the consistency of the EFG procedure. Notice that by the strong law of large numbers, $\lim_{n_0 \to \infty} \Pr\left\{\min_{n\ge n_0} \bar X_1(n) > \mu_1-\gamma_0\right\} = 1$ for any fixed value of $\gamma_0 > 0$. Then obviously, we have the following lemma.
\begin{lem}
\label{lem: alpha}
For any fixed value of $\gamma_0 > 0$ and any $\alpha\in(0,1)$, there exists a finite value of $n_0\ge 1$, which may depend on $\gamma_0$ and $\alpha$, such that 
\begin{eqnarray}
\label{ineq: alpha}
\Pr\left\{\min_{n\ge n_0} \bar X_1(n) > \mu_1-\gamma_0\right\}  \geq \alpha.
\end{eqnarray}
\end{lem}

Notice that given any $\gamma_0\in(0,\gamma)$ and any $\alpha\in(0,1)$, we may determine $n_0$ such that Equation (\ref{ineq: alpha}) holds. Furthermore, we may also determine $n_g$ based on Theorem \ref{thm: ExploreFirstGreedy_opt}, i.e., $n_g=C\left(\frac{\gamma-\gamma_0}{\bar \sigma};  n_0\right)-n_0$. Then, by Theorem \ref{thm: ExploreFirstGreedy_opt}, if $B\ge (n_0+n_g)k$, the asymptotic lower bound of the PCS of the EFG procedure is at least $\alpha$. This leads to the following theorem that shows the EFG procedure is consistent, based on Definition \ref{def: large_scale_consistency}. The proof is straightforward and is omitted. 

\begin{thm}
\label{thm: Consistency}
Suppose that Assumption \ref{assu:asyreg} holds. For any given PCS level $\alpha\in (0,1)$, if the total sampling budget $B$ satisfies $B/k = (n_0 + n_g)$ with
$n_0$ satisfying Equation \eqref{ineq: alpha} for a $\gamma_0 \in (0, \gamma)$ and $n_g \geq C\left(\frac{\gamma-\gamma_0}{\bar \sigma};  n_0\right)-n_0$, 
the PCS of the EFG procedure satisfies $\liminf_{k \to \infty} {\rm PCS} \geq \alpha$.
\end{thm}

Theorem \ref{thm: Consistency} establishes the consistency of the EFG procedure. It shows that, adding a simple exploration phase to the greedy procedure overcomes its inconsistency while keeping its sample optimality. 

We also want to highlight that Theorem \ref{thm: Consistency} also provides us a simple way to turn the fixed-budget  EFG procedure into a fixed-precision procedure from R\&S problems with $\gamma$ and $\bar\sigma^2$ are known. Notice that under the Assumption \ref{assu:asyreg}, 
\begin{eqnarray}
  \label{eq: EFG_PCS_trans}
\Pr\left\{\min_{n\ge n_0} \bar X_1(n) > \mu_1-\gamma_0\right\}  = \Pr \left\{\min_{n\ge n_0} \bar Z(n) > -\frac{\gamma_0}{ \sigma_1}\right\} \ge  \Pr \left\{ \min_{n\ge n_0} \bar Z(n) > -\frac{\gamma_0}{ \bar \sigma}\right\}. \label{ineq: alpha2}
\end{eqnarray}
Then, for any $\gamma_0\in(0,\gamma)$ (e.g., $\gamma_0=\gamma/2$), given the PCS target $1-\alpha$, we may use Equation (\ref{ineq: alpha2}) to determine $n_0$ and $n_g=C\left(\frac{\gamma-\gamma_0}{\bar \sigma};  n_0\right)-n_0$ to determine $n_g$. With such $n_0$ and $n_g$, given a total sampling budget $B=(n_0+n_g)k$, it is expected that the EFG procedure can deliver asymptotically a PCS no smaller than the targeted level $1-\alpha$.

\subsection{Budget Allocated to Exploration and Exploitation}
\label{subsec: allocation}
The EFG procedure has two phases, the exploration phase and the greedy (i.e., exploitation) phase. It is interesting to think about how the total budget should be allocated to these two phases, i.e., how to determine the relationship between $n_0$ and $n_g$. Intuitively, one might think that a significant amount of the budget, if not the most, should be allocated to the greedy phase, because it is the reason why the EFG procedure is sample optimal and also because the exploration phase alone is the equal allocation procedure, which performs poorly for large-scale problems. In this subsection we show that, surprisingly, this intuition is {\it not} true, and only a very small amount of budget is needed for the greedy phase when the total budget is sufficiently large. 

By Theorem \ref{thm: ExploreFirstGreedy_opt}, we let $n_g = C\left(\frac{\gamma-\gamma_0}{\bar \sigma}; n_0\right) - n_0 $ for some fixed $\gamma_0\in (0,\gamma)$. Then, we have the following lemma on the relationship between $n_0$ and $n_g$, which is proved in \ref{subsec: proofCto0}.

\begin{lem}
\label{lem: bct_mean_convergence}
For any $\gamma_0\in(0,\gamma)$, there exist positive constants $\beta$ and $\kappa$ such that $n_g\le \beta e^{-\kappa n_0}$, where $\kappa= \frac{(\gamma-\gamma_0)^2}{2\bar\sigma^2}$ and $\beta=\left(1-e^{-\kappa}\right)^{-1}$.
\end{lem}

Lemma \ref{lem: bct_mean_convergence} is a very interesting result, because it shows that $n_g$ goes down to zero in an exponential rate as $n_0$ increases to infinity. Therefore, when the total budget is sufficiently large (so is $n_0$), only a very small amount of budget is needed for the greedy phase, and this small amount is sufficient to turn a poor equal allocation procedure (i.e., the exploration phase) into a sample-optimal and consistent procedure. Based on this observation, we suspect that the greedy phase may also be used as a remedy for other non-sample-optimal procedures, such as the OCBA procedures, to make them sample optimal as well. We leave it as a topic for future research.

When using the EFG procedure to solve practical fixed-budget R\&S problems, $\gamma$ and $\bar\sigma^2$ are typically unknown. Therefore, we cannot use the allocation rule of Lemma \ref{lem: bct_mean_convergence} directly. Instead, we suggest to use a simple proportional rule that allocates $p\in(0,1)$ proportion of the total budget to the exploration phase and the remaining $1-p$ proportion to the greedy phase. Lemma \ref{lem: bct_mean_convergence} indicates that, for a reasonable choice of $p$ (e.g., $p=60\%$ or $p=80\%$), as long as $c$ is sufficiently large, $n_g = (1-p)c$ is enough to maintain the sample optimality. Numerical results show that the PCS of the EFG procedure is not sensitive to the value of $p$ for a relatively large $c$. We illustrate this observation in Section \ref{sec: numerical}.

\magenta{We note here that the issue of budget allocation between an initial EA phase and a subsequent sequential selection phase is very common  \citep{hartmann1991improvement, wu2020algorithm}. Different from our treatment, the most prevalent practice is to choose a small constant $n_0$ for each alternative in the EA phase regardless of how large $B$ is. Interestingly, we notice that a similar proportional allocation rule to the one discussed above is used by \cite{wu2018analyzing, wu2018provably} for a totally different purpose, where $n_0$ is set as $\lfloor p B \rfloor$ for a fixed $p \in (0, 1)$. The authors show that for several sequential OCBA procedures, doing this can accelerate the convergence of the PCS to 1  as $B \rightarrow \infty$ for a fixed $k$, and only a small value of $p$ (e.g., $p=20\%$) is sufficient. We refer the interested readers to the papers for more details.}


\subsection{Enhancing the EFG Procedure}
\label{subsec: EFGPlus}
In this subsection, we enhance the EFG procedure by altering the budget allocation in the exploration phase. From Theorem~\ref{thm: ExploreFirstGreedy_opt}, it can be found that, to ensure a higher PCS of EFG procedure, we only need to allocate more observations to the best alternative (i.e. alternative 1), instead of all alternatives in the exploration phase.
However, the identity of the best alternative is unknown. For this reason, we add a seeding phase, where a small proportion of the sampling budget is used to determine the seeding of all alternatives (for instance, based on the rankings of the sample means). We then allocate the budget of the exploration phase according to the seeds of all alternatives, so that good alternatives tend to have more observations. We call the new procedure the enhanced explore-first greedy (EFG$^+$) procedure and leave its details in Procedure \ref{procedure: EFGPlus}\magenta{, which is put in \ref{subsec: proofEFGP_Opt}}.

In the seeding phase, we generate $n_{sd}$ observations from each alternative, rank all the alternatives in a descending order based on the sample means, and then separate them into $G$ uneven groups denoted by $I^1, \ldots, I^G$. By construction, $I^1$ contains the fewest alternatives,  $I^2$ contains twice as many alternatives as $I^1$, and so on and so forth. Subsequently, in the exploration phase, we equally allocate the exploration budget $n_0 k$ among the $G$ groups, and for each group $r$, we allocate an equal number of observations $n^r=\left\lfloor \frac{n_0 (2^G-1)}{G 2^{r-1} } \right\rfloor$ to every alternative within the group. 
In this way\footnote{In \ref{subsec: proofEFGP_Opt}, we show that the total amount of observations collected in the exploration phase (i.e., $\sum_{r=1}^G n^r |I^r|$), does not exceed the given exploration budget $n_0 k$. }, since the best alternative is very likely to be seeded into the top groups with a small $r$ value, it is expected to receive more observations than $n_0$ in the exploration phase.

Following the boundary-crossing perspective, we can also prove that the EFG$^+$ procedure is sample optimal in the following theorem. The detailed proof is deferred to \ref{subsec: proofEFGP_Opt}. 

\begin{thm}
\label{thm: EFGP_Opt}
Suppose that Assumption \ref{assu:asyreg} holds. 
 For the chosen $G\geq 2$ and any $n_0 \geq G$,
if the total sampling budget $B$ satisfies
$B/k = (n_{sd} + n_0 + n_g)$ where $n_{sd}$ is a positive integer, $n_g > C \left(\frac{\gamma}{ \bar \sigma};  n^1 \right)-n_0$ and  $n^1=\left\lfloor \frac{n_0(2^G-1)}{G} \right\rfloor$, 
the PCS of the EFG$^+$ procedure satisfies
\[
\liminf\limits_{k \to \infty} {\rm PCS} \geq \Pr\left\{\min_{n\ge  n^G} \bar X_1(n) > \mu_1-\gamma_0\right\},
\]
where $n^G=\left\lfloor \frac{n_0(2^G-1)}{G2^{G-1}} \right\rfloor$ and $\gamma_0$ is a positive constant satisfying $\gamma_0 \in (0, \gamma)$ and $C\left(\frac{\gamma-\gamma_0}{ \bar \sigma};  n^1 \right) =n_0+n_g$.
\end{thm}

Compared to Theorem \ref{thm: ExploreFirstGreedy_opt}, Theorem \ref{thm: EFGP_Opt} is actually a weaker result, as the asymptotic tightness of the PCS lower bound no longer holds.
However, numerical experiments in Section \ref{sec: numerical} demonstrate that the EFG$^+$ procedure may significantly outperform the EFG procedure. Moreover, notice that $\lim_{n_0 \rightarrow \infty} n^G=\infty$.
Then, we can easily prove that the EFG$^+$ procedure is consistent as well.

We close this subsection by noting that any budget allocation rule may be applied to the seeding phase, not limited to the equal allocation rule. Since the observations taken for seeding are not used in the following phases, the sample optimality and consistency still hold no matter what rule is adopted. This opens a door for further improving the EFG$^+$ procedure, and we leave this topic for future investigation.

\subsection{\magenta{Parallelization of the EFG procedures}}
\label{subsec: parallel}

\magenta{To take advantage of the ubiquitously available parallel computing environments, we now delve into the parallelization of the EFG procedures. Given the similar structure between the EFG$^+$ procedure and the relatively simpler EFG procedure, our focus narrows down to parallelizing the EFG$^+$ procedure. Recall that the EFG$^+$ procedure has three phases, the seeding phase, exploration phase and greedy phase. The seeding phase and exploration phase are readily parallelizable, as both phases simply employ a static-allocation approach. Indeed, the challenge lies in parallelizing the inherently sequential greedy phase.}

\magenta{From preliminary numerical experiments, we observe that in the greedy phase, typically only a very small proportion (e.g., 2\%) of the alternatives are sampled, and each of them may be sampled many times. We refer to  \ref{subsec: allocationAlternativesMain} for more discussions of this phenomenon. Inspired by this, we propose to use a batching approach to parallelize the greedy phase. Specifically, at each stage of the greedy phase, instead of taking a single observation from the current best, we take a batch of $m >1$ observations from it. With multiple computing processors (e.g., CPUs), the simulation task of the $m$ observations can be distributed among these processors to achieve parallelization. Through batching, the previously sequential greedy phase is rendered parallelizable.}

\magenta{For convenience, we label the parallel EFG$^+$ procedure with a batched greedy phase as the EFG$^{++}$ procedure.  The formal description of the EFG$^{++}$ procedure is available in \ref{sec: parallelization}, along with a comprehensive numerical study of its performance and parallel efficiency in a master-worker parallel computing environment. The numerical study reveals that the EFG$^{++}$ procedure can deliver almost the same PCS and PGS as the EFG$^+$ procedure, even with a relatively large batch size (e.g., $m=200$). At the same time, it can effectively reduce the wall-clock time compared to the EFG$^+$ procedure when using a small number of parallel processors. However, its parallel efficiency diminishes as the number of processors increases. This may be due to the two potential drawbacks of the batching approach. First, synchronization among processors at each stage could be inefficient when the simulation times are random and unequal. Second, the master processor has to frequently communicate with the worker processors to assign simulation tasks and collect simulation results, incurring a substantial communication overhead when the number of processors is large. Strategies for mitigating these drawbacks are also explored in \ref{sec: parallelization}.}

\section{Numerical Experiments}
\label{sec: numerical}
In this section, we conduct numerical experiments to examine our theoretical results and test the performance of the proposed procedures.
 In Section \ref{subsec: numerical1}, we show the sample optimality of the greedy procedure and the EFG procedure. Then, in Section \ref{subsec: numerical2}, we demonstrate that the EFG procedure is consistent. Moreover, we show the impact of budget allocation between exploration and exploitation on the procedure's performance.
In Section \ref{subsec: numerical4}, we compare our EFG and EFG$^+$ procedures with existing sample-optimal fixed-budget R\&S procedures on a practical large-scale R\&S problem. 
\magenta{Lastly, in Section \ref{subsec: additional}, we summarize the additional experiments conducted to understand the EFG procedure's budget allocation behaviors and the procedure's performance when the assumptions on the problem configurations are violated. }
Program codes are written using the Python language and available at \href{https://github.com/largescaleRS/greedy-procedures}{https://github.com/largescaleRS/greedy-procedures}.

Throughout the first two subsections, we consider the following four problem configurations:
\begin{itemize}
    \item the slippage configuration of means with a common variance (SC-CV) under which
$$\mu_1=0.1,\, \mu_i = 0\,, i = 2, \ldots, k \text{ and } \sigma_i^2 = 1,\, i = 1, \ldots, k;$$
    \item the configuration with equally-spaced means and a common variance (EM-CV) under which
    $$\mu_1=0.1,\, \mu_i = {-(i-1)}/{k},\,  i = 2, \ldots, k\text{ and }  \sigma_i^2 =  1,\,  i = 1, \ldots, k;$$
    \item the configuration with equally-spaced means and increasing variances (EM-IV) under which
    $$\mu_1=0.1,\, \mu_i = {-(i-1)}/{k},\,  i = 2, \ldots, k\text{ and }  \sigma_i^2 =  1 + {(i-1)}/{k},\,  i = 1, \ldots, k;$$
    \item the configuration with equally-spaced means and decreasing variances (EM-DV) under which
    $$\mu_1=0.1,\, \mu_i = {-(i-1)}/{k},\,  i = 2, \ldots, k\text{ and }  \sigma_i^2 =  2 - {(i-1)}/{k},\,  i = 1, \ldots, k.$$
\end{itemize}
Notice that for all four configurations, alternative 1 is the unique best alternative. For each configuration, unless otherwise specified, we set the total number of alternatives as $k=2^{l}$ with $l$ ranging from 2 to 16, and for each $k$, we set the total sampling budget as $B=100k$. We estimate the PCS of a procedure in solving a particular R\&S problem based on 1000 independent macro replications.

\subsection{Sample Optimality of the Greedy Procedure and the EFG Procedure}
\label{subsec: numerical1}
In this subsection, we test the greedy procedure and the EFG procedure to verify the sample optimality and compare them with the EA  procedure. 
When applying the EFG procedure, one must decide the budget allocation to the exploration and greedy phases. 
We consider allocating a fixed proportion $p=n_0/c$ of the total sampling budget to the exploration phase.
In this experiment, we let $p=0.8$. Later we will show how the value of $p$ impacts the performance of the EFG procedure. Then, for the three procedures, we plot the estimated PCS against different $k$ under each configuration in Figure \ref{numerical: GreedyEFGRateOptimality}.

\begin{figure}[htbp]
  \centering
  \includegraphics[width=0.75\textwidth]{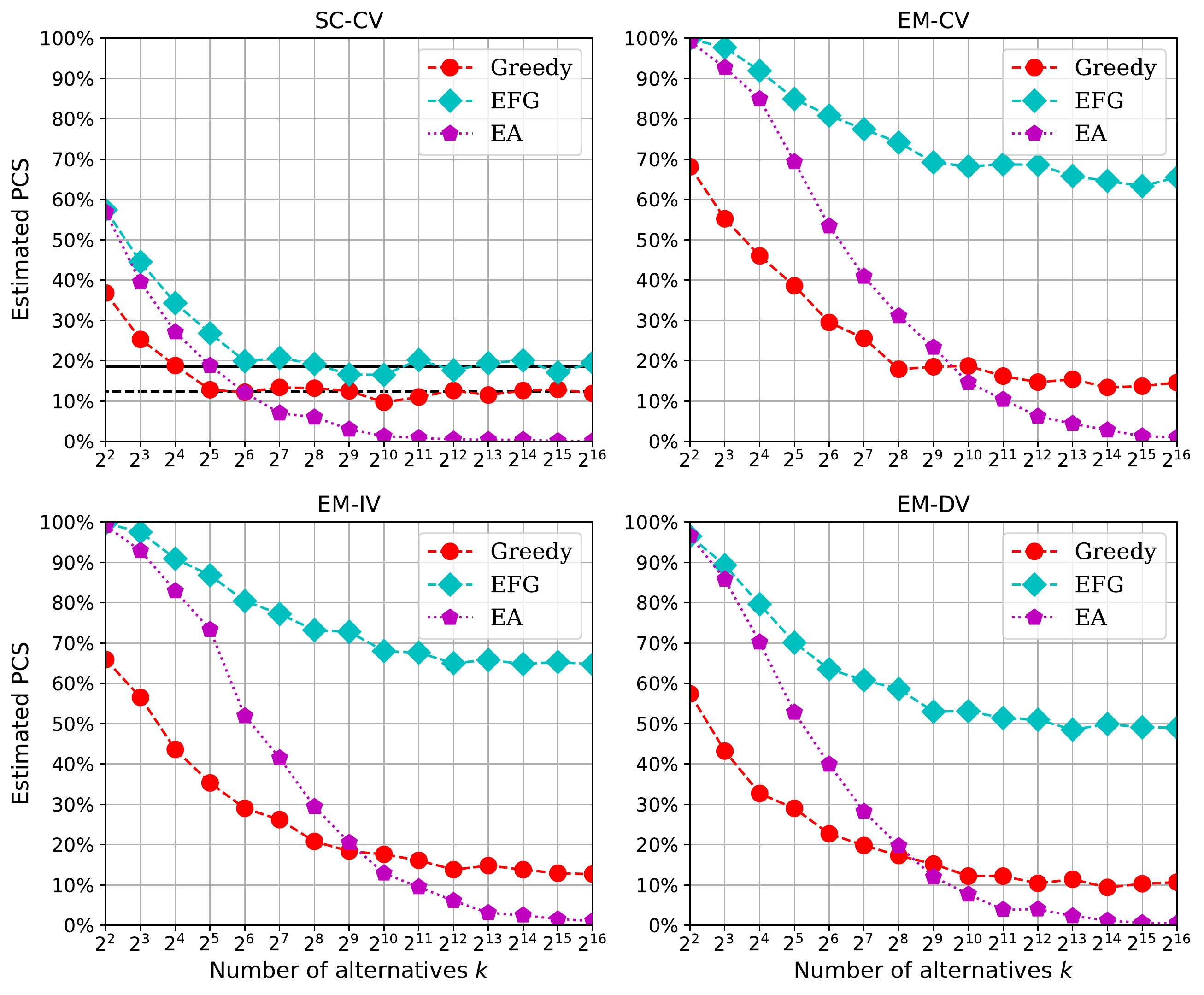}
  \caption{A comparison between the greedy procedure, the EFG procedure and the EA procedure.}
  \label{numerical: GreedyEFGRateOptimality}
\end{figure}

We highlight the main findings from Figure \ref{numerical: GreedyEFGRateOptimality}. 
First, for each configuration, the estimated PCS of the greedy procedure and the EFG procedure remain above a non-zero level regardless of the increase of $k$, proving their sample optimality. 
Second, the asymptotic PCS lower bound of the greedy (EFG) procedure stated in Theorem \ref{thm: Greedy_PCS_tightness} (Theorem \ref{thm: ExploreFirstGreedy_opt}) is tight. Under SC-CV, the estimated PCS of the greedy (EFG) procedure quickly converges to the theoretical value annotated with a black dotted (real) line.
Third, for all configurations, the EFG procedure significantly outperforms the greedy procedure. It means that using part of the sampling budget for exploration can dramatically improve the performance of the greedy procedure. Fourth, the PCS of the non-sample-optimal EA procedure decreases to zero as $k$ increases. However, using a small proportion of the sampling budget for greedy sampling can turn it into sample optimal and significantly improve its performance in solving large-scale R\&S problems. 

Furthermore, notice that for each of the three procedures, the obtained PCS under EM-CV is higher than that under SC-CV. It is because when the alternatives' means become dispersed, the problem becomes easier to solve.
\magenta{In fact,  when $k\rightarrow \infty$, the slippage configuration is least favorable for the greedy procedures. See the discussions following Theorems \ref{thm: Greedy_PCS_tightness} and \ref{thm: ExploreFirstGreedy_opt}. We provide additional numerical results to validate this result for a finite $k$ in \ref{subsec: least_favorable_EC}.}

\subsection{Additional Properties of the EFG Procedure}
\label{subsec: numerical2}

\subsubsection{Consistency.}
Now, we verify the consistency of the EFG procedure. In this experiment, we consider a fixed number of alternatives $k=2^{13}=8192$ and we vary the value of $c$ from $100$ to $800$ for the total sampling budget $B=ck$. For each $c$, we let $n_0=0.8c$. As a comparison, we also include the greedy procedure in this experiment.  We plot the estimated PCS of the EFG procedure and the greedy procedure against different values of $c$ under SC-CV and EM-CV in Figure \ref{numerical: EFGConsistency}.  \magenta{The PCS curves under EM-IV and EM-DV are similar to those shown in Figure  \ref{numerical: EFGConsistency} and therefore we put them in \ref{subsec: IVDV_EC} due to page limit. }

\begin{figure}[htbp]
  \centering
  \includegraphics[width=0.75\textwidth]{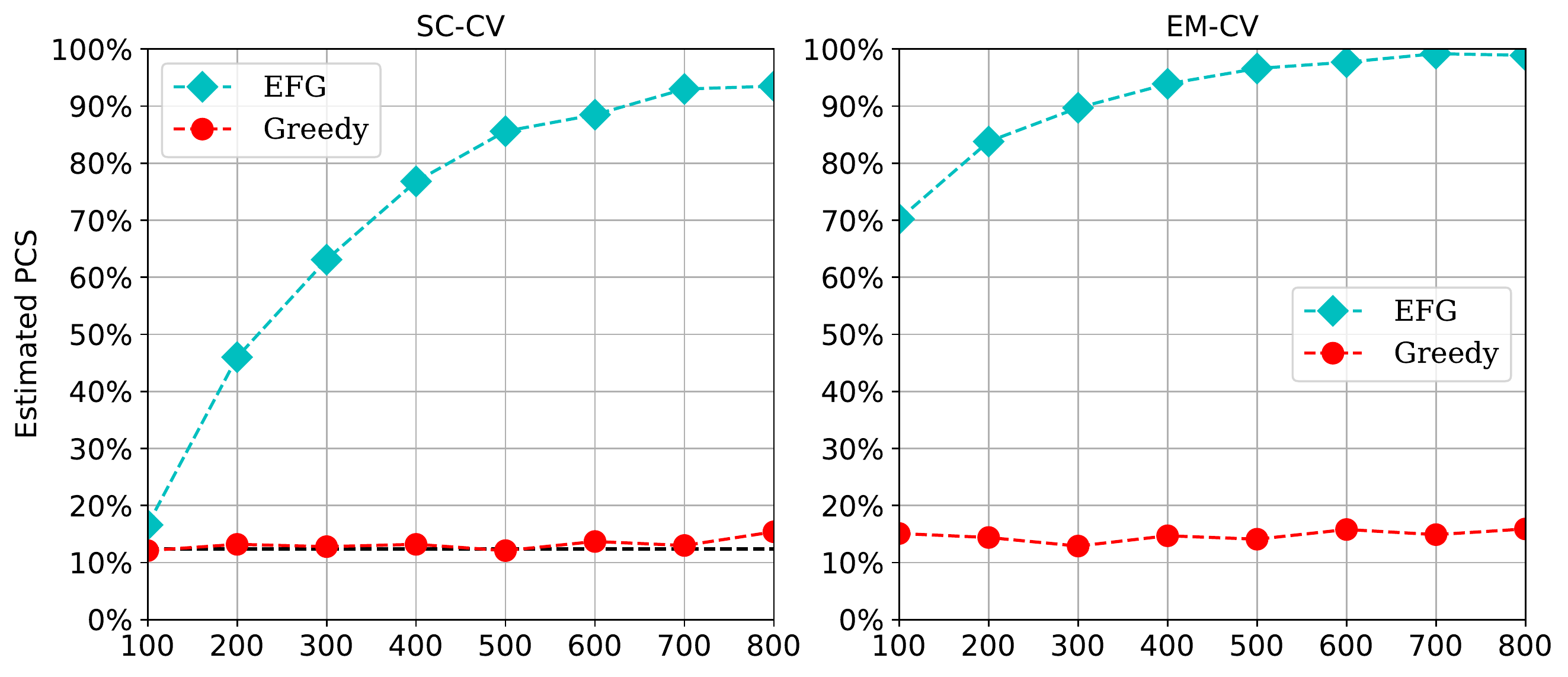}
  \caption{Estimated PCS of the EFG procedure and the greedy procedure for different values of $c$. $k=8192$. }
  \label{numerical: EFGConsistency}
\end{figure}

For this experiment, we have the following findings. First, the EFG procedure is consistent. As shown in Figure \ref{numerical: EFGConsistency}, under each configuration, as the total sampling budget increases from $100k$ to $800k$, the PCS of the EFG procedure rises to somewhere above 90\%. We expect the PCS to converge to 1 if we keep increasing the value of $c$. 
Second, in contrast to the EFG procedure, the greedy procedure is inconsistent. For all configurations, the PCS curves are almost flat. After $c=100$, they do not increase as the total sampling budget $B=ck$ grows. Therefore, they are unlikely to converge to 1 even if we keep boosting the $c$ value to infinity. In particular, under SC-CV, the PCS fluctuates around the theoretical upper bound, which is computed based on Corollary \ref{cor: inconsistency} and annotated with a black dotted line. It is because a total sampling budget $B=100k$ is enough for the greedy procedure to achieve a PCS close to the upper bound.

\subsubsection{Budget Allocation between Exploration and Exploitation.}
\label{subsec: allocationEEmain}
To understand the impact of $p=n_0/c$ on the PCS of the EFG procedure, we conduct the following experiment. We let $k=2^{13}=8192$, $B=200k$, and vary the value of $p$ from $0.005$ to $1$. Then, we plot the estimated PCS against different $p$ for SC-CV and EM-CV in Figure \ref{numerical: EFGBudgetAllocation} and summarize the main observations as follows. \magenta{The PCS curves under EM-IV and EM-DV are deferred to \ref{subsec: IVDV_EC}.} 
 
 \begin{figure}[htbp]
  \centering
  \includegraphics[width=0.75\textwidth]{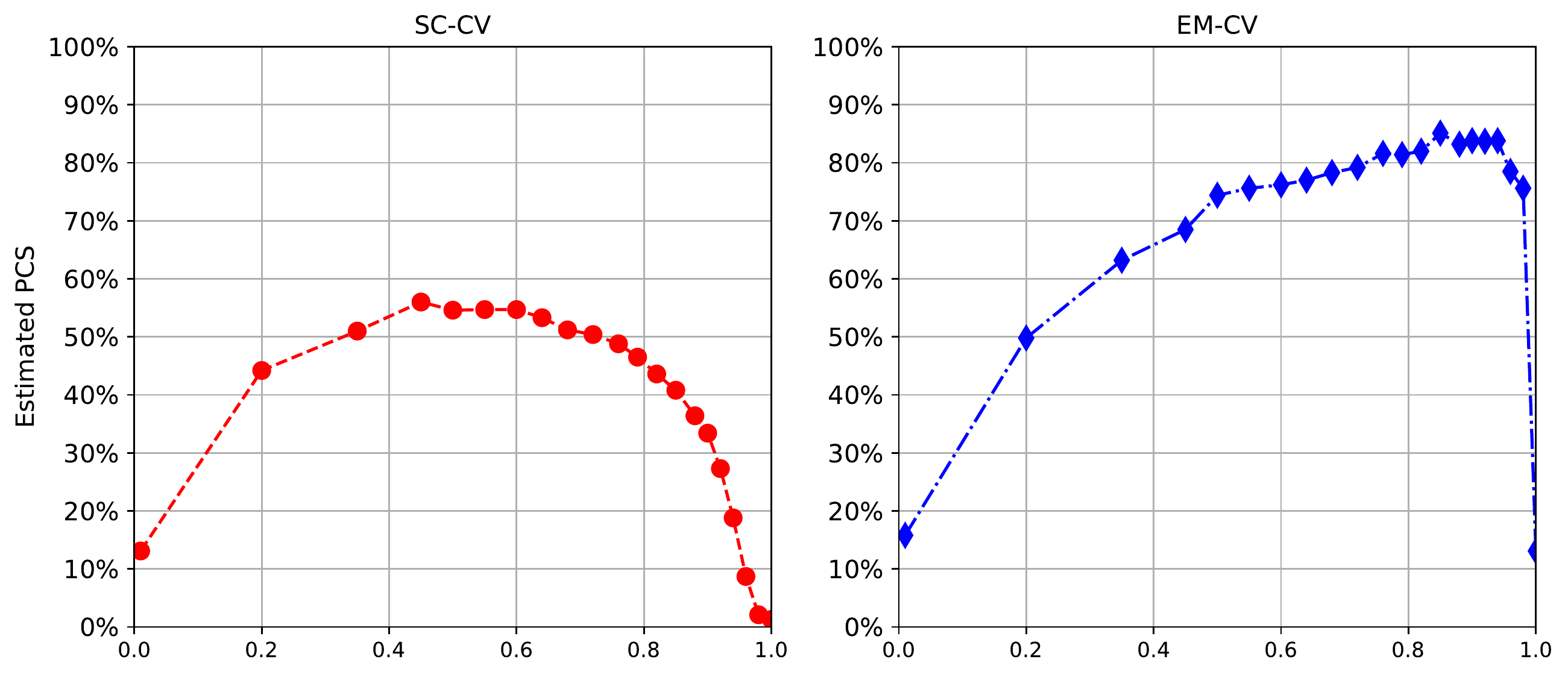}
  \caption{Estimated PCS of the EFG procedure for different values of $p=n_0/c$. $c=200$ and $k=8192$. }
  \label{numerical: EFGBudgetAllocation}
\end{figure}

It is interesting to see that for each configuration, the PCS curve has an inverted U-shape: as $p$ increases, the PCS increases first and then decreases. Specifically, first, under EM-CV, EM-IV, and EM-DV, the PCS starts to decrease only when $p$ exceeds 0.9. 
Interestingly, when $p$ rises to $1$ (the EFG procedure becomes the EA procedure), the PCS plummets to its lowest level, which should be 0 when the problem size is large enough, as shown in Figure \ref{numerical: GreedyEFGRateOptimality}. The above observations illustrate that for the EFG procedure, a small proportion of the total sampling budget for the greedy phase can be sufficient to guarantee the sample optimality and to achieve a near-optimal PCS.
Second, for SC-CV, the PCS starts to decrease when it exceeds 0.6 rather than 0.9. It may be because SC-CV is the most difficult among the four configurations and thus requires a larger greedy budget to guarantee the sample optimality than the other configurations. However, even so, letting $p$ = 0.7 can still achieve a near-optimal PCS, which again validates the above conclusion.
 
Another noteworthy finding from Figure \ref{numerical: EFGBudgetAllocation} is that, although the optimal $p$ that maximizes the PCS depends on the specific configuration, the PCS may be insensitive to the value of $p$ over a wide range of $p$. For configurations EM-CV, EM-IV, and EM-DV, this range is approximately (0.5, 0.95). In the range, the PCS does not vary much as the variation is almost within 10\%. Furthermore, for SC-CV, this range is approximately (0.5, 0.8). This fact implies that, in practice, the user does not need to decide $p$ carefully when applying the EFG procedure. 
\magenta{We also conduct this experiment for different settings of $k$
to study how the optimal $p$ and the a near-optimal range of $p$ changes as $k$ increases. See \ref{subsec: allocationEE} for the numerical results.} 


\subsection{Comparison between Sample-Optimal Fixed-Budget R\&S Procedures}
\label{subsec: numerical4}
In this subsection, we compare the EFG procedure, the EFG$^+$ procedure, and the existing fixed-budget sample-optimal R\&S procedures on a large-scale practical R\&S problem. 

\subsubsection{Procedures in Comparison.} 
We compare our EFG procedures with the  FBKT procedure and its enhanced version (i.e., the FBKT-Seeding procedure) of \cite{hong2022solving}, as well as the modified sequential halving (SH) procedure of \cite{zhao2022revisiting} which works as a fixed-budget median elimination procedure.
To the best of our knowledge, these procedures cover all the fixed-budget R\&S procedures currently known to be sample optimal. \magenta{Moreover, although the original SH procedure \citep{karnin2013almost} has not been proved to be sample optimal, we also include it in the comparison. We provide the necessary introduction and the implementation setting for each of them in \ref{subsec: EC_comparison}.} 

\subsubsection{The Throughput Maximization Problem.}
We test the procedures on the throughput maximization problem. The problem was first introduced by \cite{pichitlamken2006sequential}
and it has become a common testbed for large-scale R\&S research (e.g., see \citealt{luo2015fully, ni2017efficient} and \citealt{zhong2022knockout}).
A detailed description of the problem can be found in \ref{subsec: TP_EC}. The problem has two parameters $S_1, S_2 $, which decide the total number of alternatives and the means of the alternatives. 
We consider four problem instances with different combinations of $S_1$ and $S_2$, as summarized in Table \ref{tab: TP}. For the problem instances, we recognize that there may be multiple alternatives bearing the best mean, and we regard each of them as the best. 
Then, \magenta{$\gamma$ becomes the minimal difference between the best mean and all other strictly smaller ones.} We also summarize the information about $\gamma$, 
the number of best alternatives, and the number of good alternatives under $\delta=0.01$ in Table \ref{tab: TP}. 
\begin{table}[htb]
\centering
\begin{tabular}{c|c|c|c|c|c}
\hline
\hline
$\left(S_1,S_2\right)$ & $k$ & Highest mean & $\gamma$ & \# of best alt. & \# of good alt. ($\delta=0.01$) \\ \hline
$\left(20,20\right)$ & $3,249$ & 5.7761 & 0.0046 & 2 & 6 \\
$\left(30,30\right)$ & $11,774$ & 9.1882 & 0.0038 & 1 & 3 \\
$\left(45,30\right)$ & $27,434$ & 13.7823 & 0.0057 & 1 & 3 \\
$\left(45,45\right)$ & $41,624$ & 14.1499 & 0.0038 & 2 & 4 
\\\hline 
\hline
\end{tabular}
\caption{Information about the throughput maximization problem instances.} 
\label{tab: TP}
\end{table}

\subsubsection{Experiment Settings and Findings.}
In the comparison, we test the procedures on the problem instances reported in Table \ref{tab: TP}.  For each problem instance and every procedure, we estimate the PCS and the PGS with different total sampling budgets. Specifically, when estimating the PCS, we let $B = 50k$, $100k$, and $200k$. When estimating the PGS, we set the IZ parameter $\delta = 0.01$ and let $B=30k$, $50k$, and $100k$. \magenta{See \ref{subsec: setting_EC} for more implementation settings.} Then, we report the results for PCS in Table \ref{tab: TPPCS} and the results for PGS in Table \ref{tab: TPPGS}. Note that according to the budget allocation scheme introduced previously, the modified SH procedure requires a minimum total sampling budget of $81k$ to start the selection process. Therefore, the estimated PCS for $B = 50k$, and the estimated PGS for $B=30k$ and $B=50k$ are left empty. 

\begin{table}[htbp]
	\centering
	\begin{tabular}{c|ccc|ccc|ccc|ccc}
		\hline
		\hline
		$k$          & \multicolumn{3}{c|}{3,249}                    & \multicolumn{3}{c|}{11,774}                   & \multicolumn{3}{c|}{27,434}                   & \multicolumn{3}{c}{41,624}                    \\ \hline
		$c\, (B=ck)$ & 50      & 100     & 200       & 50        & 100     & 200       & 50      & 100     & 200     & 50      & 100     & 200       \\ \hline
		FBKT         & 0.31    & 0.40    & 0.51      & 0.24      & 0.31    & 0.36      & 0.30    & 0.34    & 0.40    & 0.37    & 0.40    & 0.50      \\
		FBKT-Seeding & 0.41    & 0.46    & 0.55      & 0.31      & 0.34    & 0.41      & 0.34    & 0.37    & 0.45    & 0.40    & 0.47    & 0.50      \\
		SH           &  0.57 &  0.63 &  0.69   &  0.44      &  0.52 &  0.60      &  0.54 &  0.61 &  0.70 &  0.62 &  0.69    &  {0.75} \\
		Modified SH  & -       & {0.61}  & 0.63      & -         & {0.41}  & {0.52}    & -       & {0.57}  & {0.57}  & -       & {0.65}  & {0.74}    \\
		EFG          & 0.37    & 0.38    & 0.43      & 0.24      & 0.24    & 0.30      & 0.22    & 0.24    & 0.29    & 0.26    & 0.30    & 0.35      \\
		EFG$^+$      & {0.51}  & 0.59    & {0.66}    & {0.35}    & 0.38    & {0.52}    & {0.37}  & 0.45    & 0.53    & {0.50}  & 0.63    & 0.68      \\
		\hline
		\hline
	\end{tabular}
	\caption{Estimated PCS on the throughput maximization problem instances. }
	\label{tab: TPPCS}
\end{table}

\begin{table}[htbp]
	\centering
	\begin{tabular}{c|ccc|ccc|ccc|ccc}
		\hline
		\hline
		$k$          & \multicolumn{3}{c|}{3,249}                    & \multicolumn{3}{c|}{11,774}                   & \multicolumn{3}{c|}{27,434}                   & \multicolumn{3}{c}{41,624}                             \\ \hline
		$c\, (B=ck)$ & 30      & 50        & 100       & 30        & 50      & 100       & 30        & 50     & 100     & 30      & 50      & 100     \\ \hline
		FBKT         & 0.78    & 0.87      & 0.94      & 0.67      & 0.70    & 0.80      & 0.69      & 0.73   & 0.82    & 0.64    & 0.65    & 0.71    \\
		FBKT-Seeding & 0.86    & 0.87      & 0.95      & 0.75      & 0.78    & 0.87      & 0.78      & 0.83   & 0.88    & 0.70    & 0.71    & 0.80    \\
		SH           &  0.98 &  0.99   &  {1.00} &  0.90    &  0.96 &  {0.98} &  0.96   &  0.99   &  {0.99}  &  0.92 &  0.96 & {0.96}  \\
		Modified SH  & -       & -         &  {1.00}    & -         & -       & {0.97}    & -         & -      &  {0.99}  & -       & -       &  {0.97}  \\
		EFG          & 0.77    & 0.83      & 0.89      & 0.53      & 0.60    & 0.67      & 0.54      & 0.60   & 0.65    & 0.48    & 0.52    & 0.58    \\
		EFG$^+$      & {0.94}  & {0.98}    &  {1.00}    & {0.82}    & {0.84}  & 0.94      & {0.79}    & {0.90} & 0.94    & {0.78}  & {0.89}  & 0.92    
		\\ \hline
		\hline
	\end{tabular}
	\caption{Estimated PGS  on the throughput maximization problem instances.}
	\label{tab: TPPGS}
\end{table}

From Table \ref{tab: TPPCS}, we have three findings regarding the PCS. First,
in each column of Table \ref{tab: TPPCS}, the PCS of the EFG$^+$ procedure is far higher than that of the EFG procedure. This demonstrates that although 20\% of the given sampling budget is sacrificed for the additional seeding phase, the enhanced exploration through the seeding phase is much more effective than the equal-allocation exploration for the EFG procedure.  Second, the EFG procedure performs weaker than the FBKT procedure and also the FBKT-Seeding procedure. 
However, the EFG$^+$ procedure performs better than the FBKT-Seeding procedure. This again verifies the effectiveness of the seeding phase for improving the EFG procedure. 
Third, when $k=3249$ and $k=11774$, the EFG$^+$ procedure is comparable to the modified SH procedure. 
\magenta{However, it is not as good as the original SH procedure, which performs the best among all the compared procedures on all problem instances, even though the original SH procedure has not been proved sample optimal \citep{karnin2013almost}.}

We have three additional findings from Table \ref{tab: TPPGS}. First, all the above findings about PCS also hold for PGS. Second, by comparing the results for $B=50k$ or $100k$ of Table \ref{tab: TPPCS} and Table \ref{tab: TPPGS}, we find that for all procedures, the obtained PGS is much larger than the corresponding PCS. This is because there are a larger number of good alternatives than the best alternative(s) in the problems as reported in Table \ref{tab: TP}, and the procedures are able to identify alternatives that are good enough but not the best. Third, even when the total sampling budget is relatively small (i.e., $B=30k$ or $50k$) and the modified SH procedure is not appliable, our EFG$^+$ procedure can still obtain a relatively high PGS for large-scale problems, e.g., $0.89$ for $k=41624$.  Furthermore, when the total sampling budget is relatively large, i.e., $B=100k$, our EFG$^+$ procedure is comparable to the modified SH procedure regarding the PGS.

To summarize, we conclude that the EFG$^+$ procedure performs significantly better than the EFG procedure, and it is comparable to the existing fixed-budget sample-optimal procedures in solving large-scale R\&S problems. \magenta{In addition,  it is also interesting to observe that the original SH procedure performs the best even though it is not provably rate optimal. However, we would like to emphasize that the ranking of these procedures depends on the problem being solved.  See \ref{subsec: more_compare_EC} for additional experiments where the EFG$^+$ procedure performs the best.}  



\subsection{\magenta{Summary of Additional Experiments}}
\label{subsec: additional}
\magenta{To expand our research scope, we perform additional experiments to further explore the behaviors of our procedures. Particularly, we center our attention on the EFG procedure, since it is more realistic than the greedy procedure and serves as a foundation for other variants of EFG procedures. We first investigate the budget allocation mechanism of the EFG procedure. Subsequently, we create problem configurations that challenge the structural assumptions (e.g., Assumption \ref{assu:asyreg}) outlined in the paper, and proceed to evaluate the procedure's PCS and PGS under these configurations. We leave the details and the results in \ref{subsec: allocationAlternativesMain} and \ref{subsec: additional_EC} and summarize the main findings here.}



\magenta{The EFG procedure intelligently concentrates the sampling budget on a small subset of highly promising alternatives during the greedy phase by leveraging the information obtained from the exploration phase. Specifically, in our experiments under the slippage configuration of means, we observe that for a relatively large $k$, only around 24\% of the non-best alternatives obtain observations in the greedy phase. 
As the means of the non-best alternatives become more dispersed, fewer non-best alternatives obtain observations in the greedy phase. In our experiments,
when the configuration is non-slippage and the means spread over a bounded set, the proportion of the non-best alternatives that obtain observations in the greedy phase converges to about 2\% as $k$ increases. Additionally, when the means of the non-best alternatives can be arbitrarily low and new alternatives are progressively worse as $k$ increases, the aforementioned proportion tends to 0, as only the few top alternatives close to the best alternative obtain observations in the greedy phase.} \orange{We emphasize that in this case,  the EA exploration phase of the EFG procedure is inefficient as a large proportion of the exploration budget will be wastefully allocated to clearly inferior alternatives. Luckily, this inefficiency can be alleviated by adding a seeding phase which is introduced in Section \ref{subsec: EFGPlus}, motivating the use of the EFG$^+$ procedure. }

Furthermore, numerical experiments show that the PCS and PGS of the EFG procedure depend on the specific problem configurations. Firstly, when the mean gap $\gamma \rightarrow 0$ as $k$ grows, the PCS of the EFG procedure with a total sampling budget $B=ck$ could decrease to 0. This indicates that Assumption \ref{assu:asyreg} might be necessary for the EFG procedure to achieve the sample optimality regarding the PCS defined in Definition \ref{def: rate_optimality}.  Secondly, regarding PGS, the EFG procedure also performs significantly better than the EA procedure when the number of good alternatives $g(k)$ is unbounded (e.g., growing at the order $O(k)$ as $k \rightarrow \infty$), showing the value of the greedy phase in improving the sample efficiency. Thirdly, there is an interesting phenomenon regarding PGS: when $g(k)$ grows at the order $O(\sqrt{k})$, the PGS even goes to 1 as $k$ increases; however, when it grows faster, i.e., at the order $O(k)$, the PGS does not exhibit a similar trend. 
\red{This discrepancy may be because, in the former case, the good alternatives constitute only a negligible proportion of the $k$ alternatives and they collectively may only utilize a tiny budget relative to $k$ particularly when $k$ is large; however, in the latter case, the good alternatives may occupy too much of the sampling budget, making it less likely for all non-best alternatives to finish their boundary-crossing processes. Such discrepancy indicates that the PGS of the EFG procedure hinges on the order in which the number of alternatives grows as $k$ increases.}

\section{Concluding Remarks}
\label{sec: conclusion}
This paper is driven by a surprising discovery that the simple greedy procedure exhibits superior performance in solving large-scale R\&S problems. Through a boundary-crossing perspective, we prove that the greedy procedure is sample optimal. More precisely, we provide a non-zero lower bound for the asymptotic PCS as the number of alternatives grows to infinity and show that the lower bound is tight under the slippage configuration of means with a common variance. In addition, we derive an analogous result for the PGS of the greedy procedure, when the number of good alternatives is finite and bounded. Then, to overcome the inconsistency of the greedy procedure, we add an exploration phase to it and design the EFG and $\mbox{EFG}^+$ procedures. Again based on the boundary-crossing perspective, we show that the EFG and $\mbox{EFG}^+$ procedures are sample optimal and consistent. 
\magenta{We further extend the $\mbox{EFG}^+$ procedure to the $\mbox{EFG}^{++}$ procedure to take advantage of the widely available parallel computing resources.}
We verify our findings by the numerical experiments and show that the greedy procedures may perform comparably well to the existing sample-optimal R\&S procedures. From our viewpoint, this research suggests that a new mindset may be necessary for designing large-scale R\&S procedures.


To close this paper, we summarize several potential research directions for future investigation. From the theoretical side, two primary research questions arise. First, the current PGS analysis is conducted under the assumption that the number of good alternatives is finite and bounded. 
\red{It is interesting and important to study how to extend the analysis to more general settings to provide a rigorous characterization of the greedy procedures' PGS when the number of good alternatives is unbounded.}
Second, apart from the PCS and PGS, the EOC is another commonly used measure of the effectiveness of an R\&S procedure. It is interesting to ask whether the greedy procedures remain sample optimal in terms of the EOC.

\magenta{From the practical side, we think that three interesting research issues emerge. Firstly, as mentioned in Section \ref{sec: intro}, in this paper, our attention is on the pure R\&S procedures, ignoring the possible structural information of ``nearby" alternatives. However, for many practical large-scale problems, the ``nearby"  structure may be too informative to be overlooked. Procedures exploiting such structural information, e.g., by using Gaussian process regression to capture the similarity between the means of neighboring alternatives \citep{semelhago2019gaussian, semelhago2021rapid}, may be much more efficient than our general-purpose large-scale R\&S procedures. Therefore, investigating how the family of greedy procedures can leverage the structural information holds significant promise. Secondly, as shown above, adding a seeding phase exhibits a great improvement on the performance for the EFG procedure. Then, how to allocate the seeding budget more strategically than an equal-allocation scheme to achieve further improvement is worthy of further investigation. Lastly, considering the performance enhancements associated with common random numbers (CRN) demonstrated by \cite{chick2001newCRN} and \cite{hong2022solving}, it becomes interesting to investigate the integration of CRN into the greedy procedures and whether a similar benefit can be achieved. In the EFG procedure, a straightforward implementation involves incorporating CRN exclusively during the equal-allocation exploration phase, while executing the subsequent greedy phase without CRN.  By doing so, similar performance enhancement may be anticipated.}

\bibliographystyle{informs2014} 
\bibliography{ref.bib} 


\ECSwitch
%
%
\ECHead{\magenta{E-Companion to ``The (Surprising) Sample Optimality of Greedy
Procedures for Large-Scale Ranking and Selection''}}
%
%
%
%
%


\section{Technical Supplements for Section \ref{sec: GreedyOpt}}
\vspace*{6pt}
\subsection{Proof of Lemma \ref{lem:argmin}}
\label{subsec: proofargmin}
We prove the lemma by contradiction. For any particular sample path $\{\bar{Z}(n;\omega),n=1,2,\ldots\}$ of the running-average process, we let $M(\omega)=\argmin_{n\geq 1}\bar{Z}(n;\omega)$ denote the location of the minima. Suppose that the lemma is violated. Then, there must exist a collection of sample paths $\Omega=\{\omega: \bar{Z}(n;\omega),n=1,2,\ldots\}$ such that 
\begin{align}
\label{proof:lemma1}
M(\omega) = \infty \mbox{ for any } \omega\in\Omega, \mbox{ and } \Pr\{\omega\in\Omega\}>0.
\end{align}
Further, we let $M_s(\omega)=\argmin_{1\leq n\leq s}\bar{Z}(n;\omega)$. If $M_s(\omega)$ has tied values, we set it as the smallest one. It is clear that, for any $\omega\in\Omega$,  $\{M_s(\omega), s=1,2,\ldots\}$ is a monotonically non-decreasing sequence. In other words, for any $\omega\in\Omega$,
\begin{align*}
M_s(\omega)\uparrow M(\omega)=\infty, \mbox{ as } s\uparrow \infty, \mbox{ and } \bar{Z}(M_1(\omega);\omega)\geq \bar{Z}(M_2(\omega);\omega)\geq \dots\geq \bar{Z}(M_s(\omega);\omega)\geq \dots.
\end{align*}
By the strong law of large number, $\bar{Z}(M_s(\omega);\omega)\downarrow 0$ as $s\uparrow \infty$. Therefore, we have that $\bar{Z}(M_s(\omega);\omega)=\min_{1\leq n\leq s}\bar{Z}(n;\omega)\geq \lim\limits_{s\uparrow\infty}\bar{Z}(M_s(\omega);\omega)=0$ for any $s\geq 1$, or equivalently, $\min_{n\geq 1}\bar{Z}(n;\omega)\geq 0$. Then, in the consideration of Lemma \ref{lem:normalbc}, we can derive
\[
\Pr\{\omega\in\Omega\}\leq\Pr\left\{\min_{n\geq 1}\bar{Z}(n;\omega)\geq 0\right\}=\Pr\{N(0)=\infty\}=\exp\left(-\sum_{n=1}^{\infty} n^{-1}\mathrm{P}(\bar{X}(n)\leq 0)\right)=0.
\]
Clearly, this contradicts with Equation \eqref{proof:lemma1}. The proof is completed.\hfill\Halmos
\vspace*{6pt}
\subsection{Proof of Lemma \ref{lem:normalbc}}
\label{subsec: proofNormalBC}
In this proof, we only prove that $C(x) = \exp \left(\sum\limits_{n=1}^{\infty} \frac{1}{n} \Phi\left(-\sqrt{n}x \right)\right)$ is continuous on $x\in(0,\infty)$. The other results are direct consequences of Corollaries 8.39 and 8.44 in \cite{itemSiegmund1985}.
Let 
\[
f_m(x) = \sum_{n=1}^{m} \frac{1}{n} \Phi\left(-\sqrt{n}x \right) \mbox{ and } f(x) =\lim_{m\to\infty}f_m(x) =\sum_{n=1}^{\infty} \frac{1}{n} \Phi\left(-\sqrt{n}x \right).
\]
Given any constant $\delta>0$, it is clear that
\begin{eqnarray*}
\left|  f_m(x) - f(x) \right|= \sum_{n=m+1}^{\infty} \frac{1}{n} \Phi\left(-\sqrt{n} x \right) \leq\sum_{n=m+1}^{\infty} \frac{1}{n} \Phi\left(-\sqrt{n} \theta \right), \forall x\in[\theta,\infty)
\end{eqnarray*}
Therefore, $f_m(x)$ uniformly converges to $f(x)$ on $x\in[\theta,\infty)$. Further, by Theorem 7.12 of \cite{itemRudin1976}, $C(x)$ is continuous on $[\theta,\infty)$. Notice that the continuity holds for any arbitrarily small $\theta>0$, then the conclusion is drawn.\hfill \Halmos

\vspace*{6pt}
\subsection{Proof of Theorem \ref{thm: Greedy_PCS_opt}}
\label{subsec: proofGreedyOpt}
As stated in Equation \eqref{eqn:pcs2}, the PCS can be formulated as follows,
\[
\mbox{PCS}  \geq  \Pr\left\{ck \geq \argmin_{n\geq 1}\bar X_1(n)+\sum_{i=2}^{k}N_i\left(\min_{n \geq 1}\bar X_1(n)\right) \right\}.
\]
Then, to complete the proof, it suffices to build several useful properties on the sum of the boundary-crossing times of non-best alternatives, $\sum_{i=2}^{k}N_i\left(\min_{n \geq 1}\bar X_1(n)\right)$.

To begin with, we first set up some necessary notations. For any given $c>C(\gamma/\overline{\sigma})$, we can arbitrarily select a constant $\epsilon$ with $0<2\epsilon<c-C\left({\gamma}/{\overline\sigma}\right)$. Since $C(\cdot)$ is a continuous and monotonically decreasing function by Lemma \ref{lem:normalbc}, we can always choose a unique positive constant $\gamma_\epsilon^-\in(0,\gamma)$ such that
\begin{eqnarray}\label{thm1: C1}
C\left(\frac{\gamma-\gamma_\epsilon^-}{\overline\sigma}\right) =c-2\epsilon.
\end{eqnarray}

Suppose that the condition $A_\epsilon = \{\min_{n\geq 1} \bar X_{1}(n) > \mu_1 -\gamma_\epsilon^-\}$ holds. Because $\mu_1-\mu_i\geq \gamma$ and  $\sigma_i^2\leq \overline\sigma^2$,
we have from Equation \eqref{eqn: Nbound} that
\[
\sum_{i=2}^{k} N_i\left(\min_{n\geq 1}\bar X_1(n)\right)\leq \sum_{i=2}^k N_i\left(\mu_1 -\gamma_\epsilon^-\right)\leq \sum_{i=2}^k \tilde{N}_i \left(\frac{\gamma-\gamma_\epsilon^-}{\overline\sigma}\right).
\]
Notice that
$\tilde{N}_i \left(\frac{\gamma-\gamma_\epsilon^-}{\overline\sigma}\right), i=2,3,\ldots,k,$ are independent and identically distributed, we further have, as $k\to\infty$,
\begin{eqnarray}\label{eqn:proofThereom1_1}
\frac{1}{k}\sum_{i=2}^{k} N_i\left(\min_{n\geq 1}\bar X_1(n)\right)\leq  \frac{1}{k}\sum_{i=2}^k  \tilde{N}_i \left(\frac{\gamma-\gamma_\epsilon^-}{\overline\sigma}\right)
\stackrel{\mbox{a.s.}}{\longrightarrow}\E\left[\tilde{N}_i\left(\frac{\gamma-\gamma_\epsilon^-}{\overline\sigma}\right)\right]= C\left(\frac{\gamma-\gamma_\epsilon^-}{\overline\sigma}\right)=c-2\epsilon, \quad\quad \quad
\end{eqnarray}
according to the strong law of large numbers. The last equality above holds by the choice of $\gamma_\epsilon^-$ given in Equation \eqref{thm1: C1}. We let
\[
\Omega_{\epsilon}=\left\{\frac{1}{k}\sum_{i=2}^{k} N_i\left(\min_{n\geq 1}\bar X_1(n)\right)- (c-2\epsilon)\leq  \epsilon \right\}.
\]
Then by the definition of the strong law of almost sure convergence and Equation \eqref{eqn:proofThereom1_1},
\begin{equation}\label{thm1: PCS_slln}
\begin{aligned}
\lim_{k\to\infty} \Pr\{\Omega_{\epsilon}|A_\epsilon\}&\geq \lim_{k\to\infty} \Pr\left\{\frac{1}{k}\sum_{i=2}^k  \tilde{N}_i \left(\frac{\gamma-\gamma_\epsilon^-}{\overline\sigma}\right)-(c-2\epsilon)\leq \epsilon\big|A_\epsilon\right\} \\
&=  \lim_{k\to\infty} \Pr\left\{\frac{1}{k}\sum_{i=2}^k  \tilde{N}_i \left(\frac{\gamma-\gamma_\epsilon^-}{\overline\sigma}\right)-(c-2\epsilon)\leq \epsilon\right\}\\
&=1.
\end{aligned}
\end{equation}

Now we are ready to study the PCS formulated in Equation \eqref{eqn:pcs2}. To be specific,
\begin{eqnarray}
\label{thm1: CScond}
\notag \mbox{PCS} &  \geq &  \Pr\left\{ck \geq \argmin_{n\geq 1}\bar X_1(n)+\sum_{i=2}^{k}N_i\left(\min_{n \geq 1}\bar X_1(n)\right) \right\}\\
\notag & \geq & \Pr\left\{ \left\{ck \geq \argmin_{n\geq 1}\bar X_1(n)+\sum_{i=2}^{k}N_i\left(\min_{n \geq 1}\bar X_1(n)\right)\right\} \cap A_\epsilon\cap \Omega_{\epsilon} \right\} \\
\notag & \geq & \Pr\left\{\left\{c k \geq \argmin_{n\geq 1} \bar X_1(n) +(c-2\epsilon)k+\epsilon k\right\} \cap A_\epsilon\cap \Omega_{\epsilon} \right\}\quad \mbox{(by definition of $\Omega_\epsilon$)} \\
\notag &=&  \Pr\left\{ \left\{\epsilon k \geq \argmin_{n\geq 1} \bar X_1(n)\right\} \cap A_\epsilon\cap \Omega_{\epsilon} \right\}\\
\notag & \geq & \Pr\{A_\epsilon\cap \Omega_\epsilon\}-\Pr\left\{\argmin_{n\geq 1} \bar X_1(n) > \epsilon k\right\}\\
\notag & = & \Pr\{\Omega_{\epsilon}|A_\epsilon\}\Pr\{A_\epsilon\}-\Pr\left\{\argmin_{n\geq 1} \bar X_1(n) > \epsilon k\right\},
\end{eqnarray}
where the last inequality arises due to \magenta{the fact that $\Pr\left\{A\cap B\right\}=\Pr\left\{B\right\}-\Pr\left\{A^c\cap B\right\}\geq \Pr\left\{B\right\}-\Pr\left\{A^c\right\}$}. From Equation \eqref{thm1: PCS_slln}, letting $k\to\infty$ in above gives rise to
\begin{equation}\label{thm1: PCS4}
\begin{aligned}
\liminf_{k\to\infty} \mbox{PCS}&\geq\limsup_{k\to\infty}  \Pr\{\Omega_{\epsilon}|A_\epsilon\}\Pr\{A_\epsilon\} -\Pr\left\{\argmin_{n\geq 1} \bar X_1(n) =\infty\right\}\\
&= \Pr\left\{\min_{n\geq 1} \bar X_{1}(n) >  \mu_1 -\gamma_\epsilon^-\right\}-\Pr\left\{\argmin_{n\geq 1} \bar X_1(n) =\infty\right\}.
\end{aligned}
\end{equation}
Notice that the above statement holds for any $\epsilon$ sufficiently close to zero. Besides, it is clear that $C(\cdot)$ is continuous, and therefore,
$
  \lim\limits_{\epsilon \to 0} \gamma_\epsilon^- = \gamma_0$
where $\gamma_0$ is a positive constant given by $C\left(\frac{\gamma-\gamma_0}{\bar \sigma}\right) =c$. Then, by Lemma \ref{lem:argmin}, letting $\epsilon\to 0$ yields
\begin{eqnarray*}
\liminf_{k \to \infty} \mbox{PCS} &\geq& \Pr\left\{\min_{n\geq 1} \bar X_{1}(n) >  \mu_1 -\gamma_0\right\}- \Pr\left\{\argmin_{n\geq 1} \bar X_1(n) =\infty\right\}\\
\notag&=& \Pr\left\{\min_{n\geq 1} \bar X_{1}(n) >  \mu_1 -\gamma_0\right\}\\
\notag &=&\left[C\left(\frac{\gamma_0}{\sigma_1}\right)\right]^{-1},
\end{eqnarray*}
This concludes the proof of the theorem. \hfill\Halmos

\vspace*{6pt}
\subsection{Proof of Theorem \ref{thm: Greedy_PCS_tightness}}
\label{subsec: proofPCStightness}
As stated in Section \ref{subsec: tightness}, it sufficient to show the following statement
\begin{align}\label{eqn: PICS}
\liminf_{k\to\infty}\, \mbox{PICS} \geq \Pr\left\{\min_{n\ge 1} \bar X_1(n) < \mu_1-\gamma_0 \right\}=\Pr\{\exists\, n \geq 1 \text{ s.t. } \bar X_1(n) < \mu_1 - \gamma_0 \}.
\end{align}

By Lemma \ref{lem:normalbc}, $C(\cdot)$ is a strictly decreasing and continuous function in $(0,\infty)$. Then, for any constant $\epsilon$ with $0<2\epsilon<c-C\left({\gamma}/{\sigma}\right)$, we can always choose  $0<\gamma_\epsilon^-, \gamma_\epsilon^+<\gamma$ such that
  \begin{eqnarray}
     \label{cor1: boundaries}
    C\left(\frac{\gamma-\gamma_\epsilon^-}{\sigma}\right) =c-2\epsilon \text{ and }
    C\left(\frac{\gamma-\gamma_\epsilon^+}{\sigma}\right) =c+2\epsilon.
  \end{eqnarray}
Further, it is easy to check that $\gamma_\epsilon^-<\gamma_0<\gamma_\epsilon^+$ and $\gamma_\epsilon^-, \gamma_\epsilon^+\to \gamma_0$ as $\epsilon\to 0$. 

Let $m_\epsilon=\inf\{n\geq 1: \bar{X}_1(n)<\mu_1-\gamma_\epsilon^+\}$. Suppose for now the condition $B_\epsilon=\{m_\epsilon<\infty, \min_{1 \leq n\leq m_\epsilon-1}\bar{X}_1(n) \geq \mu_1-\gamma_\epsilon^-\}$ is forced. When $m_\epsilon=1$, we set $\bar{X}_1(m_\epsilon-1)=\bar{X}_1(0)=\infty$. From the boundary-crossing perspective, the greedy procedure would produce an incorrect selection whenever the sample mean of alternative 1 falls below $\mu_1-\gamma_\epsilon^+$ but the remaining sampling budget is not enough to let alternative 1 become the sample best again. In other words, we may write
\begin{eqnarray}\label{falseselection}
\mbox{PICS}\geq\Pr\left\{\sum_{i=2}^k N_i\left(\min_{1 \leq n\leq m_\epsilon-1}\bar{X}_1(n)\right) +m_\epsilon\leq B < \sum_{i=2}^k N_i\left(\min_{1 \leq n\leq m_\epsilon}\bar{X}_1(n)\right)+m_\epsilon\right\}.
\end{eqnarray}
Under the SC-CV and condition $B_\epsilon$, the sum of boundary-crossing times satisfies that, as $k\to\infty$,
\begin{eqnarray*}
  \lefteqn{\frac{1}{k-1}\sum_{i=2}^k N_i\left(\min_{1\leq n\leq m_\epsilon-1} \bar{X}_1(n)\right)}\\
  &\leq& \frac{1}{k-1}\sum_{i=2}^k N_i\left(\mu_1-\gamma_\epsilon^-\right)= \frac{1}{k-1}\sum_{i=2}^k \tilde N_i\left(\frac{\gamma-\gamma_\epsilon^-}{\sigma}\right) \to C\left(\frac{\gamma-\gamma_\epsilon^-}{\sigma}\right)=c-2\epsilon,\mbox{ a.s. } 
  \end{eqnarray*}
  by the strong law of large numbers and Equation \eqref{cor1: boundaries}. Likewise, as $k\to\infty$,
\begin{eqnarray*}
\lefteqn{\frac{1}{k-1}\sum_{i=2}^k N_i\left(\min_{1\leq n\leq m_\epsilon} \bar{X}_1(n)\right)}\\
&\geq &\frac{1}{k-1}\sum_{i=2}^k N_i\left(\mu_1-\gamma_\epsilon^+\right) = \frac{1}{k-1}\sum_{i=2}^k \tilde N_i\left(\frac{\gamma-\gamma_\epsilon^+}{\sigma}\right)\to C\left(\frac{\gamma-\gamma_\epsilon^+}{\sigma}\right)=c+2\epsilon,\mbox{ a.s. }
\end{eqnarray*}
Therefore, given any $\epsilon>0$,  the collections of sample paths denoted by
\begin{equation}\label{eqn: Omega}
\begin{aligned}
\Omega_{\epsilon}^-=&\left\{\frac{1}{k-1}\sum_{i=2}^k N_i\left(\min_{1\leq n\leq m_\epsilon-1} \bar{X}_1(n)\right)-(c-2\epsilon) \leq \epsilon \right\}\\
\Omega_{\epsilon}^+=&\left\{\frac{1}{k-1}\sum_{i=2}^k N_i\left(\min_{1\leq n\leq m_\epsilon} \bar{X}_1(n)\right)-(c+2\epsilon)   \geq -\epsilon \right\}
\end{aligned}
\end{equation}
satisfy that
\begin{eqnarray}\label{cor1: PCS_slln}
\lim_{k\to\infty} \Pr\{\Omega_\epsilon^+|B_\epsilon\} = \lim_{k\to\infty} \Pr\{\Omega_\epsilon^-|B_\epsilon\} = 1.
\end{eqnarray}

Now we are ready to study the PICS in Equation \eqref{falseselection}. Plugging Equation \eqref{falseselection} and Equation \eqref{eqn: Omega} into Equation \eqref{falseselection} gives
\begin{eqnarray*}
\notag \mbox{PICS}
&\geq& \Pr\left\{\left\{\sum_{i=2}^k N_i\left(\min_{1 \leq n\leq m_\epsilon-1}\bar{X}_1(n)\right) +m_\epsilon\leq B < \sum_{i=2}^k N_i\left(\min_{1 \leq n\leq m_\epsilon}\bar{X}_1(n)\right)+m_\epsilon\right\}\bigcap\Omega_\epsilon^+\bigcap\Omega_\epsilon^-\bigcap B_\epsilon\right\}\\
\notag &\geq & \Pr\left\{\left\{(k-1)(c-\epsilon) +m_\epsilon\leq ck < (k-1)(c+\epsilon)+m_\epsilon\right\}\cap\Omega_\epsilon^+\cap\Omega_\epsilon^-\cap B_\epsilon\right\}\\
\notag &= & \Pr\left\{\left\{c-(k-1)\epsilon<m_\epsilon\leq c+(k-1)\epsilon \right\}\cap\Omega_\epsilon^+\cap\Omega_\epsilon^-\cap B_\epsilon\right\}\\
\notag &\geq& \Pr\left\{\left\{c-(k-1)\epsilon<m_\epsilon\leq c+(k-1)\epsilon \right\}\cap B_\epsilon\right\}-\Pr\left\{(\Omega_\epsilon^+)^c\cap B_\epsilon\right\}-\Pr\left\{(\Omega_\epsilon^-)^c\cap B_\epsilon\right\}\\
\notag &=& \Pr\left\{\left\{c-(k-1)\epsilon<m_\epsilon\leq c+(k-1)\epsilon \right\}\cap B_\epsilon\right\}-\Pr\left\{(\Omega_\epsilon^+)^c| B_\epsilon\right\}\Pr\{B_\epsilon\}-\Pr\left\{(\Omega_\epsilon^-)^c| B_\epsilon\right\}\Pr\{B_\epsilon\},
\end{eqnarray*}
where the last inequality holds \magenta{due to the fact that}
\magenta{
\begin{eqnarray*}
    \Pr\left\{A \cap B\cap C \cap D\right\} & = & \Pr\left\{A \cap D\right\} - \Pr\left\{A \cap (B\cap C)^c \cap D\right\} \\
     & \geq & \Pr\left\{A \cap D\right\} - \Pr\left\{ (B\cap C)^c \cap D\right\} \\
     & = & \Pr\left\{A \cap D\right\} - \Pr\left\{(B^c \cap D) \cup (C^c \cap D)\right\} \\
     & \geq &  \Pr\left\{A \cap D\right\} -\Pr\left\{B^c \cap D\right\} - \Pr\left\{C^c \cap D\right\}.
\end{eqnarray*}}
Taking $k\to\infty$ on both sides, we can derive
\begin{eqnarray}\label{proof: PICS}
\liminf_{k\to\infty}\ \mbox{PICS} \geq \Pr\left\{m_\epsilon<\infty, \min_{1 \leq n\leq m_\epsilon-1}\bar{X}_1(n) \geq \mu_1-\gamma_\epsilon^-\right\}
\end{eqnarray}
from Equation \eqref{cor1: PCS_slln}. The inequality above holds for any sufficiently small $\epsilon>0$. Recall that $\gamma_\epsilon^-\to\gamma_0$ as $\epsilon\to 0$, and accordingly, $m_\epsilon\to m=\inf\{n\geq 1:\bar{X}_1(n)<\mu_1-\gamma_0\}$ almost surely. Letting $\epsilon\to 0$ on the right hand side of the inequality above, we readily obtain that
\begin{eqnarray*}
\liminf_{k\to\infty}\ \mbox{PICS} &\geq& \limsup_{\epsilon\to 0}\Pr\left\{m_\epsilon<\infty, \min_{1 \leq n\leq m_\epsilon-1}\bar{X}_1(n) \geq \mu_1-\gamma_\epsilon^-\right\}\\
&=&\Pr\left\{m<\infty, \min_{1 \leq n\leq m-1}\bar{X}_1(n) \geq \mu_1-\gamma_0\right\} =\Pr\{\exists\, n \geq 1 \text{ s.t. } \bar X_1(n) < \mu_1 - \gamma_0 \}.
\end{eqnarray*}
The proof is completed.\hfill \Halmos

\vspace*{6pt}

\vspace*{6pt}
\subsection{Proof of Proposition \ref{prop: Greedy_PGS_opt}}
\label{subsec: proofgreedyPGS}

To analyze  the greedy procedure’s PGS from the boundary-crossing perspective, we regard the minimum of the running maximum of the good alternatives, i.e., $\min_{r\ge 1} \bar Y_\delta(r)$, as the boundary. Then, we prove the conclusion based on the PGS statement in Equation \eqref{eqn:pgs}, i.e.,
\begin{eqnarray*}
{\rm PGS}&\ge& \Pr\left\{\sum_{i\in{\cal N}} N_i\left(\min_{r\ge 1} \bar Y_\delta(r)\right)+ \argmin_{r\ge 1} \bar Y_\delta(r)\leq B \right\}.
\end{eqnarray*}
Notice that Equation \eqref{eqn:pgs} exhibits a similar form to the PCS statement in Equation \eqref{eqn:pcs2}. Following almost the same analysis of Theorem \ref{thm: Greedy_PCS_opt}, 
we obtain an analogous result of Equation \eqref{thm1: PCS4} as 
\begin{eqnarray}\label{proof: PGS1}
\liminf_{k\to\infty}\, {\rm PGS} &\ge& \limsup_{k\to\infty}\,\Pr\left\{\min_{r\ge 1} \bar Y_\delta(r) >  \mu_1 -\delta_\epsilon^-\right\}-\liminf_{k\to\infty}\,\Pr\left\{\argmin_{r\ge 1} \bar Y_\delta(r)=\infty\right\} \quad\quad
\end{eqnarray}
where $0<\epsilon<(c-C(\delta/\overline{\sigma}))/2$ and $\delta_\epsilon^-$ is chosen to satisfy
$C\left(\frac{\delta-\delta_\epsilon^-}{\overline\sigma}\right)=c-2\epsilon$.
By the continuity of $C(\cdot)$ stated in Lemma \ref{lem:normalbc},  
 $\delta_\epsilon^-\to\delta_0, \mbox{ as } \epsilon\to 0$.  
Letting $\epsilon\to 0$ in Equation \eqref{proof: PGS1} yields
\begin{eqnarray}\label{proof: PGS2}
\liminf_{k\to\infty}\, {\rm PGS} &\ge& \limsup_{k\to\infty}\,\Pr\left\{\min_{r\ge 1} \bar Y_\delta(r) >  \mu_1 -\delta_0\right\}-\liminf_{k\to\infty}\,\Pr\left\{\argmin_{r\ge 1} \bar Y_\delta(r)=\infty\right\}.  \quad\quad
\end{eqnarray}

Notice that Equation \eqref{eqn:pgs2} provides a reasonable upper bound for $\argmin_{r\ge 1} \bar Y_\delta(r)$. Specifically, 
\begin{equation*}
\argmin_{r\ge 1} \bar Y_\delta(r) \le \sum_{i\in{\cal G}} \argmin_{n\ge 1}\bar X_i(n),
\end{equation*}
Then, we can conclude that $\Pr\left\{\argmin_{r\ge 1} \bar Y_\delta(r)=\infty\right\}\le \sum_{i\in{\cal G}}\Pr\left\{\argmin_{n\ge 1}\bar X_i(n)=\infty\right\}$,  which amounts to zero by Lemma \ref{lem:argmin} and the assumption that $\limsup_{k\to\infty}|{\cal G}|<\infty$. Meanwhile, because $\min_{r\ge 1} \bar Y_\delta(r)=\max_{i\in\mathcal{G}}\min_{n\geq 1}\bar{X}_i(n)$, Equation \eqref{proof: PGS2} can be further written as
\begin{eqnarray*}
\liminf_{k\to\infty} {\rm PGS} \ge \limsup_{k\to\infty}\Pr\left\{\min_{r\ge 1} \bar Y_\delta(r) >  \mu_1 -\delta_0\right\}
=\limsup_{k\to\infty}\Pr\left\{\exists\,i\in\mathcal{G}:\min_{n\geq 1}\bar{X}_i(n)>\mu_1-\delta_0\right\}
\end{eqnarray*}
It is clear that the right-hand-side probability is at least $\Pr\left\{\min_{n\geq 1}\bar{X}_1(n) > \mu_1-\delta_0\right\}$. The proof is completed.\hfill\Halmos

\vspace*{6pt}
\section{Technical Supplements for Section \ref{sec: ExploreFirst}}
\vspace*{6pt}
\subsection{Proof of Lemma \ref{lem:argmin_EFG}}
\label{subsec: proofargminEFG}
Lemma \ref{lem:argmin_EFG} is a generalization of Lemma \ref{lem:argmin}. Similar to the proof of Lemma \ref{lem:argmin},  we prove Lemma \ref{lem:argmin_EFG} by contradiction.
For any $n_0\geq 1$, we let $M(\omega; n_0)=\argmin_{n\geq n_0}\bar{Z}(n;\omega)$ denote the location of the minima of the running-average process $\{\bar{Z}(n;\omega),n=n_0,n_0+1,\ldots\}$. Suppose that the lemma is violated. Then, there must exist a collection of sample paths $\Omega(n_0)=\{\omega: \bar{Z}(n;\omega),n=n_0,n_0+1,\ldots\}$ such that 
\begin{align}
\label{proof:lemma3}
M(\omega;n_0)=\infty \mbox{ for any } \omega\in\Omega(n_0), \mbox{ and } \Pr\{\omega\in\Omega(n_0)\}>0.
\end{align}
By repeating similar arguments in Section \ref{subsec: proofargmin} for the case $n_0=1$, we can derive
\[
\Pr\{\omega\in\Omega(n_0)\}\leq\Pr\left\{\min_{n\geq n_0}\bar{Z}(n;\omega)\geq 0\right\}=\Pr\{N(0;n_0)=\infty\}.
\]
where $N(0;n_0)$ is defined in Section \ref{subsec: runningave_explore} and we have that $\Pr\{N(0;n_0)=\infty\}=0$ by Lemma \ref{lem:normalbc_EFG}.
Clearly, this contradicts with Equation \eqref{proof:lemma3}. The proof is completed.\hfill\Halmos

\vspace*{6pt}
\subsection{Proof of Theorem \ref{thm: ExploreFirstGreedy_opt}}
\label{subsec: proofEFGopt}
Before moving to show Theorem \ref{thm: ExploreFirstGreedy_opt}, we first prepare the following Lemma \ref{lem: C_EFG}.
\begin{lem}
\label{lem: C_EFG}
Given any positive integer $n_0$, let $N(x;n_0)=\inf\{n\geq n_0:\bar{Z}(n)<x\}$. Then, $C(x; n_0)=\E[N(x;n_0)]$ is a strictly decreasing and continuous function on $x\in(0,\infty)$.
\end{lem}
\begin{proof}{Proof:}
Firstly, we show that $C(x; n_0)$ is strictly decreasing in $x$ for any fixed $n_0$. Set $x_1$ and $x_2$ with $0<x_1< x_2 < \infty$. Taking a sample-path viewpoint, it is straightforward to see that $N(x_1;n_0)\geq N(x_2;n_0)$. Moreover, we find that the strict inequality holds with a non-zero probability, i.e.,  
\begin{eqnarray*}
\Pr\left\{N(x_1;n_0)> N(x_2;n_0)\right\} & \geq & \Pr\left\{N(x_1;n_0)> n_0, N(x_2;n_0)=n_0\right\} \\
& = & \Pr\left\{x_1 \leq \bar Z(n_0) < x_2\right\} = \Phi\left(\sqrt{n_0}x_2\right) - \Phi\left(\sqrt{n_0}x_1\right) >0.
\end{eqnarray*}
As a consequence, we can conclude that $C(x_1;n_0)>C(x_2;n_0)$.

Secondly, we prove the continuity of $C(x;n_0)$ with respect to $x$. For any $x_0\geq\epsilon>0$, we denote a sequence of numbers $\{x_m,m=1,2,\ldots\}\subset  [\epsilon, \infty)$ such that $\lim_{m\to\infty} x_m=x_0$. Recalling that $N(x;n_0)$ is decreasing in $x$, we can derive  
\begin{eqnarray*}
|C(x_m;n_0)|=\E[N(x_m;n_0)]\leq \E[N(\epsilon;n_0)].
\end{eqnarray*}
From Lemma \ref{lem:normalbc_EFG}, $\E[N(\epsilon;n_0)]=C(\epsilon;n_0)<\infty$. Further, by the Lebesgue's dominated convergence theorem, 
\[
\lim_{m\to\infty} C(x_m;n_0) =\lim_{m\to\infty} \E[N(x_m;n_0)]=\E\left[\lim_{m\to\infty} N(x_m;n_0)\right]
=\E\left[N(x_0;n_0)\right] = C(x_0;n_0).
\]
In other words, $C(x;n_0)$ is continuous in $x\in[\epsilon,\infty)$. Because the continuity holds for any $\epsilon>0$, we conclude that $C(x;n_0)$ is continuous in $x\in(0,\infty)$. The proof is completed.
 \hfill\Halmos
\end{proof}

\vspace*{10pt}
\textit{\textbf{Proof of Theorem \ref{thm: ExploreFirstGreedy_opt}}:} 
The proof is conducted based on the proofs of Theorems \ref{thm: Greedy_PCS_opt} and \ref{thm: Greedy_PCS_tightness}. To analyze the explore-first greedy procedure’s PCS from the boundary-crossing perspective, we regard the minimum of the running average of the best alternative after the exploration phase, i.e., $\min_{n\ge n_0} \bar X_1(n)$, as the boundary. Indeed, following almost the same analysis, 
we can conclude the following result which is analogous to Equation \eqref{thm1: PCS4}
\begin{eqnarray}\label{proof: PCS_EFG1}
\liminf_{k\to\infty}\, {\rm PCS} &\ge& \Pr\left\{\min_{n\ge n_0} \bar X_1(n) >  \mu_1 -\gamma^-_\epsilon\right\}-\Pr\left\{\argmin_{n\ge n_0} \bar X_1(n)=\infty\right\},
\end{eqnarray}
where 
\begin{align}\label{proof: PCS_EFG2}
C\left(\frac{\gamma-\gamma_\epsilon^-}{\overline\sigma};n_0\right)=c-2\epsilon.
\end{align}
By the continuity of $C(\cdot;n_0)$ stated in Lemma \ref{lem: C_EFG}, it is easy to check that $\gamma_\epsilon^-\to \gamma_0$ as $\epsilon\to 0$. Letting $\epsilon\to$ in Equation \eqref{proof: PCS_EFG1}, we can obtain
\begin{align}\label{proof: PCS_EFG3}
 \notag  \liminf_{k\to\infty}\, {\rm PCS} &\ge \Pr\left\{\min_{n\ge n_0} \bar X_1(n) >  \mu_1 -\gamma_0\right\}-\Pr\left\{\argmin_{n\ge n_0} \bar X_1(n)=\infty\right\}\\
   &=\Pr\left\{\min_{n\ge n_0} \bar X_1(n) >  \mu_1 -\gamma_0\right\}.
\end{align}
The second inequality arises from Lemma \ref{lem:argmin_EFG}.

In addition, we define $m_\epsilon(n_0)=\inf\{n\geq n_0:\bar{X}_1(n)<\mu_1-\gamma_\epsilon^-\}$ with $\gamma_\epsilon^-$ defined in Equation \eqref{proof: PCS_EFG2}. Under the SC-CV, we derive the analogous result of Equation \eqref{proof: PICS}, namely,
\begin{eqnarray*}
\liminf_{k\to\infty}\ \mbox{PICS} \geq \Pr\left\{m_\epsilon(n_0)<\infty, \min_{1 \leq n\leq m_\epsilon(n_0)-1}\bar{X}_1(n) \geq \mu_1-\gamma_\epsilon^-\right\}.
\end{eqnarray*}
Recalling that $\gamma_\epsilon^-\to \gamma_0$ as $\epsilon\to0$, we can have $m_\epsilon(n_0)\to m(n_0)=\inf\{n\geq n_0:\bar{X}_1(n)<\mu_1-\gamma_0\}$ almost surely. Therefore, by letting $\epsilon\to 0$, the above statement is further written as
\begin{eqnarray}\label{proof: PCS_EFG4}
\liminf_{k\to\infty}\ \mbox{PICS} \geq \Pr\left\{m(n_0)<\infty, \min_{1 \leq n\leq m(n_0)-1}\bar{X}_1(n) \geq \mu_1-\gamma_0\right\}=\Pr\left\{m(n_0)<\infty\right\}.
\end{eqnarray}
Considering Equation \eqref{proof: PCS_EFG3} and Equation \eqref{proof: PCS_EFG4}, the conclusion is drawn.\hfill\Halmos

\vspace*{6pt}
\subsection{The EFG Procedure's Sample Optimality Regarding PGS}
\label{subsec: proofEFGPGS}


\begin{repeatproposition}[\textbf{Proposition 2}]
		\label{prop: EFG_PGS_opt}
Suppose that $\sigma^2_i \leq \bar \sigma^2 < \infty$ for all $i=1,\ldots,k$ and $\limsup_{k\to\infty}|{\cal G}|<\infty$ for a given $\delta>0$. If the total sampling budget $B$ satisfies $B/k = n_0+n_g$ and $n_g > C\left(\frac{\delta}{\bar \sigma}; n_0\right)-n_0$, the PGS of the explore-first greedy procedure satisfies
\[
\liminf\limits_{k \to \infty} {\rm PGS} \geq \limsup_{k\to\infty}\Pr\left\{ \exists\ i\in{\cal G} : \min_{n\ge n_0}\bar X_{i}(n) > \mu_1-\delta_0 \right\}\geq \Pr\left\{\min_{n\ge n_0}\bar X_{1}(n) > \mu_1-\delta_0\right\},
\]
where $\delta_0$ is a positive constant such that $\delta_0 \in (0, \delta)$ and $C\left(\frac{\delta-\delta_0}{\bar \sigma};n_0\right) =n_0+n_g$.
\end{repeatproposition}

\begin{proof}{Proof.}
For the explore-first greedy procedure, we again let $\bar Y_\delta(r)=\max_{i\in{\cal G}} \bar X_i(n_i)$, where $n_i\geq n_0$ denotes the sample size of alternative $i$ after the exploration phase and $r=\sum_{i\in{\cal G}} n_i$. Note that $r \geq |\cal G| n_0$. 
To analyze  the explore-first greedy procedure’s PGS from the boundary-crossing perspective, we regard the minimum of the running maximum of the good alternatives after the exploration phase, i.e., $\min_{r\ge |\cal G| n_0} \bar Y_\delta(r)$, as the boundary.
Moreover, for each alternative $i \in {\cal N}$, let $N_i(x; n_0)=\inf\{n\geq n_0: \bar X_i(n) < x\}$ denote the boundary-crossing time after $n_0$ observations $w.r.t.$ a boundary $x$. 
Then, analogous to Equation (\ref{eqn:pgs}), we have
\begin{equation}
\label{eqn:EFGpgs}
{\rm PGS}\ge \Pr\left\{\sum_{i\in{\cal N}} N_i\left(\min_{r\ge |\cal G|n_0} \bar Y_\delta(r); n_0\right)+ \argmin_{r\ge |\cal G|n_0} \bar Y_\delta(r)\leq B \right\}.
\end{equation}
We prove the conclusion based on the PGS statement in Equation \eqref{eqn:EFGpgs}.
Notice that Equation \eqref{eqn:EFGpgs} exhibits a similar form to the PCS statement in Equation \eqref{eqn:pcs_EFG}. Following the same logic of obtaining Equation \eqref{proof: PGS1} for proving Proposition \ref{prop: Greedy_PGS_opt}, we can obtain an analogous result of Equation \eqref{proof: PCS_EFG1} as follows,
\begin{eqnarray}\label{proof: EFGPGS1}
\liminf_{k\to\infty}\, {\rm PGS} &\ge& \limsup_{k\to\infty}\,\Pr\left\{\min_{r\ge |\cal G|n_0} \bar Y_\delta(r) >  \mu_1 -\delta_\epsilon^-\right\}-\liminf_{k\to\infty}\,\Pr\left\{\argmin_{r\ge |\cal G|n_0} \bar Y_\delta(r)=\infty\right\},  \quad\quad
\end{eqnarray}
where $0<\epsilon<(c-C(\delta/\overline{\sigma}; n_0))/2$ and 
$\delta_\epsilon^-$ is chosen to satisfy
$C\left(\frac{\delta-\delta_\epsilon^-}{\overline\sigma};n_0\right)=c-2\epsilon.$
By the continuity of $C(\cdot;n_0)$ stated in Lemma \ref{lem: C_EFG}, it is easy to check that $\delta_\epsilon^- \to \delta_0$ as $\epsilon\to 0$.
Letting $\epsilon\to 0$ in Equation \eqref{proof: EFGPGS1} yields
\begin{eqnarray}\label{proof: EFGPGS2}
\liminf_{k\to\infty}\, {\rm PGS} &\ge& \limsup_{k\to\infty}\,\Pr\left\{\min_{r\ge |\cal G|n_0} \bar Y_\delta(r) >  \mu_1 -\delta_0\right\}-\liminf_{k\to\infty}\,\Pr\left\{\argmin_{r\ge |\cal G|n_0} \bar Y_\delta(r)=\infty\right\}, \quad\quad
\end{eqnarray}

Notice that one reasonable upper bound for $\argmin_{r\ge |\cal G|n_0} \bar Y_\delta(r)$ is that
\begin{eqnarray*}
  \argmin_{r\ge |\cal G|n_0} \bar Y_\delta(r) \le \sum_{i\in{\cal G}} \argmin_{n\ge n_0}\bar X_i(n),
\end{eqnarray*}
i.e., the  process $\left\{\bar Y_\delta (r), r\ge |\cal G|n_0\right\}$ reaches its minimum before all $\left\{\bar X_i(n), n\geq n_0\right\}$ processes reach their minimums for all $i\in{\cal G}$. 
Then, we can conclude that $\Pr\left\{\argmin_{r\ge |\cal G|n_0} \bar Y_\delta(r)=\infty\right\}\le \sum_{i\in{\cal G}}\Pr\left\{\argmin_{n\ge n_0}\bar X_i(n)=\infty\right\}$,  which amounts to zero by Lemma \ref{lem:argmin_EFG} and the assumption that $\limsup_{k\to\infty}|{\cal G}|<\infty$. Meanwhile, because $\min_{r\ge |\cal G|n_0} \bar Y_\delta(r)=\max_{i\in\mathcal{G}}\min_{n\geq n_0}\bar{X}_i(n)$, Equation \eqref{proof: EFGPGS2} can be further written as
\begin{eqnarray*}
\liminf_{k\to\infty} {\rm PGS} \ge \limsup_{k\to\infty}\Pr\left\{\min_{r\ge |\cal G|n_0} \bar Y_\delta(r) >   \mu_1 -\delta_0\right\}
=\limsup_{k\to\infty}\Pr\left\{\exists\,i\in\mathcal{G}:\min_{n\geq n_0}\bar{X}_i(n) > \mu_1-\delta_0\right\}.
\end{eqnarray*}
It is clear that the right-hand-side probability is at least $\Pr\left\{\min_{n\geq n_0}\bar{X}_1(n) > \mu_1-\delta_0\right\}$. The proof is completed.\hfill\Halmos
\end{proof}

\vspace*{6pt}
\subsection{Proof of Lemma \ref{lem: bct_mean_convergence}}
\label{subsec: proofCto0}
For any fixed $n_0$, the function $C(x;n_0)$ for any $x>0$ can be given by
\begin{eqnarray*}
C(x;n_0)-n_0 &=& \E\left[\inf\{n\geq n_0: \bar{Z}(n)< x\}\right]-n_0\\
&=& \E\left[\inf\{n\geq 0: \bar{Z}(n+n_0)< x\}\right]:=\E[M(x;n_0)]
\end{eqnarray*}
where we denote $M(x;n_0) = \inf\{n\geq 0: \bar{Z}(n+n_0)<x\}$ above for simplicity of notation. By this definition, we must have the following fact holds: when $M(n_0)\geq n\geq 1$ is true, we must have that the event $\bar{Z}(n_0+n-1)\ge x$ holds. This indicates that
\begin{eqnarray*}
\E[M(n_0;x)] = \sum_{n=1}^{\infty} \Pr\{M(n_0;x)\geq n\}\leq \sum_{n=1}^{\infty} \Pr\{\bar{Z}(n+n_0-1)\ge x\}=\sum_{n=n_0}^{\infty} \Pr\{\bar{Z}(n)\ge x\}=\sum_{n=n_0}^{\infty} \bar\Phi(\sqrt{n}x),
\end{eqnarray*}
where $\bar\Phi(\cdot)$ denotes the complementary cumulative distribution function of a standard normal distribution. Because $\bar\Phi(x)\leq \exp(-x^2/2)$ for $x>0$,
we further have
\[
\E[M(n_0;x)] \leq  \sum_{n=n_0}^{\infty} \exp(-n x^2/2)
=\frac{\exp\left(-n_0 x^2/2\right)}{1-\exp(-x^2/2)}.
\]

Set $x=(\gamma-\gamma_0)/\overline\sigma$, and we can have
\begin{eqnarray}
n_g = C\left(\frac{\gamma-\gamma_0}{\overline\sigma}\right)-n_0 =\E\left[M\left(\frac{\gamma-\gamma_0}{\overline\sigma};n_0\right)\right]\leq \frac{\exp\left(-n_0 \frac{(\gamma-\gamma_0)^2}{2\overline\sigma^2}\right)}{1-\exp\left(-\frac{(\gamma-\gamma_0)^2}{2\overline\sigma^2}\right)}.
\end{eqnarray}
Letting $\kappa=\frac{(\gamma-\gamma_0)^2}{2\overline\sigma^2}$ and $\beta=(1-e^{-\kappa})^{-1}$ leads to the conclusion. The proof is completed.
\hfill \Halmos

\vspace*{6pt}
\newpage
\subsection{The Enhanced Explore-First Greedy (EFG$^+$) Procedure and Proof of Theorem \ref{thm: EFGP_Opt}}
\label{subsec: proofEFGP_Opt}

\begin{procedure}[ht]
  \caption{\textbf{Enhanced Explore-First Greedy (EFG$^+$) Procedure}}
  \label{procedure: EFGPlus}
  \begin{algorithmic}[1]
  \REQUIRE {$k$ alternatives $X_1,\ldots,X_k$, the total sampling budget $B=(n_{sd} + n_0 + n_g) k $, and the total number of groups $G$ satisfying that $2^G-1 \leq k$ and $2 \leq G\leq n_0$.
  } 
  \STATE For each alternative $i=1,\ldots, k$, take $n_{sd}$ independent observations $X_{i1},\ldots,X_{in_{sd}}$ and set $\bar X^{sd}_i =\frac{1}{n_{sd}}\sum_{j=1}^{n_{sd}} X_{ij}$.
  \STATE According to $\bar X^{sd}_i$, sort the alternatives in descending order as $\{ (1), (2), \ldots, ( k) \}$.
  \STATE 
  Let $I^r$ denote the group of alternatives for $r=1,\ldots, G$. Set $\Delta = 2^G-1$. Then, let $I^{1} = \{ (1), (2), \ldots, ( \lfloor  k / \Delta \rfloor ) \}$,  $I^{r} = \{ (\lfloor k 2^{r-2} / \Delta \rfloor+1), \ldots, ( \lfloor k 2^{r-1} / \Delta \rfloor) \}$ for $r=2, \ldots, G-1$, and let $I^{G} = \{ (\lfloor k 2^{G-2} / \Delta \rfloor+1), \ldots, (k) \}$.
  \FOR{$r \in \{1, 2, 3, \ldots, G\}$}
  \STATE Set the sample size allocated to each alternative $i \in I^r$ as $n^r = \left\lfloor \frac{n_0 (2^G-1)}{G 2^{r-1} } \right\rfloor$. 
  \STATE For each alternative $i \in I^r$, take $n^r$ independent observations $X_{i1},\ldots,X_{in^r}$, set $\bar X_i(n^r)=\frac{1}{n^r}\sum_{j=1}^{n^r} X_{ij}$ and  let $n_i=n^r$.
  \ENDFOR
  \WHILE{$\sum_{i=1}^k n_i < (n_0 +n_g) k$}
  \STATE Let $s = \arg\max_{i=1,\ldots,k}  \bar{X}_{i}\left(n_i\right)$ and take one observation $x_{s}$ from alternative $s$;
  \STATE Update $\bar X_{s}(n_{s}+1) = \frac{1}{n_{s}+1}\left[n_{s}\bar X_{s}(n_s) + x_s\right]$ and let $n_s = n_s+1$;
  \ENDWHILE
   \STATE Select $\arg\max_{i\in\{1,\ldots,k\}} \bar{X}_{i}\left(n_i\right)$  as the best.
  \end{algorithmic}
  \end{procedure}

 In the following, we first show that the total amount of observations allocated in the exploration phase, $\sum_{r=1}^G n^r |I^r|$,  does not exceed the given exploration budget $n_0 k$. Specifically, we have,
\begin{eqnarray*}
    \sum_{r=1}^G n^r |I^r| 
= \sum_{r=1}^{G-1} n^r |I^r| + n^G\times \left(k - \sum_{r=1}^{G-1}|I^r|\right) 
= \sum_{r=1}^{G-1} (n^r-n^G)|I^r| + n^G k.
\end{eqnarray*}
As $n^r\geq n_G$ and $|I^r|$ is chosen as $\left\lfloor \frac{k 2^{r-1}}{2^G-1} \right\rfloor$ for $r=1,\ldots, G-1$, we can further derive
\begin{eqnarray}
\label{eq: EFGP_exploration}
  \notag \sum_{r=1}^G n^r |I^r| &\leq&  \sum_{r=1}^{G-1} (n^r-n^G)\frac{k 2^{r-1}}{2^G-1}  + n^G k
   = \sum_{r=1}^{G-1} n^r\frac{k 2^{r-1}}{2^G-1}  -n^G k + n^G k\\
   &=& \sum_{r=1}^{G} \left\lfloor \frac{n_0(2^G-1)}{G 2^{r-1}} \right\rfloor \frac{k 2^{r-1}}{2^G-1}\leq \sum_{r=1}^{G}  \frac{n_0(2^G-1)}{G 2^{r-1}}  \frac{k 2^{r-1}}{2^G-1}= n_0k.
\end{eqnarray}

Now, we study the PCS of the EFG$^+$ procedure. 
For each alternative $i$, we denote its allocated sample size in the exploration phase by $e_i$. 
For simplicity, we let $\mathbf{e}$ represent the vector $[e_1, \ldots, e_k]$. 
Applying the boundary-crossing perspective, we can  formulate the PCS of the EFG$^+$ procedure as
\begin{eqnarray}
	\label{PCS: EFGPlus}
	\mbox{PCS}  \geq  \mbox{E}_{\mathbf{e}} \left[\Pr\left\{(n_0+n_g) k \geq \argmin_{n\geq e_1}\bar X_1(n)+\sum_{i=2}^{k}  N_i\left(\min_{n \geq e_1}\bar X_1(n); e_i\right)  \Bigm| \mathbf{e} \right\}\right].
\end{eqnarray}
The above PCS lower bound is harder to analyze compared to that for the Greedy and  EFG procedures.  It is because the $e_i$ are mutually correlated random variables as they are calculated based on all alternatives' ranking information in the seeding phase. 
 Thus, we consider eliminating the impact of the $\mathbf{e}$ on the PCS lower bound inside the expectation and thus simplify the analysis. 
 
Notice the fact that
$n^G \leq e_i \leq n^1$, and thus $\argmin_{n\geq n^1}\bar X_1(n) \geq \argmin_{n\geq e_1}\bar X_1(n)\geq \min_{n \geq n^G}\bar X_1(n)$ from the sample-path viewpoint. Then,  for the PCS lower bound in \eqref{PCS: EFGPlus}, we further have
\begin{eqnarray}
\label{PCS: EFGPlus2}
	\notag \mbox{PCS} & \geq &  \mbox{E}_{\mathbf{e}} \left[\Pr\left\{(n_0+n_g) k \geq \argmin_{n\geq n^1}\bar X_1(n)+\sum_{i=2}^{k}  N_i\left(\min_{n \geq n^G}\bar X_1(n); e_i\right)  \Bigm| \mathbf{e} \right\}\right]\\
& \geq & \Pr\left\{(n_0+n_g) k \geq \argmin_{n\geq n^1}\bar X_1(n)+\sum_{i=2}^{k}  N_i\left(\min_{n \geq n^G}\bar X_1(n); n^1\right)  \right\}.
\end{eqnarray}
The second inequality holds because $N_i(x;n)$ is non-decreasing in $n$ for any fixed $x$.
We note that the PCS lower bound in Equation \eqref{PCS: EFGPlus2} shares almost the same form as the PCS lower bound of the EFG procedure in Equation \eqref{eqn:pcs_EFG}.
Thus, we can prove Theorem \ref{thm: EFGP_Opt} by similar arguments of showing the sample optimality of the EFG procedure stated in  \ref{subsec: proofEFGopt}.

Specifically, facilitated by the condition
that the boundary $\min_{n\ge n^G} \bar X_1(n)$ is always above the level $\mu_1 -\gamma^-_\epsilon$
where $\gamma^-_\epsilon$ is the constant satisfying
\begin{align}\label{proof: PCS_EFGP2}
C\left(\frac{\gamma-\gamma_\epsilon^-}{\overline\sigma};n^1\right)=n_0 + n_g-2\epsilon,
\end{align}
we can get the following result which is analogous to Equation \eqref{proof: PCS_EFG1}, 
\begin{eqnarray}\label{proof: PCS_EFGP1}
\liminf_{k\to\infty}\, {\rm PCS} &\ge& \Pr\left\{\min_{n\ge n^G} \bar X_1(n) >  \mu_1 -\gamma^-_\epsilon\right\}-\Pr\left\{\argmin_{n\ge n^1} \bar X_1(n)=\infty\right\}.
\end{eqnarray}
Notice that $\gamma_\epsilon^-\to \gamma_0$ as $\epsilon\to 0$, where $\gamma_0$ is the unique constant satisfying $C\left(\frac{\gamma-\gamma_0}{\overline\sigma};n^1\right)=n_0 + n_g$. 
Then, by letting $\epsilon\to 0$ in Equation \eqref{proof: PCS_EFGP1} and applying Lemma  \ref{lem:argmin_EFG}, we can further derive
\begin{align}\label{proof: PCS_EFGP3}
 \notag  \liminf_{k\to\infty}\, {\rm PCS} &\ge \Pr\left\{\min_{n\ge n^G} \bar X_1(n) >  \mu_1 -\gamma_0\right\}-\Pr\left\{\argmin_{n\ge n^1} \bar X_1(n)=\infty\right\}\\
   &=\Pr\left\{\min_{n\ge n^G} \bar X_1(n)>  \mu_1 -\gamma_0\right\}.
\end{align}
The conclusion of interest is now drawn. \hfill \Halmos

\newpage
\section{\magenta{Additional Numerical Experiments \& Results}}
\label{sec: numerical_EC}

\vspace*{6pt}
\subsection{Additional Numerical Results for the Configurations EM-IV, EM-DV}
\label{subsec: IVDV_EC}

\begin{figure}[htbp]
  \centering
  \includegraphics[width=0.75\textwidth]{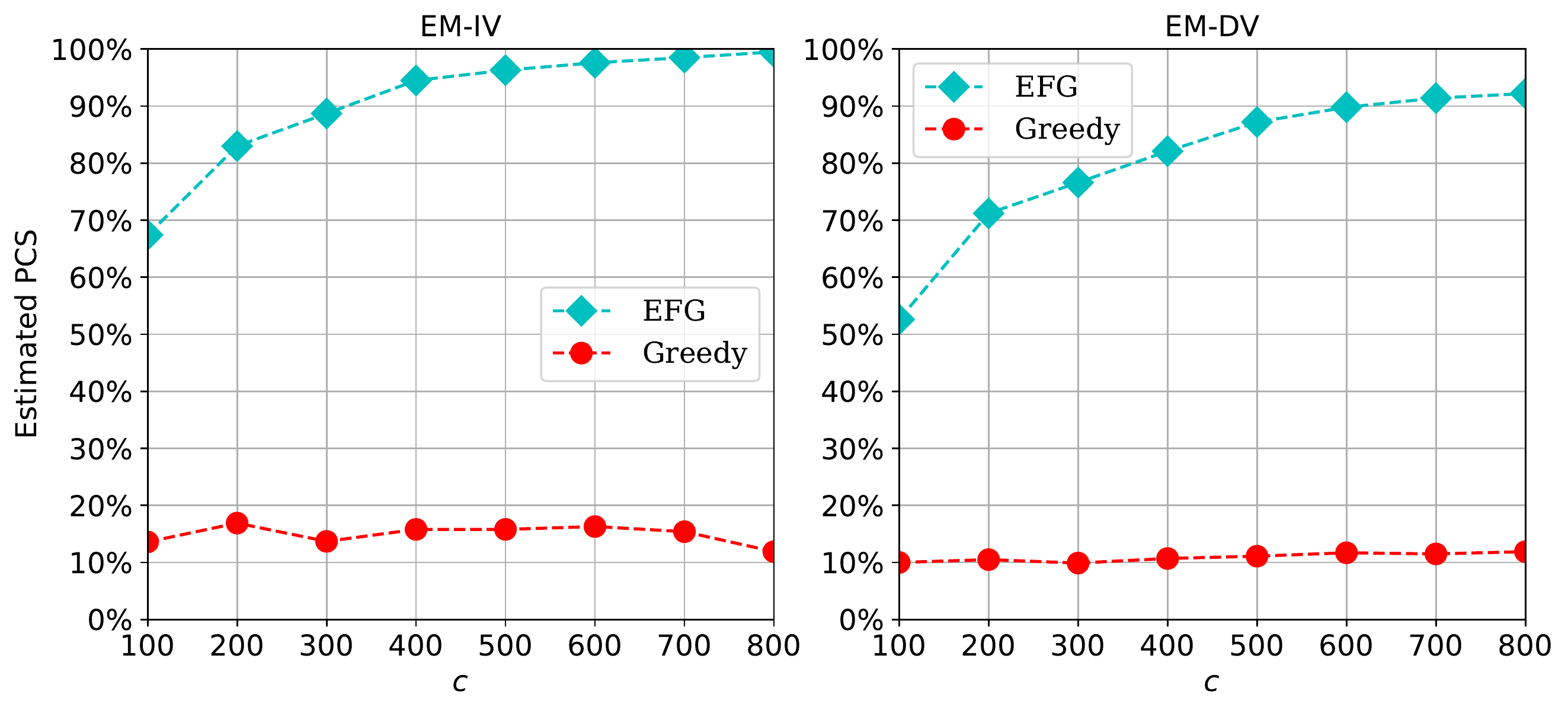}
  \caption{Estimated PCS of the EFG procedure and the greedy procedure for different values of $c$. $k=8192$. }
  \label{numerical: EFGConsistency2}
\end{figure}

\begin{figure}[htbp]
  \centering
  \includegraphics[width=0.75\textwidth]{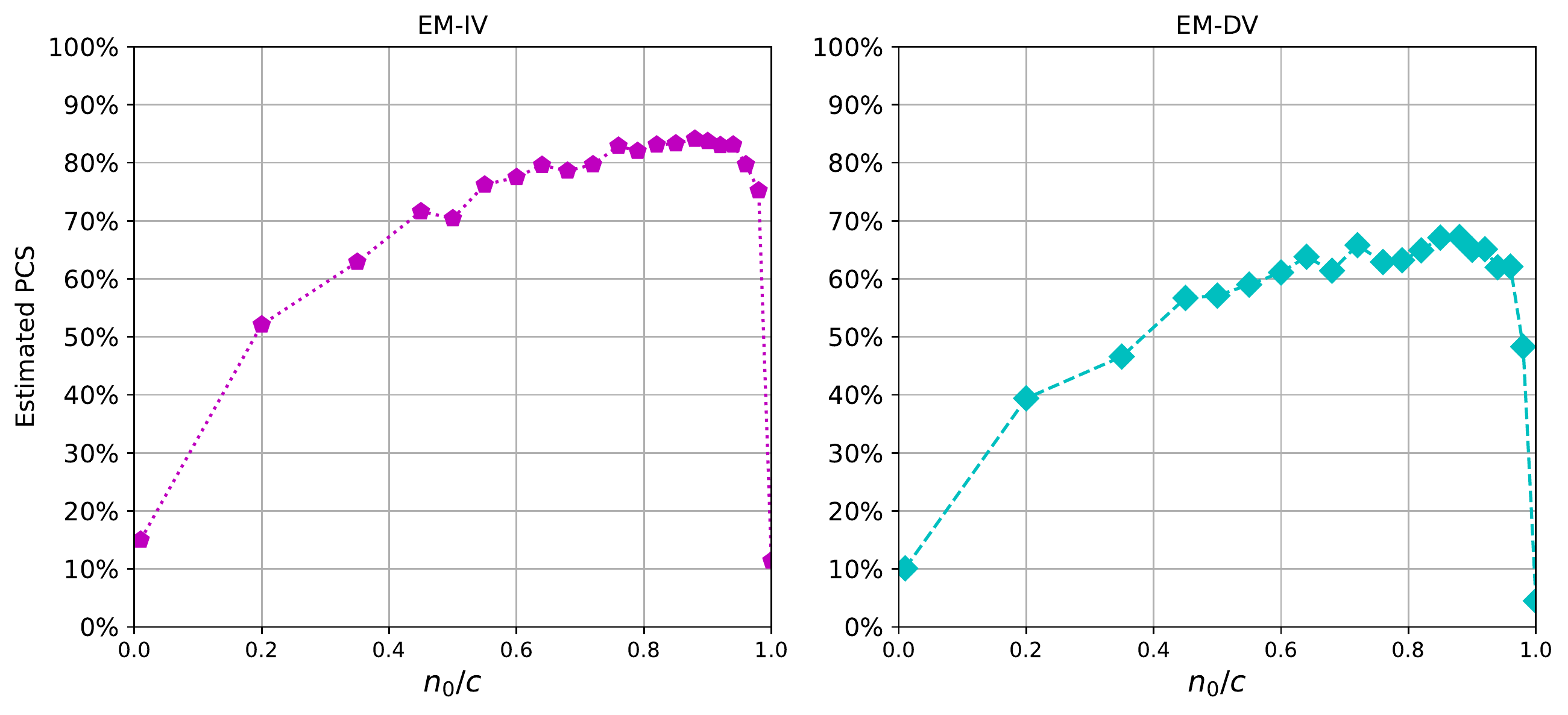}
  \caption{Estimated PCS of the EFG procedure for different values of $p=n_0/c$. $c=200$ and $k=8192$. }
  \label{numerical: EFGBudgetAllocation2}
\end{figure}

\subsection{The Slippage Configuration is Least Favorable for the Greedy Procedures}
\label{subsec: least_favorable_EC}
Theorem \ref{thm: Greedy_PCS_tightness}, together with Theorem \ref{thm: ExploreFirstGreedy_opt}, show that for the greedy procedures, when $k \rightarrow \infty$, the slippage configuration is least favorable among all configurations of means satisfying Assumption \ref{assu:asyreg}. We now provide additional numerical results to validate this result and show that the slippage configuration is also least favorable for a finite $k$. 

We consider the following problem configuration 
\[\mu_1 = 0.1, \mu_i = -\frac{\lambda(k-2)}{k}, \forall i=2, \dots, k; \sigma_i^2 = 1, \forall i=1, \dots, k,\]
where $\lambda>0$ can be a constant or a function of $k$. In the configuration, the gap between the best and second-best means, i.e., $\mu_2-\mu_1$, remains $0.1$, regardless of the choice of $\lambda$. Besides, the inferior means of the non-best alternatives are distributed in the set $[-\lambda, 0]$. 
In this experiment, we consider the following five choices of $\lambda\in\{ 0, 0.2, 2, \sqrt{k}, 0.1k\}$. When $\lambda=0$, the mean configuration is the slippage configuration. As $\lambda$ is increased, the inferior means become more dispersed.
When $\lambda = 0.1k$, the inferior means are progressively worse as $k$ grows. 

For each choice of $\lambda$, we test the greedy procedure and the EFG procedure with the same experiment setting used in Section \ref{subsec: numerical1}. We then plot the estimated PCS of the two procedures against different $k$ for each $\lambda$ in Figure \ref{fig: least_favorable}. From Figure \ref{fig: least_favorable}, we can find that for both procedures, the PCS under $\lambda=0$ is the lowest for different values of $k$. This shows that the slippage configuration is least favorable for the greedy procedures.
Furthermore, as the means of the alternatives become more dispersed, the PCS of the procedures become higher. 

\begin{figure}[htbp]
  \centering
  \includegraphics[width=0.8 \textwidth]{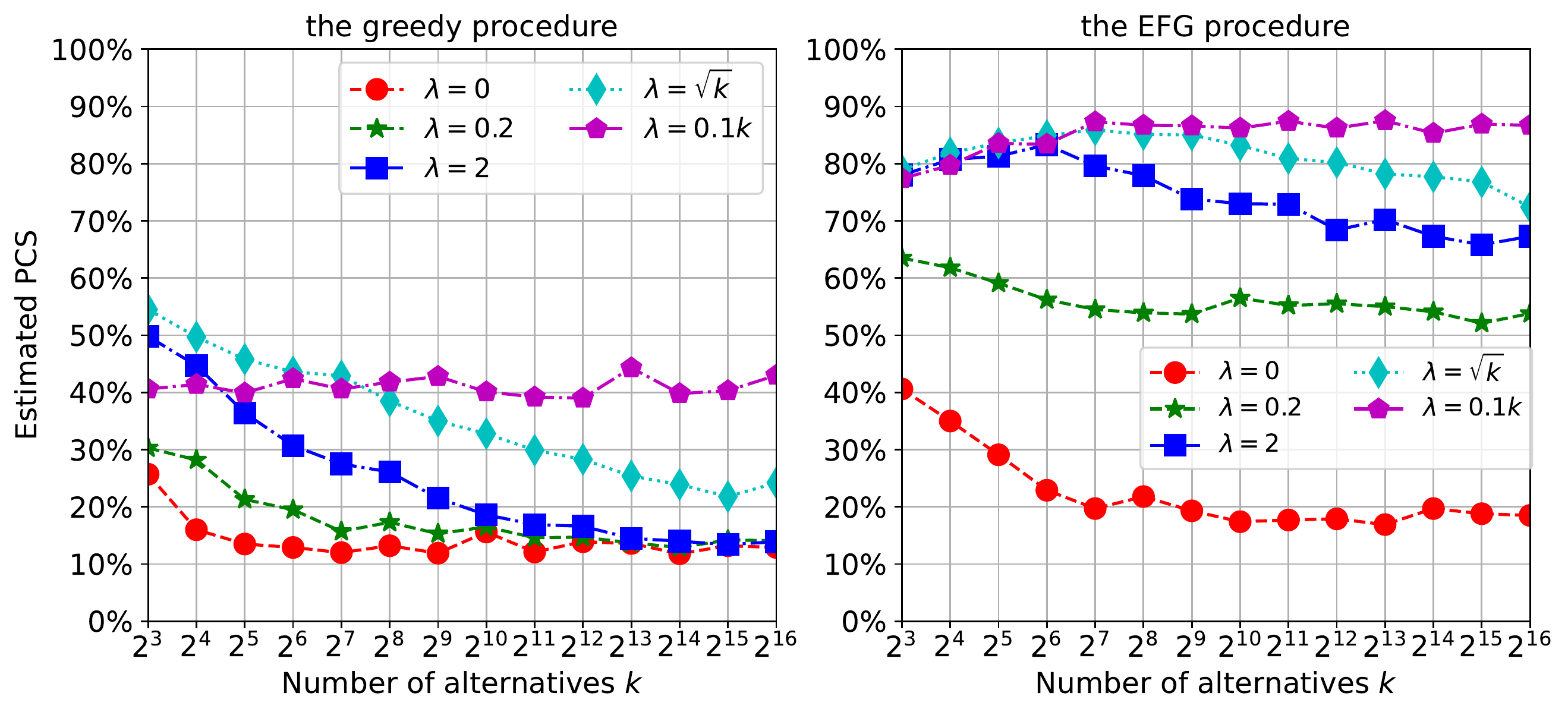}
  \caption{Estimated PCS of the greedy procedure and the EFG procedure under different mean configurations.}
  \label{fig: least_favorable}
\end{figure}

\subsection{Budget Allocation between Exploration and Exploitation}
\label{subsec: allocationEE}

 In Section \ref{subsec: allocationEEmain}, using the problems with $k=8192$ alternatives, we show that the PCS of the EFG procedure may be insensitive to the value of $p$, i.e., the proportion of the exploration budget to the total sampling budget, over a wide near-optimal range of $p$. Now we study how the optimal $p$ and the near-optimal range of $p$ vary for different $k$.

In this experiment, we again use the configurations SC-CV, EM-CV, EM-IV and EM-DV, and set the number of alternatives as $k=2^l$ with $l$ ranging from $3$ to $14$. For each $k$, we let $p=0.02a$ where $a$ is an integer and let $a$ increase from 0 to 50. 
For every combination of $k$ and $p$, we use a total sampling budget $B=200k$ and estimate the PCS based on 2000 independent macro replications.
Recall that the PCS curve against different values of $p$ has an inverted U-shape. For each $k$, we can find an optimal value of $p$ that maximizes the PCS and then estimate a near-optimal range of $p$ by setting the lower (upper) bound of the range as the first (last) $p$ that maintains a PCS within $10\%$ to the maximal PCS. For each problem configuration, we plot the optimal $p$ and the near-optimal range of $p$ against different $k$ in Figure \ref{fig: allocationEE}.
\begin{figure}[htbp]
  \centering
    \includegraphics[width=0.75\textwidth]{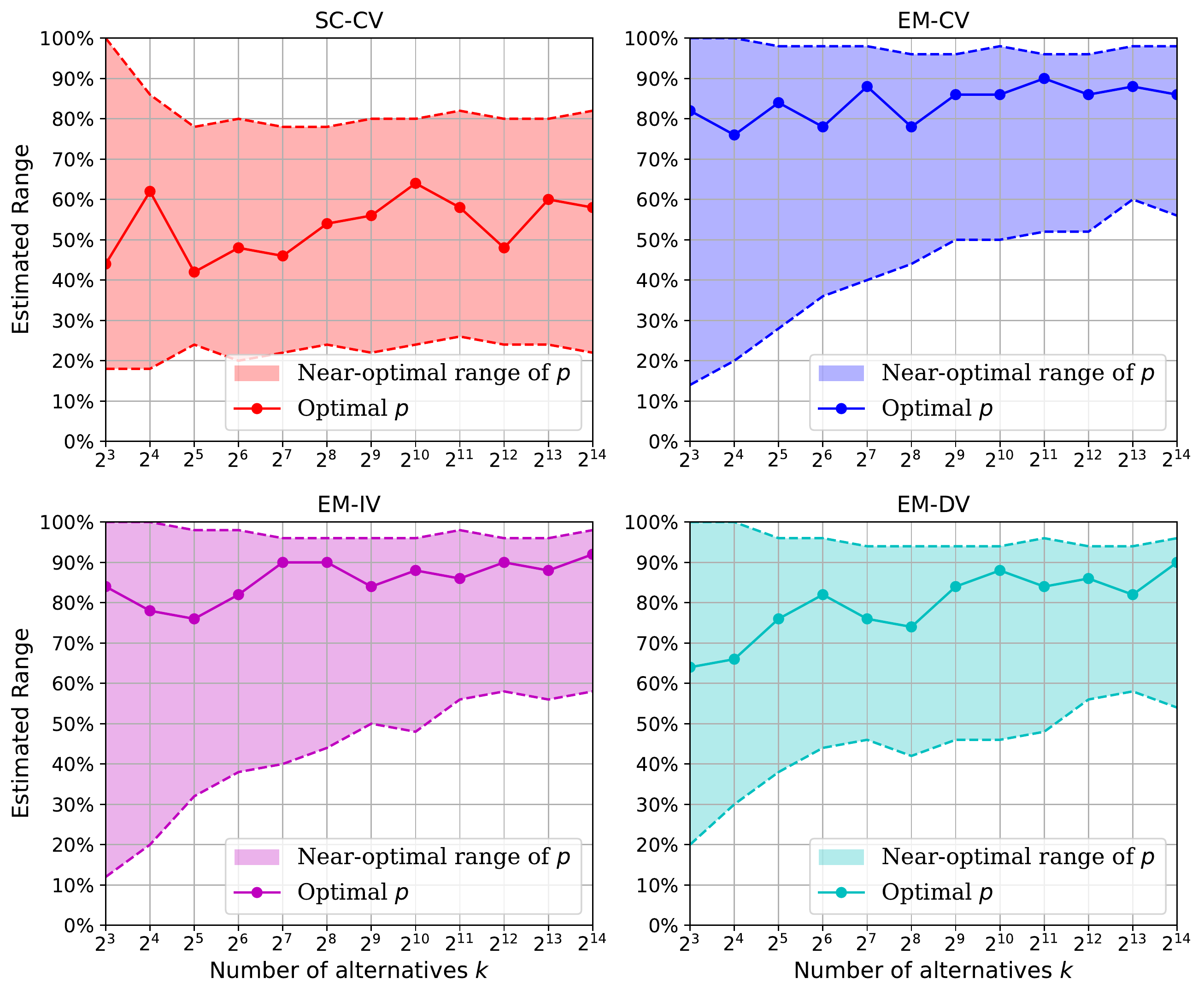}
    \caption[]{Estimated optimal $p$ and the near-optimal range under different configurations.}
    \label{fig: allocationEE}
  \end{figure}

From Figure \ref{fig: allocationEE}, we can observe that the PCS of the EFG procedure is not sensitive to the value of $p$ over a wide range for both small-scale and large-scale problems. For all four configurations, as $k$ grows, the near-optimal range of $p$ shrinks slowly. However, after $k$ is large enough, e.g., $k \geq 2^{13}=8192$, the range will remain almost the same. Notice that under SC-CV, the optimal $p$ and the lower and upper bounds of the near-optimal range are typically much lower than those under other configurations. This is because SC-CV is a much harder configuration under which the EFG procedure needs more greedy budget to maintain the sample optimality. 

\subsection{Comparison between Sample-Optimal Fixed-Budget R\&S Procedures}
\label{subsec: EC_comparison}

\subsubsection{Procedures in Comparison.}
\label{subsec: settings_EC}
 The necessary introduction and the implementation settings of the compared procedures are summarized as follows. 
\begin{itemize}
  \item \textbf{EFG procedure and EFG$^+$ procedure}.
   For the EFG procedure,  when solving the practical problem, we allocate $90\%$ of the total sampling budget to the exploration phase. Then,
for the EFG$^+$ procedure, we allocate $20\%$ of the total sampling budget to the seeding phase and another $10\%$ to the greedy phase. Moreover, we let $G=11$, which is selected to let each group have at least one alternative for the problem instances tested in the experiments. 
  \item \textbf{FBKT procedure and FBKT-Seeding procedure}. We have briefly introduced the FBKT procedure in Section \ref{sec: intro}. Readers can go there for a review. The FBKT-Seeding procedure improves the FBKT procedure with a seeding approach, and it is expected to perform better than the FBKT procedure when all alternatives' means are scattered over a wide range. Readers can refer to \cite{itemHong2022} for more details. When implementing the procedures, we follow the default settings extracted from the paper.
  \item \textbf{Sequential halving (SH) procedure and the modified SH procedure}. 
  The original SH procedure is developed by \cite{itemKarnin2013}, and the modified procedure is introduced in Appendix B.6 of \cite{itemZhao2023}. They both proceed in $L=\lceil \log_2k\rceil$ rounds. In each round $l$, $l=1,\ldots, L$, the (modified) SH procedure allocates a sampling budget $\left(T_{\ell} = \frac{B}{L}\right)$ $T_{\ell}=\left\lfloor\frac{B}{81 k} \cdot\left(\frac{16}{9}\right)^{\ell-1} \cdot \ell\right\rfloor$ equally to the survived alternatives and eliminates the lower half in terms of the sample means. Then, the procedures select the unique alternative that survives all the $L$ rounds as the best.
  \end{itemize}

\subsubsection{The Throughput Maximization Problem.}
\label{subsec: TP_EC}
The throughput maximization involves a three-stage flow line with a single service station per stage and each service station has $i.i.d.$ exponential service time. For stage $i$, $i=1,2,3$, the service rate is denoted as $x_i$. In stage 1, there are an infinite number of jobs waiting to be processed in front of the service station. Each job needs to pass through the three stages in sequence and be processed by all three service stations in turn. For stage $i$, $i=2,3$, the buffer size (or storage space) $b_i$ is finite, and it includes the position in the service station. The problem adopts a service protocol called \textit{production blocking}. It means that for stage $i$, $i=1,2$, if the downstream stage is full, the service station is blocked, then it cannot transfer a finished job to the next stage and serve a new job until the buffer of the downstream stage becomes available. 

The goal of the problem is to maximize the expected throughput of the three-stage flow line by allocating a total service rate $S_1 \in\mathbb{Z}_+^5$ among the three service stations and a total buffer size $S_2 \in\mathbb{Z}_+^5$  among stage 2 and stage 3. Formally, we write the optimization problem as   
\begin{eqnarray*}
    \max_x&&\mbox{E}\left[f\left(x;\xi\right)\right]\\
    s.t.&&x_1+x_2+x_3=S_1\\
    &&b_2+b_3=S_2\\
    &&x=\left(x_1,x_2,x_3,b_2,b_3\right)\in\mathbb{Z}_+^5
\end{eqnarray*}
where $f$ represents the flow line's throughput and $\xi$ represents the randomness. A nice feature of the problem is that for each alternative (i.e., feasible solution), we can evaluate its true mean throughput by solving the balancing equations of the underlying Markov chains (see p. 189 of  \cite{itemBuzacott1993}).  Then, we can easily find the best alternative or identify good alternatives given the IZ parameter $\delta$. 

\subsubsection{Implementation Settings.}
\label{subsec: setting_EC}
For each alternative, every time we generate an observation of the throughput, we independently simulate the flow line until 1050 jobs are finished and estimate the throughput based on the last 50 jobs. In every experiment, we run 500 independent macro replications to estimate the PCS or the PGS, and round the results to two decimal places. 
 
  \subsubsection{A Comparison Using Configurations with Bounded Means.} 
  \label{subsec: more_compare_EC}
  We now compare the sample-optimal fixed-budget R\&S procedures on the configurations EM-IV and EM-DV used previously in Section \ref{sec: numerical}. 
  A common feature of these configurations is that the means of all alternatives are bounded between -1 and 0.1 regardless of how large $k$ is. For each configuration, we set $k=2^l$ with $l$ ranging from 11 to 16 and set $B=100k$ for each $k$. The estimated PCS based on 1000 independent macro replications of the six procedures against different $k$  under each configuration is plotted in  Figure \ref{fig: compareRateOptimal}.
  
  \begin{figure}[htbp]
  \centering
    \includegraphics[width=0.85\textwidth]{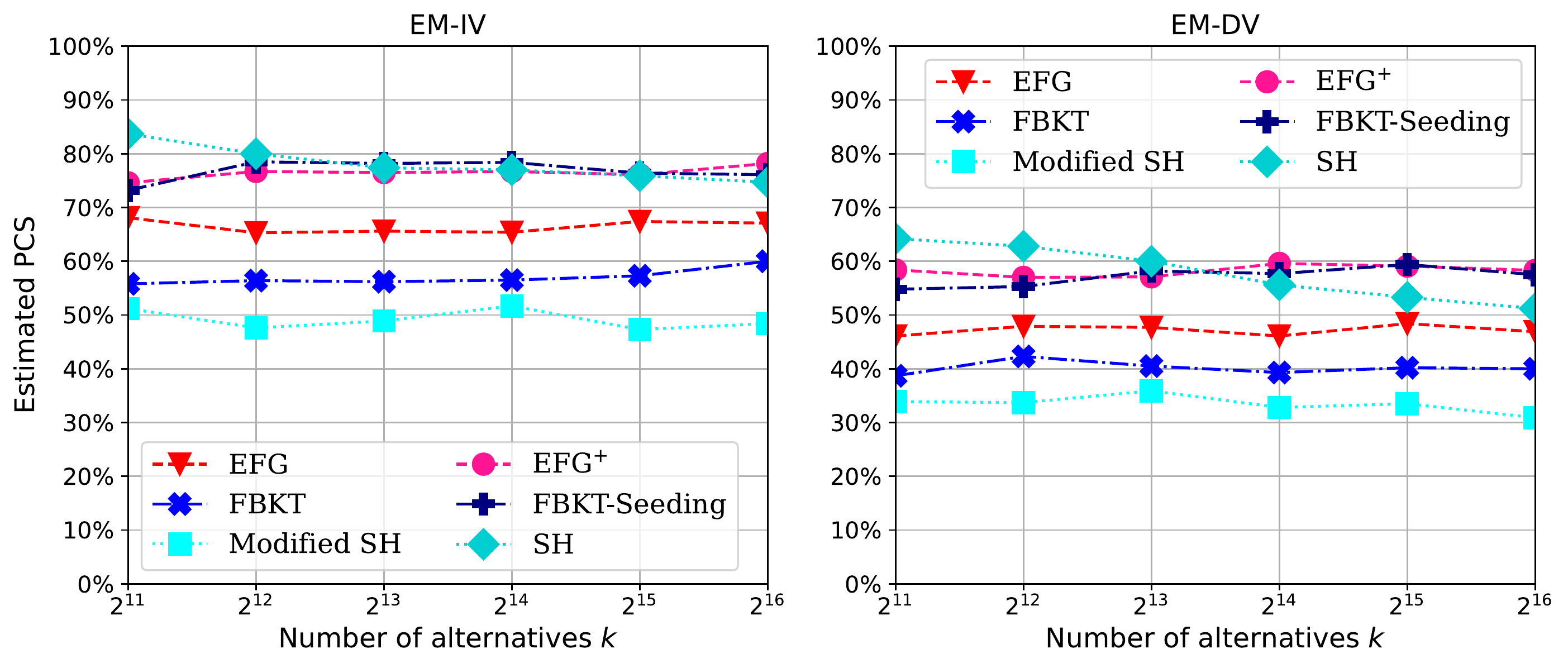}
    \caption[]{A comparison between the EFG procedures, the FBKT procedures, and the SH procedures.}
    \label{fig: compareRateOptimal}
  \end{figure}
  
  We have the following findings from Figure \ref{fig: compareRateOptimal}. First, the EFG$^+$ procedure performs significantly better than the EFG procedure for the configurations with bounded means. Under both configurations, the gap between the PCS of the two procedures remains around 10\% as $k$ grows. 
  Second, when $k$ is large enough, the EFG$^+$ procedure achieves the same level of PCS as the FBKT-Seeding procedure and the SH procedure under EM-IV and a higher PCS than the SH procedure under EM-DV. This result indicates that, although the greedy procedures achieve the sample optimality in a completely different and surprisingly simple manner, they have the potential to perform better than the median elimination procedures.

\subsection{Budget Allocation Among the Alternatives}
\label{subsec: allocationAlternativesMain}
 
In this subsection, we study how the EFG procedure allocates the sampling budget to all the alternatives. Since the exploration phase is simply an EA phase, we focus on the budget allocation in the greedy phase. Specifically, we study the budget allocation using  the following three different problem configurations with a unit variance: 
 \begin{itemize}
  \item the slippage configuration where $\mu_1 = 0.1, \mu_i = 0, \forall  i=2, \dots, k$;
  \item the configuration with bounded means where  $\mu_1 = 0.1, \mu_i = -\frac{2(k-2)}{k},\forall  i=2, \dots, k$;
  \item the configuration with progressively worse means where $\mu_1 = 0.1, \mu_i = -0.1(k-2), \forall  i=2, \dots, k$.
 \end{itemize}

We test the EFG procedure for each configuration with the same implementation setting used in Section \ref{subsec: numerical1} and analyze the budget allocation based on 1000 independent macro replications.
\orange{For a start, we plot the PCS of the EFG procedure under each configuration against different $k$ in Figure \ref{fig: budget_allocation_pcs}. The PCS curves show that the EFG procedure is sample optimal for all three configurations. Besides, the PCS of the EFG procedure becomes higher as the means become more scattered from the slippage configuration to the configuration with progressively worse means.}
Next, we plot the average proportion of the non-best alternatives that obtain observations in the greedy phase and the average minimal mean of these alternatives against different $k$
in Figure \ref{fig: budget_allocation1}.  We then plot the average sample size obtained in the greedy phase of the non-best alternatives ever selected in the greedy phase against different $k$ in Figure \ref{fig: budget_allocation3}, which also displays how the proportion of the greedy budget allocated to the best alternative changes as $k$ increases. 

\begin{figure}[htbp]
  \centering
  \includegraphics[width=0.45 \textwidth]{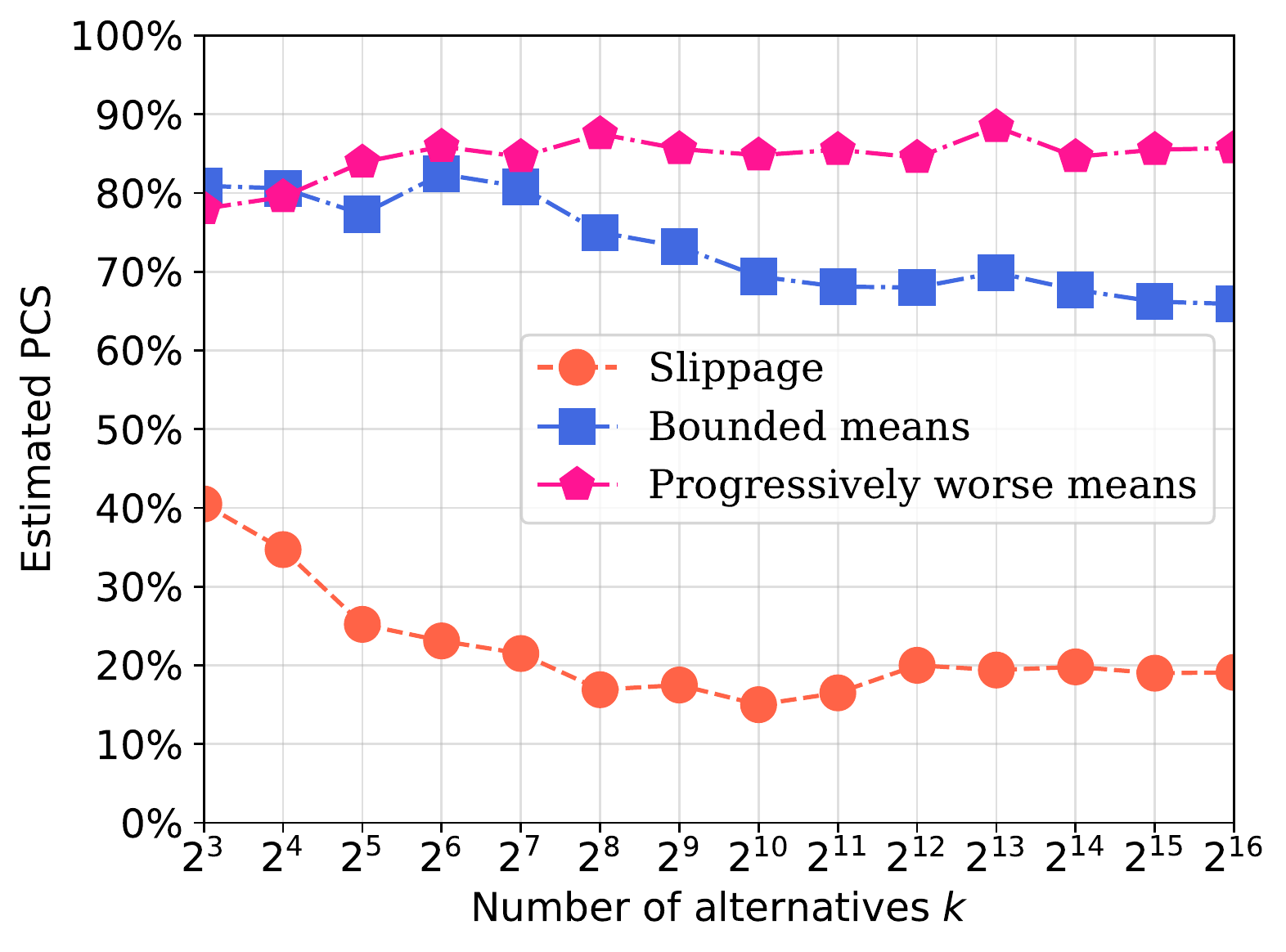}
  \caption{Estimated PGS of the EFG procedure under different configurations.}
  \label{fig: budget_allocation_pcs}
\end{figure}

\begin{figure}[htbp]
  \centering
  \includegraphics[width=0.75 \textwidth]{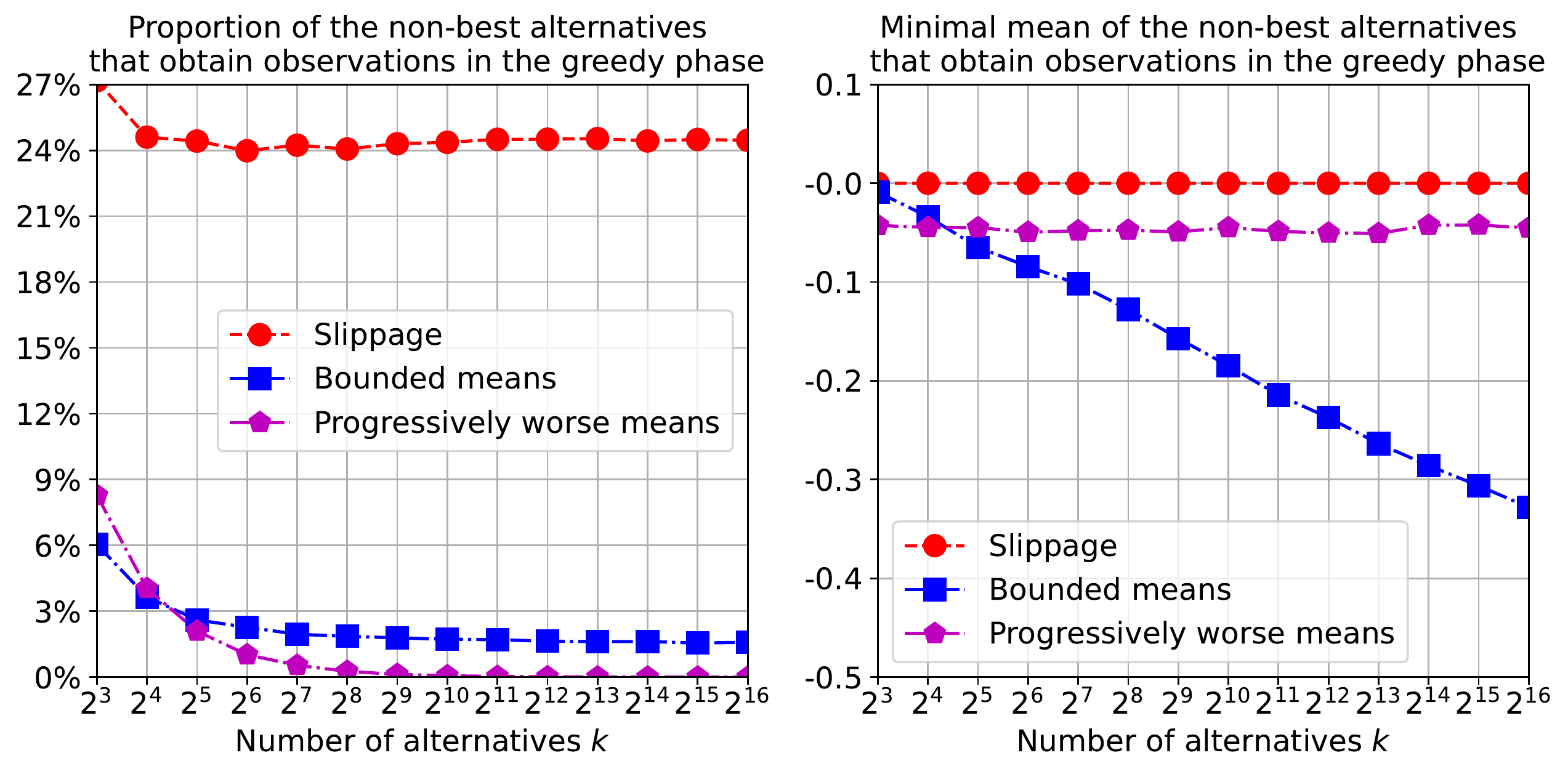}
  \caption{Characterizations of the budget allocation in the greedy phase for the EFG procedure: Part I.}
  \label{fig: budget_allocation1}
\end{figure}

We summarize the findings from Figure \ref{fig: budget_allocation1} as follows. First, in the greedy phase of the EFG procedure, only a small proportion of the non-best alternatives obtain observations. 
Even under the slippage configuration, the proportion is only about $24\%$. Furthermore, as the means of the alternatives become more dispersed, the aforementioned proportion becomes smaller. 
Second, the non-best alternatives that obtain observations in the greedy phase typically have a mean performance very close to the best mean (0.1 in this experiment). For instance, under the configuration with progressively worse means, only a few top alternatives whose means are within 0.2 to the best mean on average obtain observations in the greedy phase. 

The above results indicate that the purely exploitative budget allocation of the EFG procedure in the greedy phase may be very efficient. 
Intuitively speaking, for the non-best alternatives, even if the sample mean is upper-biased after the exploration phase, only the ones close to the best alternative may have a larger sample mean than the best alternative and need more observations to reduce the estimation bias. As demonstrated before, in the greedy phase, the EFG procedure will focus on those alternatives that are worthy of further sampling. \red{Notice that the intuition here may also explain the PCS result shown in Figure \ref{fig: budget_allocation_pcs}. As the means become more scattered, there will be fewer alternatives that are strong competitors to the best, making it easier for the procedure to reduce the bias causing a possible wrong selection and select the true best.}

\begin{figure}[htbp]
  \centering
  \includegraphics[width=0.75 \textwidth]{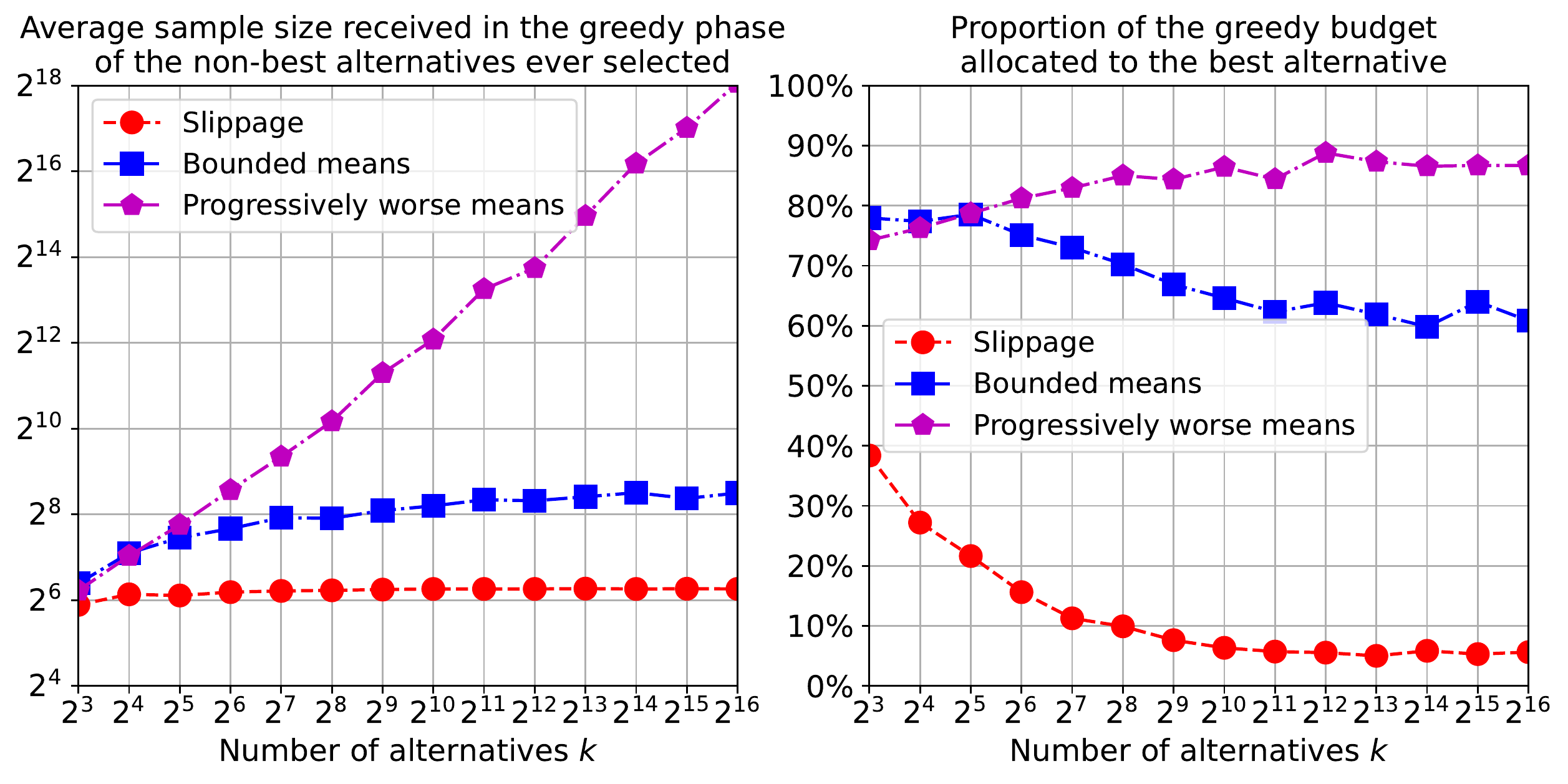}
  \caption{Characterizations of the budget allocation in the greedy phase for the EFG procedure: Part II.}
  \label{fig: budget_allocation3}
\end{figure}
We obtain the following findings from Figure \ref{fig: budget_allocation3}. First, for each non-best alternative, once it is selected at any stage in the greedy phase, it may be selected many more times in the subsequent stages. This may be because the non-best alternatives that obtain observations in the greedy phase typically bear a very competitive (e.g., highly upper-biased) sample mean after the exploration phase; then, it may require a considerable number of additional observations to debias the sample mean for each of them to finish the boundary-crossing process. Furthermore, we notice that for the configuration with progressively worse means, the average sample size received in the greedy phase of the non-best alternatives ever selected in the greedy phase grows linearly in $k$. This may be because even when $k$ is large, as discussed above, only a few top alternatives obtain observations in the greedy phase; then, they share all the greedy budget $n_g k$, which also grows linearly in $k$. Second, under all three configurations, the best alternative obtains a fixed proportion of the sampling budget in the greedy phase, and the proportion will converge as $k$ increases. Notice that under both non-slippage configurations, the proportion exceeds $50\%$ for all values of $k$, indicating that we may allocate too much sampling budget to the greedy phase. 

\begin{figure}[htbp]
  \centering
  \includegraphics[width=0.75 \textwidth]{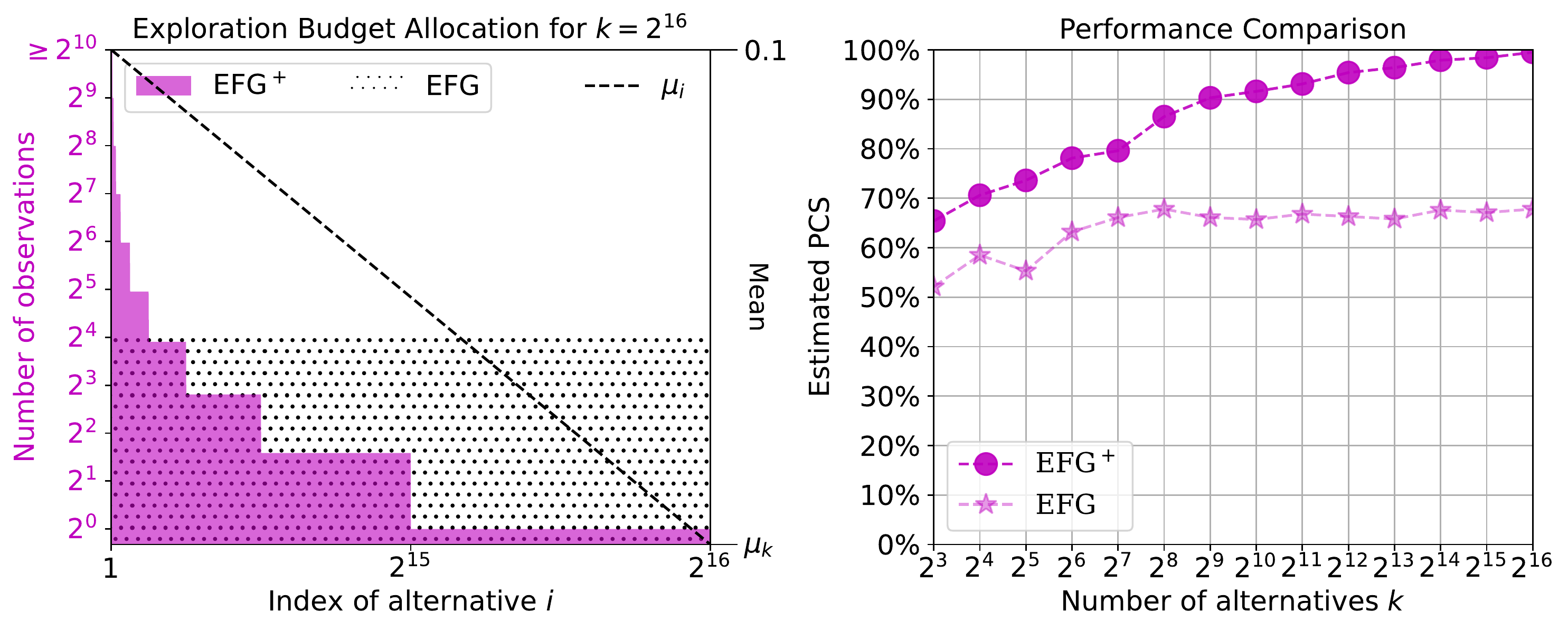}
  \caption{A comparison between the EFG procedure and the EFG$^+$ procedure under the configuration with progressively worse means.}
  \label{fig: budget_allocation_efgplus}
\end{figure}

\red{To end this subsection, we note that for the configuration with progressively worse means, the EFG procedure may be inefficient in the exploration phase. Under the configuration, a large portion of the alternatives may be clearly inferior; for each of them, only a few observations may be sufficient to reveal its inferiority. However, the EFG procedure allocates the exploration budget equally to all alternatives, which is extravagant since the majority of the total sampling budget is used for exploration. This inefficiency further motivates the use of a seeding phase in the EFG$^+$ procedure which is introduced in Section \ref{subsec: EFGPlus}. To illustrate the effect of the seeding phase, we compare the exploration budget allocation of the EFG procedure and the EFG$^+$ procedure (using an additional $10\%$ of the sampling budget $B$ in the seeding phase) for $k=2^{16}$ and $B=20k$ in Figure \ref{fig: budget_allocation_efgplus}, which also displays the procedures' PCS for different $k$. From Figure \ref{fig: budget_allocation_efgplus}, we can see that compared to the EFG procedure, the EFG$^+$ procedure allocates much less budget to the inferior alternatives (indexed with large $i$), validating the seeding phase’s usefulness for enhancing the exploration. Consequently, the EFG$^+$ procedure obtains a much higher PCS than the EFG procedure. }

\subsection{On Assumption \ref{assu:asyreg} and the PGS of the EFG Procedure}
\label{subsec: additional_EC}
Recall that our sample optimality results on the PCS rely on Assumption \ref{assu:asyreg} that requires the gap between the best and second-best means to remain above a positive constant no matter how large $k$ is. Besides, the results on the PGS rely on the assumption that the number of good alternatives remains bounded as $k$ grows to infinity.
We now consider problem configurations violating the structural assumptions to study the EFG procedure’s performance. In the experiments, unless otherwise specified, we estimate the PCS and PGS for the tested procedure based on 1000 independent macro replications and set the IZ parameter $\delta$ as 0.05.
Regarding notation, for different $k$, we use $\gamma(k)$ to denote the mean gap between the best and second-best alternatives and $g(k)$ to denote the total number of good alternatives.

\subsubsection{Problem Configurations with Randomly Generated Means.}
We first consider the following two problem configurations: (1) the configuration with means $i.i.d.$ generated from a Normal distribution and a common variance (Normal-CV) under which $\mu_i \sim \text{Normal}(0,1), \, \sigma_i^2=9, \, \forall i=1,\ldots, k$;
     (2) the configuration with means $i.i.d.$ generated from a Beta distribution and a common variance (Beta-CV) under which
    $\mu_i  \sim \text{Beta}(1.5,2),\, \sigma_i^2=9,  \, \forall i=1,\ldots, k$.
Notice that for the two configurations, $\gamma(k)$ is the difference between the two order statistics $\mu_{(k)}$ and $\mu_{(k-1)}$ of the randomly generated means. 
For the two configurations, we plot the $\gamma(k)$ and the $g(k)$ for different $\delta$  estimated based on one particular sample path against $k$  in Figure \ref{fig: Normal-CV} and Figure \ref{fig: Beta-CV}, respectively. Histograms of means for $k=2^{16}$ are also visualized there. 

\begin{figure}[htbp]
  \centering
  \includegraphics[width=\textwidth]{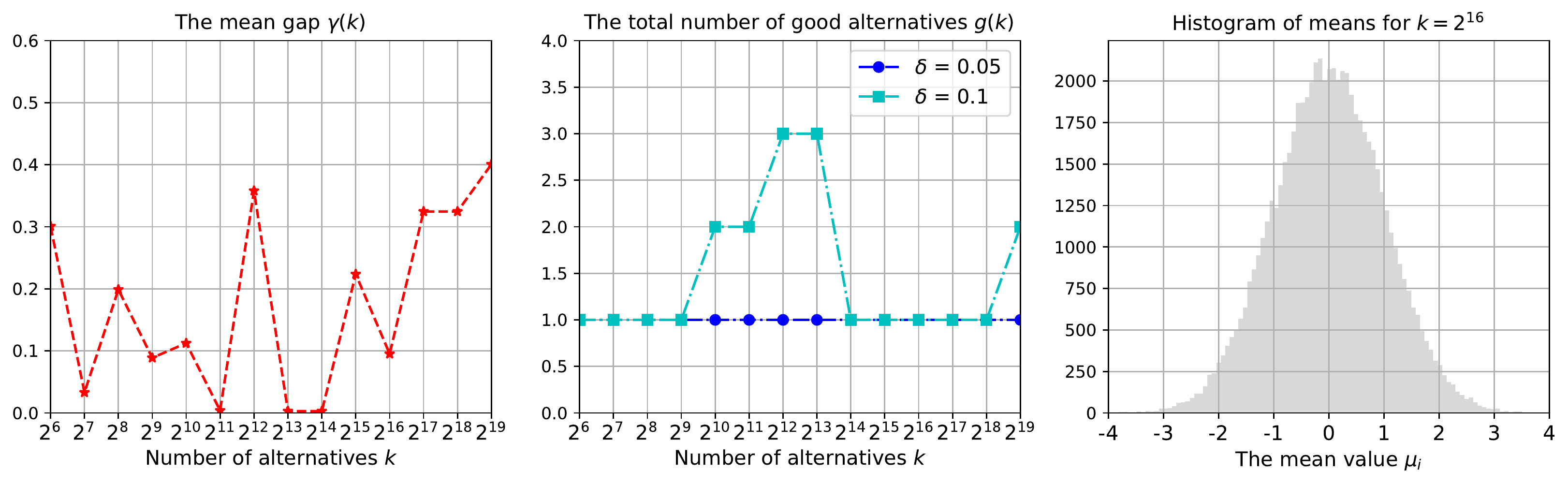}
  \caption{Characterizations of the configuration Normal-CV.}
  \label{fig: Normal-CV}
\end{figure}

\begin{figure}[htbp]
  \centering
  \includegraphics[width=\textwidth]{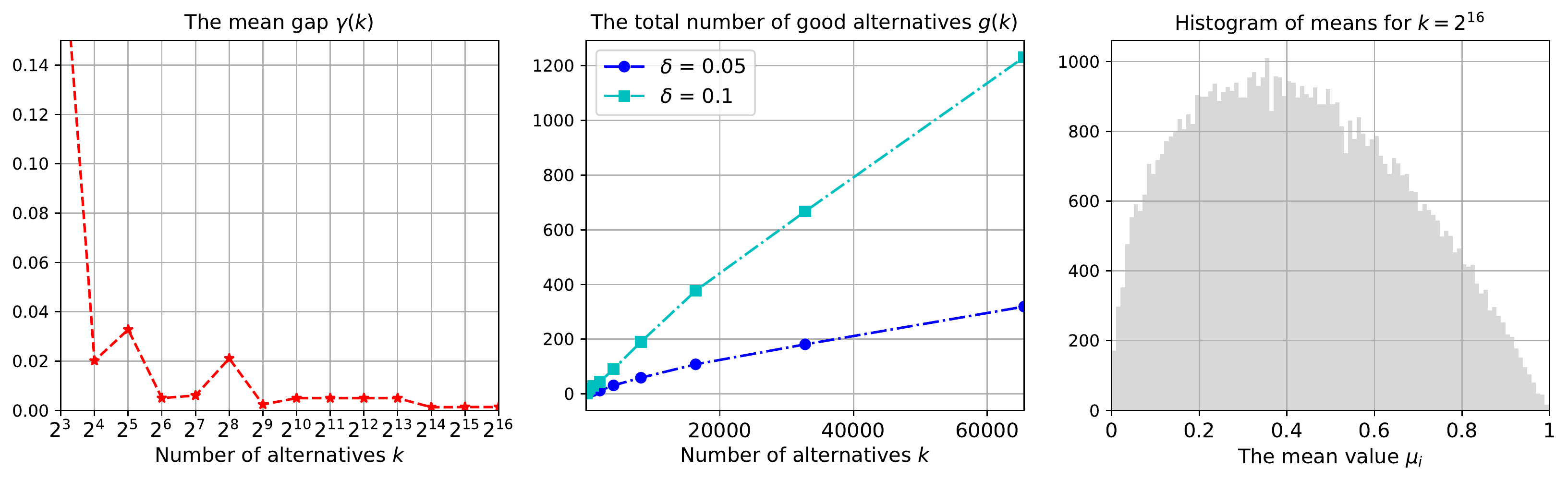}
  \caption{Characterizations of the configuration Beta-CV.}
  \label{fig: Beta-CV}
\end{figure}

From Figure \ref{fig: Normal-CV}, we can find that under Normal-CV, $\gamma(k)$ fluctuates up and down as $k$ increases, and it does not remain above some positive constant as required by Assumption \ref{assu:asyreg}.  However, it still exhibits a key structure of Assumption \ref{assu:asyreg}: it does not shrink to zero as $k$ increases. 
This may be because, for the normal distribution, the distribution of the difference of two order statistics may only depend on the order lag, which is known to hold for the exponential distribution \citep{itemBalakrishnan2014}. 
Additionally, under Normal-CV, $g(k)$ shows no apparent growth as $k$ increases.  Interestingly,  the configuration Beta-CV exhibits completely different features. From Figure \ref{fig: Beta-CV}, we can see that as $k$ increases, the mean gap $\gamma(k)$ shrinks to zero quickly, and $g(k)$ grows linearly in $k$.

We test the EFG procedure and compare it with the EA procedure for the two configurations. In the experiment,  we let $k=2^l$ with $l$ ranging from 3 to 16, and for each $k$, we let $B=100k$ and allocate $10k$ to the greedy phase for the EFG procedure. Then, for the two configurations,  we plot the estimated PCS and PGS of the two procedures against different $k$  in Figure \ref{fig: N_PCS_PGS} and Figure \ref{fig: B_PCS_PGS}, respectively.  
We have the following findings from Figure \ref{fig: N_PCS_PGS}. First, under Normal-CV, which violates Assumption \ref{assu:asyreg}, the EFG procedure may also be sample optimal for the PCS (PGS) as the estimated PCS (PGS) remains above $60\%$ ($70\%$) for all values of $k$. 
Second, as $k$ grows, the PCS of the EA procedure does not plummet to zero, as in Figure \ref{numerical: GreedyEFGRateOptimality}. This may be because, under Normal-CV, the number of alternatives generated that are close to the best does not increase as $k$ grows, which may also lead to the non-diminishing PGS of the EA  procedure. 
However, despite this, the EFG procedure still significantly outperforms the EA procedure for both the PCS and PGS, showing the substantial performance gains of adding a greedy phase.

\begin{figure}[htbp]
  \centering
  \includegraphics[width=0.75 \textwidth]{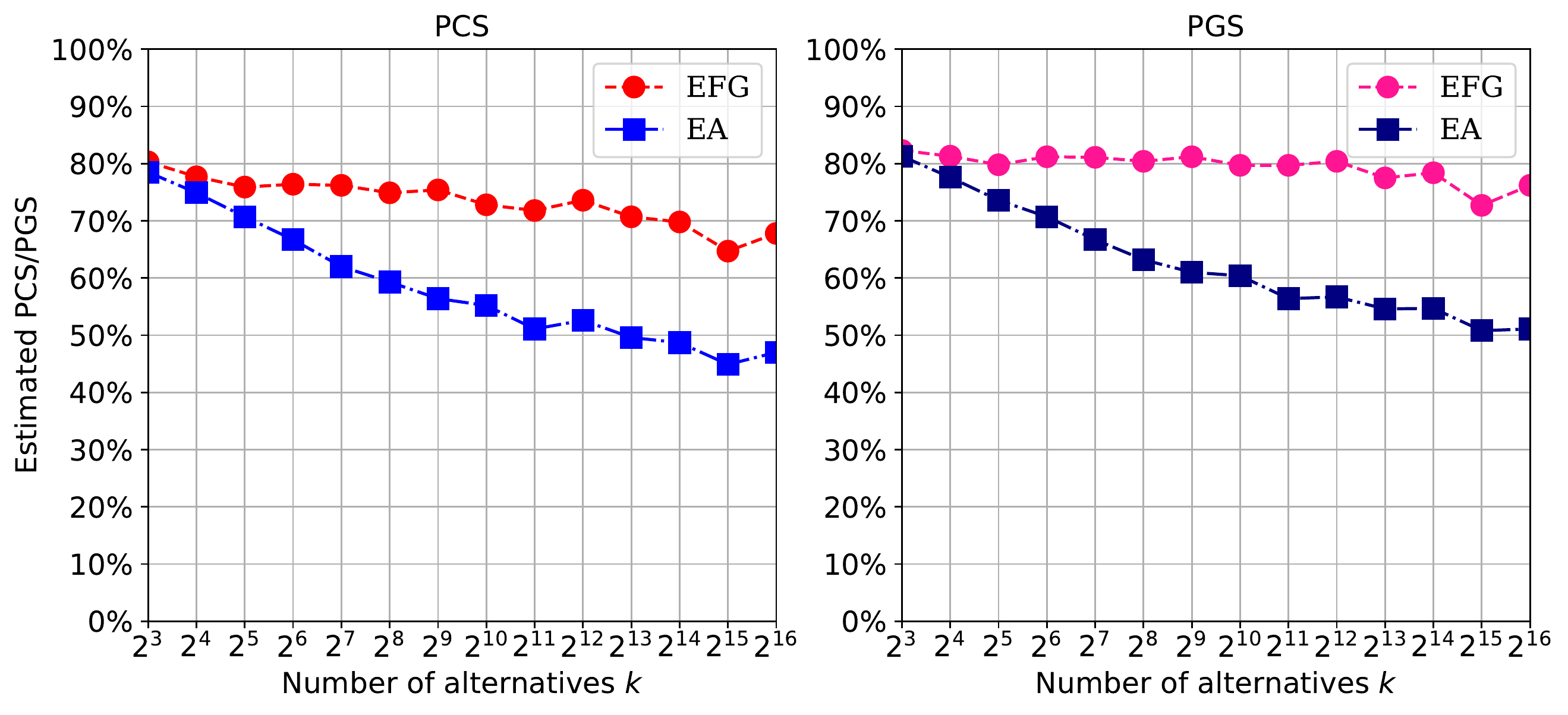}
  \caption{A comparison between the EA procedure and the EFG procedure under Normal-CV.}
  \label{fig: N_PCS_PGS}
\end{figure}
\begin{figure}[htbp]
  \centering
  \includegraphics[width=0.75 \textwidth]{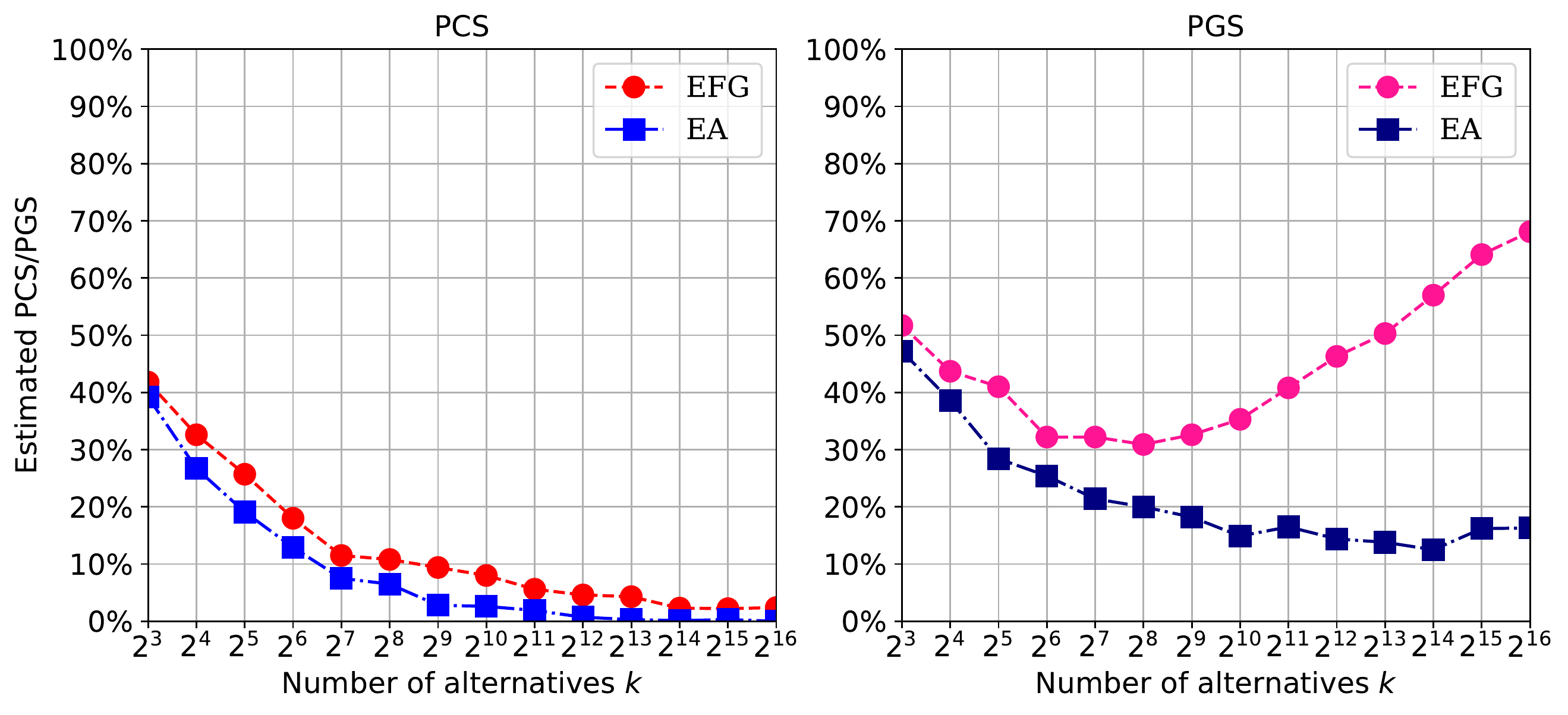}
  \caption{A comparison between the EA procedure and the EFG procedure under Beta-CV.}
  \label{fig: B_PCS_PGS}
\end{figure}

We obtain three findings from Figure \ref{fig: B_PCS_PGS}. First, under Beta-CV, both the PCS of the EFG and EA procedures diminish to zero as $k$ grows. While the EFG procedure's PCS remains higher than that of the EA procedure, this shows that when $\gamma(k) \rightarrow 0$, the EFG procedure may fail to maintain a non-zero PCS given a total sampling budget $B$ that is $O(k)$. Furthermore, when $\gamma(k) \rightarrow 0$, the attainable lower bound of the growth order of the total sampling budget required for a positive PCS may no longer be $O(k)$. In such scenarios, we conjecture that the EFG procedure may also achieve the sample optimality, i.e., attain the new minimal order of the total sampling budget. However, we will not go further on this issue because the objective of selecting a good alternative may be more meaningful when $\gamma(k) \rightarrow 0$.


Second, under Beta-CV, the PGS of the EA procedure does not decrease to zero, which is not surprising. When $g(k)$ grows at the order $O(k)$ as under Beta-CV, obtaining a non-zero PGS is trivial. One can even randomly select an alternative, and the resulting PGS $g(k)/k$ is positive no matter how large $k$ is. Therefore,  when $g(k)$ grows at the order of $O(k)$, a persistent non-zero PGS may not indicate the sample optimality.
Last, the EFG procedure's PGS exhibits an interesting monotonic increase after $k\geq 2^8$, and it becomes much larger than the EA procedure's PGS when $k$ is large. Again, this demonstrates the greedy phase's power in improving the sample efficiency. 

We note here that in the above experiment, $g(k)$ either remains bounded or grows linearly in  $k$. However, it may only grow at a sub-linear order such that $g(k) \rightarrow \infty$ but $g(k)/k \rightarrow 0$ as $k\rightarrow \infty$. The growth order of $g(k)$ may have a significant impact on the PGS. Therefore, we conduct additional experiments to compare the PGS of the EFG procedure under different growth orders of $g(k)$. See  \ref{subsec: explicit_rates} for the details and the results.

\subsubsection{Problem Configurations with Explicit Orders of  $g(k) \rightarrow \infty$.}
\label{subsec: explicit_rates}
 To characterize the impact of the growth order of $g(k)$ on the PGS of the EFG procedure, we consider the following configuration with a common unit variance and 
\begin{eqnarray}
\label{eq: configuration_infty}
    \mu_1=\delta,\, \mu_2 = \delta-\frac{\delta}{k},\, \mu_i \sim \mathcal{U}(0, \mu_2), \, \forall i=3, \ldots, g(k) ; \mu_j \sim \mathcal{U}(-1, 0), \, \forall j= g(k) +1, \ldots, k \quad \quad 
\end{eqnarray}
where $\mathcal{U}(a, b)$ is the uniform distribution with the support $[a,b]$  and $g(k)$ is a function of $k$.  In this configuration, $\gamma(k)$ converges to zero quickly as $k$ increases and our focus is only on the PGS.  
 In our experiment, we let $\delta = 0.05$, and consider the following two choices of $g(k)$: (1) $g(k) =  \left\lceil 0.05k \right\rceil $;  (2) $g(k) =  \left\lceil 0.5\sqrt{k} \right\rceil$. When $k$ is very small, the total number of good alternatives is set as $\max\{g(k), 2\}$. 
In this experiment, for each $k$, we let $B=100k$ and allocate $10k$ to the greedy phase for the EFG procedure.  See Figure \ref{fig: explicit_rates2} for the comparison between the PGS under $g(k) =   \left \lceil 0.05k \right\rceil $ and the PGS under $g(k) =   \left\lceil 0.5\sqrt{k} \right\rceil$ that are estimated based on 2000 independent macro replications. 



\begin{figure}[htbp]
  \centering
  \includegraphics[width=0.45 \textwidth]{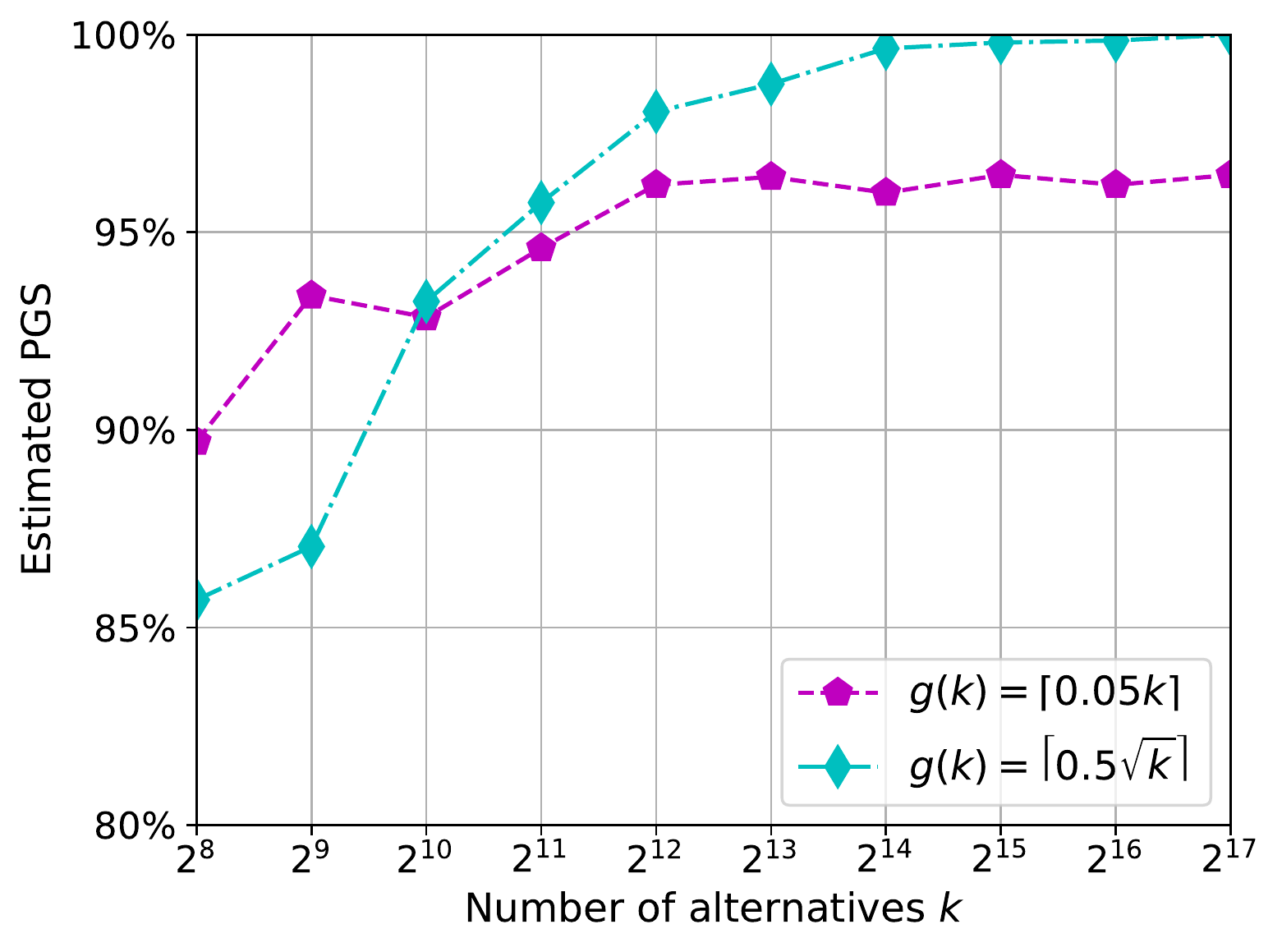}
  \caption{Estimated PGS of the EFG procedure under different $g(k)$.}
  \label{fig: explicit_rates2}
\end{figure}

From Figure \ref{fig: explicit_rates2}, we can see that
the PGS under $g(k) =   \left \lceil 0.05k \right\rceil $ becomes lower than that under $g(k) =   \left\lceil 0.5\sqrt{k} \right\rceil$ when $k \geq 2^{10}$. The former grows to 1 as $k$ increases, while the latter remains around $97\%$.
This may be because, when $g(k) =   \left\lceil 0.5\sqrt{k} \right\rceil$, the good alternatives only represent a negligible proportion of the alternatives  (notice that $\frac{g(k)}{k} \rightarrow 0$ as $k\rightarrow \infty$),  and they together may only use a small proportion of the total sampling budget when $k$ is large, e.g., $1k$. If so, the left budget is sufficient to let all inferior alternatives finish their boundary-crossing processes and ensure a correct selection as $k \rightarrow \infty$. However, when $g(k) =   \left \lceil 0.05k \right\rceil $, this is less likely to happen as the good alternatives may wastefully occupy too much sampling budget in the selection process.

{The phenomenon discussed above indicates that under the asymptotic regime of $k \rightarrow \infty$, if the number of good alternatives is allowed to be unbounded, the sample optimality regarding the PGS should be defined with more delicacy. For instance, when $g(k)$ grows at a sub-linear order as $k\rightarrow \infty$, for an R\&S procedure, being sample optimal may require its PGS to converge to 1 as $k\rightarrow \infty$ within a total sampling budget $B=ck$ for some positive constant $c$.  Then, characterizing the EFG procedure’s performance in terms of redefined sample optimality becomes an interesting and important problem.} 
However, as argued in Section \ref{subsec: GreedyOptPGS}, we need a better upper bound of the sampling budget allocated to the good alternatives in the boundary-crossing process. New technical treatment required for resolving this issue may lie beyond the scope of this paper, and we leave it as a future work. 

\section{\magenta{Parallelization of the EFG$^+$ Procedure}}
\label{sec: parallelization}

In this section, we first introduce the parallel version of the EFG$^+$ procedure, namely the EFG$^{++}$ procedure. We then test its performance in solving the large-scale throughput maximization problem and its parallel efficiency when implemented in a master-worker parallel computing environment. Lastly, we show how to reduce the communication overhead and achieve asynchronization for the EFG$^{++}$ procedure to improve the parallel efficiency.

\subsection{Description of the EFG$^{++}$ procedure}

A detailed description of the EFG$^{++}$ procedure is in Procedure \ref{procedure: EFGPlusPlus}. We design the procedure following a master-worker paradigm of parallel computing. In such a paradigm, the unique master processor (e.g., CPU) controls the main computing logic and manages a bunch of parallel worker processors to execute specific parallelizable tasks, e.g., simulating the alternatives. We implement the procedure in the Python language, and the program codes are available at \href{https://github.com/largescaleRS/greedy-procedures/tree/main/updates}{https://github.com/largescaleRS/greedy-procedures}.

\begin{procedure}[htbp]
  \caption{\textbf{Parallel Explore-First Greedy (EFG$^{++}$) Procedure}}
  \label{procedure: EFGPlusPlus}
  \begin{algorithmic}[1]
    \REQUIRE {$k$ alternatives $X_1,\ldots,X_k$, $q$ parallel worker processors, the total sampling budget $B=(n_{sd} + n_0 + n_g) k $, the total number of groups $G$, and the in-processor mini-batch size $z$.
    }
    \STATE \label{line: task_assign1} Find a task assignment scheme $\{n_{i, j}\}_{i=1, \ldots, k; j=1, \ldots, q}$ by Algorithm \ref{procedure: filling} satisfying
    \begin{equation}
      \label{eq: taskassign_seeding}
      \begin{aligned}
         & \sum_{i=1}^k n_{i, j} \leq\left\lceil n_{sd} k / q\right\rceil,
        \quad  \sum_{j=1}^q n_{i, j}=n_{sd},
        \quad n_{i, j} \in \mathbb{N}_{+}, \quad \forall\, i=1, \ldots, k,\, j=1, \ldots, q.
      \end{aligned}
    \end{equation}
    \FOR{ \textbf{processor}  $j \in \{1, 2, 3, \ldots, q\}$}
    \STATE For each alternative $i=1,\ldots, k$, if $n_{i,j} > 0$, take $n_{i,j}$ independent observations $X_{i1},\ldots,X_{in_{i,j}}$ and set $Y^{sd}_{i,j} =\sum_{l=1}^{n_{i, j}} X_{il}$; otherwise, let $Y^{sd}_{i,j}=0$. Return $\{Y^{sd}_{i,j}\}_{i=1,\ldots, k}.$
    \ENDFOR
    \STATE Block until all $\{Y^{sd}_{i,j}\}_{i=1,\ldots, k},  j=1, \ldots, q$ are received.  For each alternative $i=1,\ldots, k$, let
    $
    \bar X^{sd}_{i} = \frac{\sum_{j=1}^{q} Y^{sd}_{i,j}}{n_{sd}}
    $.
    \STATE According to $\bar X^{sd}_i$, sort the alternatives in descending order as $\{ (1), (2), \ldots, ( k) \}$.
    \STATE
    Let $I^r$ denote the group of alternatives for $r=1,\ldots, G$. Set $\Delta = 2^G-1$. Then, let $I^{1} = \{ (1), (2), \ldots, ( \lfloor  k / \Delta \rfloor ) \}$,  $I^{r} = \{ (\lfloor k 2^{r-2} / \Delta \rfloor+1), \ldots, ( \lfloor k 2^{r-1} / \Delta \rfloor) \}$ for $r=2, \ldots, G-1$, and let $I^{G} = \{ (\lfloor k 2^{G-2} / \Delta \rfloor+1), \ldots, (k) \}$.
    \STATE \label{line: task_assign2}  Find a task assignment scheme $\{n_{i, j}\}_{i=1, \ldots, k; j=1, \ldots, q}$ by Algorithm \ref{procedure: filling}  satisfying
    \begin{equation}
      \label{eq: taskassign_explore}
      \begin{aligned}
         & \sum_{i=1}^k n_{i, j} \leq\left\lceil n_0 k / q\right\rceil,
         \quad \sum_{j=1}^q n_{i, j}=\sum_{r=1}^G n^r \mathbb{I}_{i \in I^r}, \quad n_{i, j} \in \mathbb{N}_{+}, & \forall\, i=1, \ldots, k, j=1, \ldots, q.
      \end{aligned}
    \end{equation}
    \FOR{\textbf{processor}  $j \in \{1, 2, 3, \ldots, q\}$}
    \STATE For each alternative $i=1,\ldots, k$, if $n_{i,j} > 0$, take $n_{i,j}$ independent observations $X_{i1},\ldots,X_{in_{i,j}}$ and set $Y_{i,j} =\sum_{l=1}^{n_{i, j}} X_{il}$; otherwise, let $Y_{i,j}=0$. Return $\{Y_{i,j}\}_{ i=1,\ldots, k}$.
    \ENDFOR
    \STATE Block until all $\{Y_{i,j}\}_{i=1,\ldots, k}, j=1, \ldots, q$ are received. For each alternative $i=1,\ldots, k$,
    let $n_i=\sum_{r=1}^G n^r \mathbb{I}_{i \in I^r}$ and
    $
      \bar X_{i}(n_i) = \frac{\sum_{j=1}^{q} Y_{i,j}}{n_i}
    $.
    \WHILE{$ \sum_{i=1}^{k}n_i + qz \leq  (n_0+n_g) k$}
    \STATE Let $s = \arg\max_{i= 1, \ldots, k} \bar{X}_{i}\left(n_i\right)$;
    \FOR{\textbf{processor}  $j \in \{1, 2, 3, \ldots, q\}$}
    \STATE For alternative $s$, take $z$ independent observations $X_{s1},\ldots,X_{sz}$ and return $Y_j = \sum_{l=1}^{z}X_{sl}$;
    \ENDFOR
    \STATE Block until all $Y_j, j=1,\ldots, q$ are received.
    Let $\bar{X}_{s}\left(n_s\right) = \frac{n_s \bar{X}_{s}\left(n_s\right) + \sum_{j=1}^{q} Y_j}{n_s + qz}$ and $n_s=n_s + qz$.
    \ENDWHILE
    \STATE Select $\arg\max_{i= 1, \ldots, k} \bar{X}_{i}\left(n_i\right)$  as the best.
  \end{algorithmic}
\end{procedure}

As a parallel version of the EFG$^+$ procedure, the EFG$^{++}$ procedure also has three phases: the seeding phase, exploration phase, and greedy phase. 
In Procedure \ref{procedure: EFGPlusPlus}, the three phases are described in lines 1-5, lines 6-12, and lines 13-20, respectively. 
For the seeding and exploration phases, the key to efficient parallelization is load balancing of simulation tasks among the available worker processors, i.e., allocating different processors (almost) the same number of observations to simulate. 
Thus, in both phases, the master processor needs to find a load-balancing simulation task assignment scheme and then assign the tasks to the worker processors accordingly. 
To facilitate this, we also design a simple yet efficient sequential-filling algorithm for solving the task assignment schemes and introduce it in Section \ref{subsec: load_balancing}. 
In Procedure \ref{procedure: EFGPlusPlus}, the algorithm is used in line \ref{line: task_assign1} for the seeding phase and in line \ref{line: task_assign2} for the exploration phase. 
In each of the two phases, after simulation tasks are assigned,  the master processor waits for all the worker processors to finish and return the results. 
Then, the master processor updates the sample sizes and sample means for all alternatives and enters the next phase.

As introduced in Section \ref{subsec: parallel}, we consider using a batching approach to parallelize the fully sequential greedy phase for the EFG$^{++}$ procedure. In the batched greedy phase of the procedure, at each stage, the current-best alternative $s$ will be simulated more than once, and the simulation task will be equally assigned to the worker processors. Then, after all the worker processors finish simulating alternative $s$ and return the results, the master processor updates the sample size and sample mean for alternative $s$ and enters the next stage. When the total sampling budget is exhausted, the master processor delivers the alternative with the largest sample mean as the best. For the worker processors, we use $z$ to denote the number of observations to simulate at each stage and name it the \textit{in-processor mini-batch size}. With the presence of $q$ worker processors, the batch size of the greedy phase equals $ q\times z$. In our implementation, we regard $z$ as a parameter of the procedure.

In this section's experiments, we are concerned about not only one procedure's performance in terms of the PCS and PGS but also its parallel efficiency when implemented in parallel computing environments. Following  \cite{itemNi2017} and \cite{itemHong2022}, we quantify one procedure's  parallel efficiency in  a master-worker parallel computing environment using the measure \textit{utilization} defined by
\[\text{Utilization} = \frac{\text{Total simulation time}}{\text{Wall-clock time} \times \text{Number of parallel worker processors}},
\]
where the total simulation time is the total system time used by all the parallel worker processors for simulating all the observations. The higher the utilization is, the more efficient the procedure is in using the available processors, and the more suitable the procedure is for parallelization. Notice that we can estimate the utilization for each phase of the EFG$^{++}$ procedure.

\subsection{Performance and Parallel Efficiency of the EFG$^{++}$ procedure}
\label{subsec: test_EFGPlusPlus}

We test the performance and parallel efficiency of the EFG$^{++}$ procedure on the throughput maximization (TP) problem, which has been introduced in Section \ref{subsec: TP_EC}. In our experiments, we use a total sampling budget $B=100k$ and allocate 20\%, 70\%, and 10\% of the total sampling budget to the seeding phase, exploration phase, and greedy phase, respectively. Furthermore, unless otherwise specified, we estimate the concerned quantities like the PCS, PGS, wall-clock time, and the utilization of the procedure based on 300 independent macro replications. 

\subsubsection{Testing  the EFG$^{++}$  Procedure on the TP Problem.} 
\label{subsubsec: test_EFGPlusPlus_z1}
We first test the EFG$^{++}$procedure using the problem with $k=11774$ alternatives.   When running the procedure, we try four choices for the number of available worker processors $q \in \{5, 10, 20, 40\}$. For each $q$, we set $z=1$, i.e., let each worker processor simulate the current best only once at each stage of the greedy phase. Then, we report the PCS, PGS, wall-clock time, and utilization of the EFG$^{++}$procedure for different $q$ in Table \ref{tab: EFGPP11774}.

\begin{table}[htbp]
  \centering
  \begin{tabular}{cccccc}
    \hline
    \hline
    $\begin{array}{c}\text {\# of worker} \\
         \text {processors } q\end{array}$   &
    $\begin{array}{c}\text {In-processor} \\
         \text {mini-batch size } z\end{array}$ & PCS & PGS $(\delta=0.01)$ & $\begin{array}{c}\text {Wall-clock} \\
                                                                               \text {time (s)}\end{array}$ & Utilization                     \\
    \hline
    1                                                                                & 1                                                                                       & 0.42 & 0.94                & 107.572                                                                                                                                           & 83.85\%     \\
    5                                                                                & 1                                                                                       & 0.44 & 0.96                & 27.346                                                                                                                                            & 65.19\%     \\
    10                                                                               & 1                                                                                       & 0.46 & 0.96                & 19.059                                                                                                                                            & 46.13\%     \\
    20                                                                               & 1                                                                                       & 0.40 & 0.95                & 15.306                                                                                                                                            & 28.78\%     \\
    40                                                                               & 1                                                                                       & 0.47 & 0.92                & 13.132                                                                                                                                            & 16.68\%     \\ 
    \hline
    \hline
  \end{tabular}
  \caption{Performance and parallel efficiency of the EFG$^{++}$ procedure on the TP problem with $11774$ alternatives.}
  \label{tab: EFGPP11774}
\end{table}



From Table \ref{tab: EFGPP11774}, we have the following findings. First, there is no apparent decrease in the PCS and PGS as the number of worker processors grows from $1$ to $40$. This indicates that using a parallelizable batched greedy phase will not influence the performance of the EFG$^{++}$ procedure when the batch size is not large. Second, the utilization of the EFG$^{++}$ procedure is less than $90\%$ even when $q=1$. This is due to the fact that the time spent to find the current best at each stage of the greedy phase is not negligible in this medium-size problem. 
Third, when the number of worker processors is small (e.g., 5 or 10), the EFG$^{++}$ procedure can efficiently reduce the wall-clock time. In Table \ref{tab: EFGPP11774}, when the number of worker processors increases from 1 to 10, the wall-clock time is reduced from 107.572 seconds to 13.132 seconds. However, if the number of processors continues to increase, the procedure can not be further effectively accelerated. As a result, the utilization keeps decreasing to a deficient level as $q$ grows. To understand this phenomenon, we report the wall-clock times of different phases for the EFG$^{++}$ procedure in Table \ref{tab: EFGPP11774Times}.

\begin{table}[htbp]
  \centering
  \begin{tabular}{ccccc}
    \hline
    \hline
    \multirow{2}{*}{\begin{tabular}[c]{@{}c@{}}\# of worker \\  processors $q$\end{tabular}}  & \multirow{2}{*}{\begin{tabular}[c]{@{}c@{}}In-processor\\ mini-batch size $z$\end{tabular}}  & \multicolumn{3}{c}{Wall-clock time (s)}                                       \\ \cline{3-5} &                                  & \multicolumn{1}{c}{seeding phase} & \multicolumn{1}{c}{exploration phase} & greedy phase \\
     \hline
     1  & 1 & 16.799 & 56.840  & 33.933 \\
     5  & 1 & 3.442  & 11.812 & 12.092 \\
     10 & 1 & 1.762  & 5.957  & 11.341 \\
     20 & 1 & 0.894  & 2.997  & 11.415 \\
     40 & 1 & 0.492  & 1.551  & 11.089 \\
    \hline
    \hline
  \end{tabular}
  \caption{Each phase's wall-clock time of the EFG$^{++}$ procedure on the TP problem with $11774$ alternatives.}
  \label{tab: EFGPP11774Times}
\end{table}
From Table \ref{tab: EFGPP11774Times}, we can see that both the seeding and exploration phases are parallelized very efficiently. As the number of worker processors used $q$ is doubled for the two phases, the wall-clock time is approximately halved. 
The problem lies in the greedy phase. When $q\geq 10$, increasing $q$ shows almost no impact on the wall-clock time of the greedy phase.
This may arise due to the frequent communication between the master processors and the worker processors in the greedy phase. In this experiment, we set $z=1$ for different $q$. Then, in the greedy phase, to obtain one single simulation observation of the alternatives, the master processor has to communicate with one worker processor twice, once to specify which alternative to simulate and once to receive the observation from the worker processor, no matter how many worker processors there are. 

\subsubsection{Improving Utilization by Reducing Communication.} 
\label{subsubsec: test_EFGPlusPlus_commu}
A simple way to reduce the communication cost of running the batched greedy phase in parallel is to enlarge the value of $z$. By doing so, the batch size of each stage is increased, and the total number of rounds of communications is reduced. To show the impact of $z$, we also test the EFG$^{++}$ procedure with $z=10$ for $q=10$, $20$ and also $40$. The estimated PCS, PGS, wall-clock time and utilization for different $z$ and $q$ are summarized in Table \ref{tab: reducing_commu}.
We also keep track of the greedy phase's wall-clock times and report them in Table \ref{tab: reducing_commu}.

\begin{table}[htbp]
  \centering
  \begin{tabular}{ccccccc}
    \hline
    \hline
    $\begin{array}{c}\text {\# of worker} \\
         \text {processors } q\end{array}$    &
    $\begin{array}{c}\text {In-processor} \\
         \text {mini-batch size } z\end{array}$ & PCS                                            &
    $\begin{array}{c}\text {PGS} \\
         (\delta=0.01)\end{array}$          & $\begin{array}{c}\text {The greedy phase's} \\
                                              \text{wall-clock time (s)}\end{array}$ & $\begin{array}{c}\text {Wall-clock} \\
                                                                                        \text {time (s)}\end{array}$ & Utilization                                                   \\
                                                                                        \hline                                                                               10 & 1  & 0.46 & 0.96 & 12.884 & 19.059 & 46.13\% \\
                                                                                        10 & 10 & 0.44 & 0.96 & 1.286  & 8.673  & 94.77\% \\ \hline
                                                                                        20 & 1  & 0.40 & 0.95 & 11.415 & 15.306 & 28.78\% \\
                                                                                        20 & 10 & 0.48 & 0.94 & 1.315  & 5.423  & 78.72\% \\ \hline
                                                                                        40 & 1  & 0.47 & 0.92 & 11.089 & 13.132 & 16.68\% \\
                                                                                        40 & 10 & 0.46 & 0.95 & 1.131  & 3.174  & 66.65\% \\ \hline
    \hline
  \end{tabular}
  \caption{Examining the impact of in-processor mini-batch size $z$ on the performance and parallel efficiency of the EFG$^{++}$ procedure using the TP problem with 11774 alternatives.}
  \label{tab: reducing_commu}
\end{table}

We have the following findings from Table \ref{tab: reducing_commu}. First, we find that for each value of $q$, increasing $z$ does not reduce the PCS and PGS of the EFG$^{++}$ procedure. Recall that the stage batch size of the batched greedy phase $m$ equals $q\times z$. This finding shows that even when the batch size is large, e.g., $m=400$ in the last row of Table \ref{tab: reducing_commu}, using a batched greedy phase for parallelization can deliver the same level of selection accuracy. Second, increasing $z$ can effectively reduce the communication overhead between processors. For each value of $q$, when we increase $z$ from 1 to 10, the wall-clock time of the greedy phase is significantly reduced. As a result, we reduce the wall-clock time of the procedure and improve the parallel efficiency.
These results indicate that with a proper choice of $z$ (e.g., 10), we can parallelize the EFG$^{++}$ procedure efficiently without damaging its selection accuracy in solving large-scale R\&S problems.

\subsection{Asynchronization of the EFG$^{++}$ procedure}
\label{subsec: AsynEFGPlusPlus}

\subsubsection{Procedure Design and the Performance.} 
\label{subsubsec: test_AsynEFGPlusPlus}
We now introduce how to implement the EFG$^{++}$ procedure in an asynchronized fashion. Notice that for both the seeding phase and exploration phase, synchronization of sampling information only happens once at the end of the phase. Therefore, it suffices to consider asynchronizing the batched greedy phase only. Our design is straightforward. At the beginning of the greedy phase, the master processor still requires every worker processor to take $z$ observations from the first current-best alternative. Immediately after that, whenever the master processor receives the simulation results from any worker processor, it updates the sample size and sample mean for the simulated alternative and asks the idle processor to take $z$ observations from the new current-best alternative. After the total sampling budget is exhausted, the procedure delivers the alternative with the largest sample mean as the best. We name the new parallel procedure with an asynchronized greedy phase the Asyn-EFG$^{++}$ procedure. A formal description of the procedure is in Procedure \ref{procedure: Asyn_EFGPlusPlus}. 

\begin{procedure}[htbp]
  \caption{\textbf{Asynchronized Parallel Explore-First Greedy (Asyn-EFG$^{++}$) Procedure}}
  \label{procedure: Asyn_EFGPlusPlus}
  \begin{algorithmic}[1]
    \REQUIRE {$k$ alternatives $X_1,\ldots,X_k$, $q$ parallel worker processors, the total sampling budget $B=(n_{sd} + n_0 + n_g) k $, the total number of groups $G$, and the in-processor mini-batch size $z$.
    }
    \STATE Execute lines 1 - 12 of Procedure \ref{procedure: EFGPlusPlus}.
    \STATE Let $s = \arg\max_{i= 1, \ldots, k} \bar{X}_{i}\left(n_i\right)$;
    \FOR{\textbf{processor}  $j \in \{1, 2, 3, \ldots, q\}$}
    \STATE For alternative $s$, take $z$ independent observations $X_{s1},\ldots,X_{sz}$, return $s$, $j$ and $r=\sum_{l=1}^{z}X_{sl}$;
    \ENDFOR
    \WHILE{$ \sum_{i=1}^{k}n_i + qz \leq  (n_0+n_g) k$}
    \STATE  Block until receive from any processor the alternative id $s$, the sum $r$ and the processor id $j$;
    \STATE Let $\bar{X}_{s}\left(n_s\right) = \frac{n_s \bar{X}_{s}\left(n_s\right) + r}{n_s + z}$ and $n_s=n_s + z$;
    \STATE Let $s = \arg\max_{i= 1, \ldots, k} \bar{X}_{i}\left(n_i\right)$;
    \STATE {Let} \textbf{processor} $j$ \textbf{do}
    \STATE \quad For alternative $s$, take $z$ independent observations $X_{s1},\ldots,X_{sz}$, return $s$, $j$ and $r=\sum_{l=1}^{z}X_{sl}$;
    \ENDWHILE
    \FOR{$t \in \{1, 2, 3, \ldots, q\}$}
    \STATE  Block until receive from any processor the alternative id $s$, the sum $r$ and the processor id $j$;
    \STATE Let $\bar{X}_{s}\left(n_s\right) = \frac{n_s \bar{X}_{s}\left(n_s\right) + r}{n_s + z}$ and $n_s=n_s + z$;
    \ENDFOR
    \STATE Select $\arg\max_{i= 1, \ldots, k} \bar{X}_{i}\left(n_i\right)$  as the best.
  \end{algorithmic}
\end{procedure}

\begin{table}[htbp]
  \centering
  \begin{tabular}{cccccc}
    \hline
    \hline
    $\begin{array}{c}\text {\# of worker} \\
         \text {processors } q\end{array}$   &
    $\begin{array}{c}\text {In-processor} \\
         \text {mini-batch size } z\end{array}$ & PCS & PGS $(\delta=0.01)$ & $\begin{array}{c}\text {Wall-clock} \\
                                                                               \text {time (s)}\end{array}$ & Utilization                     \\
    \hline
    5  & 1  & 0.46 & 0.95 & 29.900 & 59.05\% \\
    5  & 10 & 0.44 & 0.92 & 17.497 & 95.07\% \\ \hline
    10 & 1  & 0.43 & 0.95 & 20.297 & 41.59\% \\
    10 & 10 & 0.44 & 0.96 & 8.673  & 94.77\% \\ \hline
    20 & 1  & 0.44 & 0.94 & 16.824 & 25.17\% \\
    20 & 10 & 0.45 & 0.95 & 5.432  & 76.46\% \\ \hline
    40 & 1  & 0.45 & 0.98 & 15.604 & 13.66\% \\
    40 & 10 & 0.45 & 0.97 & 3.340  & 61.62\% \\ \hline
    \hline
  \end{tabular} \caption{Performance and parallel efficiency of the Asyn-EFG$^{++}$ procedure on the TP problem with $11774$ alternatives.}
  \label{tab: Asyn11774}
\end{table}

Now we test the selection accuracy of the Asyn-EFG$^{++}$ procedure. Here, we use the same problem and implementation settings used in Section \ref{subsec: test_EFGPlusPlus}. The estimated PCS, PGS and wall-clock times for different $q$ and $z$ are summarized in Table \ref{tab: Asyn11774}.  From Table  \ref{tab: Asyn11774}, we can see that all the results regarding the performance and parallel efficiency of the EFG$^{++}$ procedure and the impact of $z$ obtained in Section \ref{subsec: test_EFGPlusPlus} also hold for the Asyn-EFG$^{++}$ procedure.

\subsubsection{Random Simulation Times.}
\label{subsubsec: test_random_times}
 In the above experiments, it only takes 0.07ms on average to take one observation from the alternatives\footnote{We achieve this speed by utilizing the Python library Cython to compile the simulation program as faster C codes. }. When simulating the alternatives is so fast, synchronization in the greedy phase may be fine. To show the value of asynchronization, we conduct the following experiment. For each worker processor, every time it generates one simulation observation, it sleeps for $T$ ms, which is $i.i.d.$ generated from the uniform distribution with the support $[0.5, 1.5]$. We use the problem instance with $k=41624$ alternatives and $40$ worker processors in the experiment. For both the EFG$^{++}$ procedure and the Asyn-EFG$^{++}$ procedure, we estimate the utilization for each of the three phases and summarize them in Table \ref{table: utils_random_times}. 

\begin{table}[htbp]
  \centering
  \begin{tabular}{cccccc}
    \hline
    \hline
    \multirow{2}{*}{\begin{tabular}[c]{@{}c@{}}In-processor\\ mini-batch size $z$\end{tabular}} &  \multirow{2}{*}{Procedure} & \multicolumn{3}{c}{Utilization}                                       \\ \cline{3-5} &                                  & \multicolumn{1}{c}{seeding phase} & \multicolumn{1}{c}{exploration phase} &  \multicolumn{1}{c}{greedy phase} \\ 
    \hline
\multirow{2}{*}{1}  & EFG$^{++}$     & 99.20\% & 99.60\% & 46.71\% \\
                    & Asyn-EFG$^{++}$ & 99.18\% & 99.57\% & 40.48\% \\ \hline
\multirow{2}{*}{10} & EFG$^{++}$     & 99.18\% & 99.57\% & 80.45\% \\
                    & Asyn-EFG$^{++}$ & 99.17\% & 99.56\% & 99.48\% \\ 
\hline
\hline
  \end{tabular}
  \caption{Each phase's parallel efficiency of the EFG$^{++}$ procedure and the Asyn-EFG$^{++}$ procedure on the TP problem with 41624 alternatives.}
  \label{table: utils_random_times}
\end{table}


We obtain the following findings from Table \ref{table: utils_random_times}. 
First, the seeding and exploration phases can be parallelized very efficiently for both procedures. The utilization in the two phases are all above $99\%$. Second, for the Asyn-EFG$^{++}$ procedure, increasing $z$ can also help increase the utilization in the greedy phase due to the reduction of communication overhead. Third, the impact of asynchronization on the parallel efficiency of the EFG$^{++}$ procedure in the greedy phase varies for different values of $z$. When $z=1$, the Asyn-EFG$^{++}$ procedure’s utilization in the greedy phase can be lower than that of the EFG$^{++}$ procedure. This may be because, in the greedy phase, the Asyn-EFG$^{++}$ procedure has to update the current best alternative every time a worker processor returns the simulation results, thereby incurring a larger overhead of finding the \emph{argmin} than the EFG$^{++}$ procedure. However, when $z$ is increased to $10$,  asynchronizing the greedy phase improves the utilization significantly, showing that the profit of asynchronization surpasses the excess overhead of frequent updating the current best alternative. From the results, we conclude that with a proper choice of $z$, asynchronization can improve the EFG$^{++}$ procedure's robustness to random and unequal simulation times.

\subsection{Load-balancing Simulation Task Assignment}
\label{subsec: load_balancing}
The sequential-filling algorithm for load-balancing simulation task assignment among the available worker processors is shown in Algorithm \ref{procedure: filling}. Given $\sum_{i=1}^{k}n_i$ observations to simulate   in total, each of the $q$ processors should simulate no more than $\lceil \sum_{i=1}^{k}n_i/q \rceil$ observations. Informally, we regard $\lceil \sum_{i=1}^{k}n_i/q \rceil$ as the maximum volume of each processor. In the procedure, we use $l_i, i=1, \dots, k$ to denote the unassigned number of observations for alternative $i$, and use $v_j, j=1, \dots, q$ to denote the unfulfilled ``volume'' of processor $j$. The procedure proceeds by filling the simulation task into different processors sequentially and uses $n_{i, j}$ to keep track of how many observations of alternative $i$ are assigned to processor $j$. 

This algorithm is used in the seeding and exploration phases of the EFG$^{++}$ procedure and the Asyn-EFG$^{++}$ procedure. It works very efficiently in our experiments. As displayed in Table \ref{table: utils_random_times}, the utilization of the two procedures in the seeding and exploration phases are always above $99\%$.

\renewcommand{\thealgorithm}{1}
\begin{algorithm}[htbp]
  \caption{\textbf{Sequential-Filling Algorithm}}
  \label{procedure: filling}
  \begin{algorithmic}[1]
    \REQUIRE {The number of alternatives $k$, the number of processors $q$, and the sample size of each alternative $n_i,\, i=1,\ldots, k$.
    }
    \STATE Let $l_i=n_i$, $v_j=\lceil \sum_{i=1}^{k}n_i/q \rceil$, and $n_{i,j}=0$ for $i=1, \ldots, k$ and $j=1, \ldots, q$.
    Set counter $j=0$.
    \FOR{$i \in \{1, 2, 3, \ldots,q\}$}
    \WHILE{\textbf{True}}
    \IF{$l_i \leq v_j$}
    \STATE Let $n_{i,j} = n_{i,j} + l_j$ and $v_j=v_j - l_i$; 
    \STATE \textbf{break}
    \ELSE
    \STATE Let $n_{i,j} = n_{i,j} + v_j$, $l_i = l_i-v_j$, and update the counter $j=j+1$.
    \ENDIF
    \ENDWHILE
    \ENDFOR
    \STATE Return $\{n_{i,j}\}_{i=1,\ldots, k;  j=1, \ldots, q}$.
  \end{algorithmic}
\end{algorithm}

\end{document}